\documentclass[11pt]{article}

\usepackage{breakurl}

\usepackage{url}            
\usepackage{wrapfig}

\usepackage{amsmath,amssymb,fullpage,graphicx,enumitem}
\usepackage{comment}
\usepackage{color}
\usepackage{setspace}
\usepackage{pdflscape}
\usepackage{enumitem}

\usepackage{amssymb,amsmath,amsthm}
\usepackage[english]{babel}

\usepackage{breakurl} 

\RequirePackage[colorlinks,citecolor=blue,urlcolor=blue]{hyperref}
\usepackage[latin1]{inputenc}
\usepackage{pifont}
\usepackage{enumitem}
\usepackage{amsmath,amssymb,amsfonts,amsthm,bbm}
\usepackage[ruled,vlined,linesnumbered,noresetcount]{algorithm2e}
\usepackage{array}
\usepackage{color}
\usepackage{comment}

\usepackage{floatpag}
\usepackage{afterpage}
\usepackage{float}
\usepackage{subfig}
\usepackage{epsfig}

\usepackage[normalem]{ulem}

\RequirePackage{natbib}
\usepackage{wrapfig}
\usepackage{parskip}
\usepackage{tcolorbox}
\usepackage{soul}
\usepackage{comment}
\usepackage{rotating}
\usepackage{bm}
\usepackage{array}

\usepackage{soul}

\newcommand{\norm}[1]{\left\lVert#1\right\rVert}

\renewcommand{\theequation}{\thesection.\arabic{equation}}
\numberwithin{equation}{section}

\renewcommand{\hat}{\widehat}

\renewcommand{\hat}{\widehat}

\newcommand{\bfm}[1]{\ensuremath{\mathbf{#1}}}

   \def\bI{\bfm I}

   \def\bM{\bfm M}

\def\bq{\bfm q}

\def\bt{\bfm t}     
   \def\bU{\bfm U}  
   \def\bV{\bfm V}  
     
\def\bx{\bfm x}     
\def\by{\bfm y}   \def\bY{\bfm Y}  
\def\bz{\bfm z}

                  \def\hbfeta {\hat{\bfsym {\eta}}}

 \def\htheta{\hat {\theta}}             \def\hbtheta {\hat{\bfsym {\theta}}} 
               
                                \def\bphi{\bfsym{\phi}}

 \def\hpi{\hat{\pi}}

\def\bbR {\mathbb R}
\def\bbP {\mathbb P}
\def\calC{\mathcal C}
\def\calF{\mathcal F}
\def\calG{\mathcal G}

\def\calJ{\mathcal J}

\def\calL{\mathcal L}

\def\calM{\mathcal M}

\def\calP{\mathcal P}
\def\calQ{\mathcal Q}
\def\calS{{\mathcal S}}

\def\calY{\mathcal Y}
\def\calV{\mathcal V}
\def\calI{\mathcal I}

\def\wtG    {\widetilde G}
\def\hp {\hat p}

\def\tbtheta {\tilde{\btheta}}
\def\tbfeta  {\tilde{\bfeta}}
\def\talpha  {\tilde{\alpha}}

\def\aG     {\check{G}}
\def\atheta {\check{\theta}}
\def\abtheta{\check{\btheta}}
\def\abfeta {\check{\bfeta}}
\def\aeta   {\check{\eta}}

\def\aalpha {\check{\alpha}}
\def\aomega {\check{\omega}}

\newcommand{\name}{GSF }
\def\package{\texttt{GroupSortFuse}}

\def\bbE{\mathbb E}

\def\hG      {\hat{G}}

\def\KL{\text{KL}}

\newcommand{\bfsym}[1]{\ensuremath{\boldsymbol{#1}}}

 \def\bbeta{\bfsym \beta}
 \def\bgamma{\bfsym \gamma}             
            
 \def\bfeta{\bfsym {\eta}}              
 \def\bmu{\bfsym {\mu}}                 
                  \def\bPsi {\bfsym {\Psi}}
 \def\btheta{\bfsym {\theta}}           \def\bTheta {\bfsym {\Theta}}
           
              \def\bSigma{\bfsym \Sigma}

 \def\brho   {\bfsym {\rho}}
 
 \def\bxi{\bfsym {\xi}}
 
 \def\bpi{\bfsym {\pi}}
 
 \def\bvarsigma{\bfsym{\varsigma}}

\DeclareMathOperator*{\argmax}{argmax}
\DeclareMathOperator*{\argmin}{argmin}

\DeclareMathOperator*{\esssup}{ess \, sup}

\DeclareMathOperator{\tr}{tr}

\newtheorem{lemma}{Lemma}
\newtheorem{theorem}{Theorem}
\newtheorem{proposition}{Proposition}
\newtheorem{corollary}{Corollary}

\newcounter{CondCounter}
  
  \theoremstyle{theorem}
\newtheorem{definition}{Definition}

\newtheorem{innercustomthm}{Theorem}
\newenvironment{customthm}[1]
  {\renewcommand\theinnercustomthm{#1}\innercustomthm}
  {\endinnercustomthm}

\newtheorem{innercustomlem}{Lemma}
\newenvironment{customlem}[1]
  {\renewcommand\theinnercustomlem{#1}\innercustomlem}
  {\endinnercustomlem}

\newtheorem{innercustomprop}{Proposition}
\newenvironment{customprop}[1]
  {\renewcommand\theinnercustomprop{#1}\innercustomprop}
  {\endinnercustomprop}

\SetKwFor{kwEstep}{E-Step}{:}{}%
\SetKwFor{kwMstep}{M-Step}{:}{}%

\SetKwInput{init}{Input}

\usepackage{graphicx}
\graphicspath{{images/}}

\setlist{  
  listparindent=\parindent,
  parsep=3pt,
}

\makeatletter
\def\thm@space@setup{%
  \thm@preskip=\parskip \thm@postskip=0pt
}
\makeatother

\definecolor{grey}{HTML}{C0C0C0}
\definecolor{brown}{HTML}{595035}
\definecolor{navyblue}{HTML}{0B0084} 
\definecolor{violet}{rgb}{0, .42, 0} 
\definecolor{dgreen}{rgb}{0, .42, 0}

\date{\vspace{-5ex}}

\begin{document}

\begin{center} {\LARGE{\bf{Estimating the Number of Components in Finite \\[0.06in] Mixture Models
via the Group-Sort-Fuse Procedure}}}
 
\vspace*{.3in}

{\large 


{ 
\begin{tabular}{cccc}
Tudor Manole$^1$, Abbas Khalili$^2$\\
\end{tabular}
 
{
\vspace*{.1in}
\begin{tabular}{c}
				$^1$Department of Statistics and Data Science, Carnegie Mellon University\\
				$^2$Department of Mathematics and Statistics, McGill University  \\
 				\texttt{tmanole@andrew.cmu.edu, abbas.khalili@mcgill.ca}
\end{tabular} 
}

\vspace*{.1in}

}}

\today

\end{center}
\vspace*{.2in}

 \begin{abstract}
Estimation of the number of components (or order) of 
a finite mixture model is a long standing and challenging problem in statistics.
We propose the Group-Sort-Fuse (GSF) procedure---a new penalized likelihood approach
for simultaneous estimation of 
the order and mixing
measure in multidimensional finite mixture models. 
Unlike methods which fit and compare mixtures with varying 
orders 
using criteria involving model complexity, our approach directly penalizes a 
continuous function of the model parameters. More specifically,
given a conservative upper bound on the order, 
the \name groups and sorts mixture component parameters 
to fuse those which are redundant. 
For a wide range of finite mixture models, we show that
the \name is consistent in estimating the true 
mixture order
and achieves the $n^{-1/2}$ convergence 
rate for parameter estimation up to polylogarithmic factors. 
The GSF is implemented for several univariate and multivariate
mixture models in the R package \package. 
Its finite sample performance is supported by a
thorough simulation study, and its application
is illustrated on two real data examples.
\end{abstract} 
 
\section{Introduction}

Mixture models are a flexible tool for modelling data from a population consisting of 
multiple hidden homogeneous subpopulations. Applications in economics \citep{BOSCH2010}, 
machine learning \citep{goodfellow2016}, 
genetics \citep{BECHTEL1993} and other life sciences \citep{THOMPSON1998, MORRIS1996} 
frequently employ mixture distributions. A comprehensive review of statistical inference and 
applications of finite mixture models can be found in the book by \cite{MCLACHLAN2000}.

Given integers $N, d \geq 1$, let ${\cal F} = \{ f(\by; \btheta) : \btheta= (\theta_1, \ldots, \theta_d)^{\top} 
\in \Theta \subseteq \bbR^d, \ \by \in \mathcal{Y} \subseteq \bbR^N \}$ 
be a parametric family of density  
functions with respect to a $\sigma$-finite measure $\nu$, 
with a compact parameter space $\Theta$.
The density function of a finite mixture model with respect to $\calF$ is given by
\begin{equation}
\label{mixtureModel}
p_G(\by) = \int_{\Theta} f(\by; \btheta) \text{d} G(\btheta) = 
\sum_{j=1}^K \pi_j f(\by; \btheta_j),
\end{equation}
where
\begin{equation}
G = \sum_{j = 1}^K \pi_j \delta_{\btheta_j}
\end{equation}
 is the mixing measure with $\btheta_j = (\theta_{j1}, \dots, \theta_{jd})^{\top}\in \Theta$, 
 $j=1, \dots, K$,
and the mixing probabilities $0 \leq \pi_j \leq 1$ 
satisfy $\sum_{j=1}^K \pi_j = 1$. 
Here, $\delta_{\btheta}$ denotes a Dirac measure placing mass at $\btheta \in \Theta$. 
The $\btheta_j$ are said to be atoms of $G$, and $K$ 
is called the \textit{order} of the model. 

 Let $\bY_1, \dots, \bY_n$ be a random sample from a finite mixture model \eqref{mixtureModel}
 with true mixing measure $G_0 = \sum_{j=1}^{K_0} \pi_{0j} \delta_{\btheta_{0j}}$.  
 The true order $K_0$ is defined as the smallest number 
 of atoms of $G_0$ for which the component densities $f(\cdot; \btheta_{0j})$ are different, and 
 the mixing proportions $\pi_{0j}$ are non-zero.
 This paper is concerned with parametric 
 estimation of 
 $K_0$.
 
In practice, the order of a finite mixture model may not be known. An assessment of the order is 
important even if it is not the main object of study. Indeed, a mixture model whose order is less 
than the true number of underlying subpopulations provides a poor fit, while a model with too 
large of an order, which is said to be overfitted, 
may be overly complex and hence uninformative. From a theoretical standpoint, estimation 
of overfitted finite mixture models 
leads to a deterioration in rates of convergence of standard parametric estimators. Indeed, 
given a consistent estimator 
$G_n$ of $G_0$ with $K > K_0$ atoms, the parametric $n^{- 1/ 2}$ convergence rate is 
generally not achievable. 
Under the so-called second-order strong identifiability condition, 
\cite{CHEN1995} and \cite{HO2016strong} showed that 
the optimal pointwise rate of 
convergence in estimating $G_0$ is bounded below by $n^{-1 /4}$ with respect to an appropriate 
Wasserstein metric. In particular, this rate is achieved by the 
maximum likelihood estimator up to a polylogarithmic factor. 
Minimax rates of convergence have also been established 
by \cite{HEINRICH2018}, under stronger regularity conditions on the parametric family $\calF$. 
Remarkably, these rates deteriorate as the upper bound $K$ increases. 
This behaviour has also been noticed for pointwise estimation 
rates 
in mixtures which do not satisfy the second-order strong identifiability assumption---see 
for instance \cite{CHEN2003} and \cite{HO2016weak}. 
These results warn against fitting finite mixture models with an incorrectly specified order. 
In addition to poor convergence rates, 
the consistency of $G_n$ does not guarantee the consistent estimation 
of the mixing probabilities 
and atoms of the true mixing measure, 
though they are of greater interest in most applications.
 
The aforementioned challenges have resulted in the development of many
methods for estimating the order of a finite mixture model. It is difficult to provide 
a comprehensive list of the research on this problem, and thus we give a selective overview. 
One class of methods involves hypothesis testing on the order using likelihood-based 
procedures \citep{MCLACHLAN1987, DACUNHA1999, LIUSHAO2003}, and 
 the EM-test \citep{CHEN2009, LI2010}. 
These tests typically assume knowledge of a candidate order; 
when such a candidate is unavailable, estimation methods can be employed. 
Minimum distance-based methods for estimating $K_0$ have been considered by \cite{CHEN1996}, 
\cite{JAMES2001}, \cite{WOO2006}, \cite{HEINRICH2018}, and \cite{HO2017robust}.
The most common parametric methods 
involve the use of an information criterion, whereby a penalized likelihood function is evaluated for a sequence of candidate models. Examples  
include Akaike's Information Criterion (AIC; \citet{AKAIKE1974}) and the Bayesian Information Criterion (BIC; \citet{SCHWARZ1978}).  
The latter is arguably the most frequently used method 
for mixture order estimation \citep{LEROUX1992, KERIBIN2000, MCLACHLAN2000},  
though it was not originally developed for non-regular models.
This led to the development of information criteria such as
the Integrated Completed Likelihood (ICL; \citet{BIERNACKI2000}), 
and the Singular BIC (sBIC; \cite{DRTON2017}).   
Bayesian approaches include the method of Mixtures of Finite Mixtures, 
whereby a prior is placed on the number of components 
\citep{nobile1994,RichardsonGreen1997,STEPHENS2000,miller2018},
and model selection procedures based on Dirichlet Process mixtures, 
such as those of \cite{ISHWARAN2001}
and the Merge-Truncate-Merge method of 
\cite{guha2019}.
Motivated by regularization techniques in regression, \cite{CHENKH2008} proposed a 
penalized likelihood method for order estimation 
in finite mixture models with
a one-dimensional parameter space $\Theta$,
where the regularization is applied to the difference between sorted 
atoms of the overfitted mixture model. 
\cite{HUNG2013} adapted this method to estimation of  
the number of states in Gaussian Hidden Markov models, which was also limited to 
one-dimensional parameters for different states. 
Despite its model selection consistency and good finite sample
performance, the extension of this method to multidimensional mixtures has not been addressed.
In this paper, we take on this task and propose a far-reaching generalization 
called the Group-Sort-Fuse (GSF) procedure.

The \name postulates an overfitted
finite mixture model with a large tentative order $K > K_0$.
The true order $K_0$ and the mixing measure $G_0$
are simultaneously estimated by merging redundant mixture components, 
by applying two penalties to the log-likelihood function of the model. 
The first of these penalties groups the
estimated atoms, while the second penalty  
shrinks the distances between those which are in high proximity. 
The latter is achieved by applying a sparsity-inducing regularization function
to consecutive distances between these atoms, sorted using 
a so-called \textit{cluster ordering} (Definition \ref{cluster-order}).
Unlike most existing methods, this form of regularization, 
which uses continuous functions of the model parameters as 
penalties, circumvents the 
fitting of mixture models of all orders $1, 2, \dots, K$. 
In our simulations we noticed that using
EM-type algorithms \citep{DEMPSTER1977},
the \name is less sensitive to the choice of starting values than methods which 
involve maximizing likelihoods of mixture models  
with different orders. By increasing the amount of regularization,
the \name produces a series of fitted mixture models with decreasing orders, as shown 
in Figure \ref{fig:coeff-fig1} for a simulated dataset. 
This qualitative representation, 
inspired by coefficient plots in penalized regression \citep{FRIEDMAN2008}, 
can also provide insight on the mixture order and parameter estimates for 
purposes of exploratory data analysis.
\begin{figure}[t]
    \includegraphics[width=.95\textwidth]{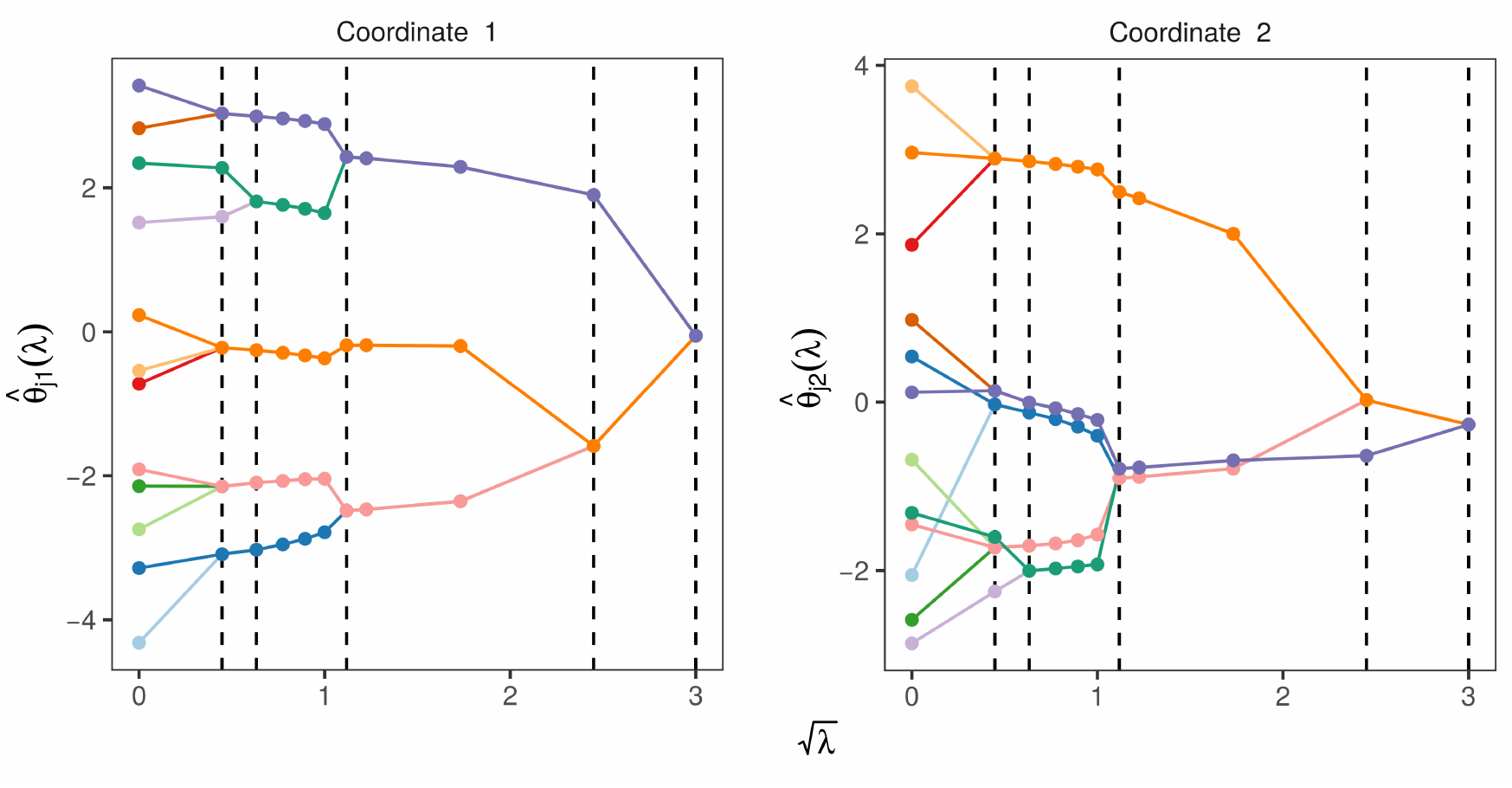}%
\caption{\label{fig:coeff-fig1}
Regularization plots based on simulated data from a 
location-Gaussian mixture  
with $K_0=5,d=2$.
The fitted atoms
$\hbtheta_j(\lambda) = (\htheta_{j1}(\lambda), \htheta_{j2}(\lambda))^\top$, $j=1, \dots, K=12$, 
are plotted against a regularization parameter $\lambda$. 
Across coordinates, each estimated atom is identified by a unique color.} 
\end{figure}

The main contributions of this paper are summarized as follows.
For a wide range of second-order strongly identifiable parametric families, the 
\name is shown to consistently
estimate the true order $K_0$, and achieves  the $n^{-1/2}$ rate of 
convergence in parameter estimation up to polylogarithmic factors. 
To achieve this result, the sparsity-inducing penalties used in the GSF 
must satisfy conditions which are nonstandard in the regularization literature. 
 We also derived, for the first time, sufficient conditions
 for the strong identifiability of multinomial mixture models.
Thorough simulation studies based on 
multivariate location-Gaussian and 
multinomial mixture models
show that the \name performs well in practice. 
The method is implemented for several univariate and multivariate
mixture models in the R package \package\footnote{
\url{https://github.com/tmanole/GroupSortFuse}}.

The rest of this paper is organized as follows.  
We describe the \name method, and compare it to a naive alternative 
in Section \ref{sec:GSF}. Asymptotic properties of the method are studied 
in Section \ref{sec:asymptotic}. 
Our simulation results and two real data examples are respectively 
presented in 
Sections \ref{sec:simulation} and \ref{sec:data},
and Supplement E.6. 
We close with some discussions in 
Section \ref{sec:discussion}. 
Proofs, numerical implementation, 
and additional simulation results are given in Supplements A--F.

\noindent{\bf Notation}. 
Throughout the paper, $|A|$ denotes the cardinality of a set $A$, 
and for any integer $K \geq 1$, $A^K = A \times ... \times A$ denotes 
the $K$-fold Cartesian product of $A$ with itself. $S_K$ denotes the set 
of permutations on $K$ elements $\{1, 2,\ldots, K \}$.
Given a vector $\bx =(x_1, \dots, x_d)^{\top} \in \bbR^d$, we denote its $\ell_p$-norm  
by $\norm{\bx}_p = \left(\sum_{j=1}^d |x_j|^p \right)^{1/p}$, for all $1 \leq p < \infty$. 
In the case of the Euclidean norm $\norm{\cdot}_2$, we omit the subscript and write $\norm{\cdot}$. 
The diameter of a set $A \subseteq \bbR^d$ is denoted $\text{diam}(A) = \sup\{\|x-y\|: x,y \in A\}$.
Given two sequences of real numbers $\{a_n\}_{n=1}^{\infty}$ and $\{b_n\}_{n=1}^{\infty}$, we write $a_n \lesssim b_n$ to indicate that there exists a constant $C > 0$ such that $a_n \leq C b_n$ for all $n \geq 1$. We write $a_n \asymp b_n$ if $a_n \lesssim b_n \lesssim a_n$. For any $a, b \in \bbR$, we write $a \wedge b = \min\{a, b\}$, $a \vee b = \max\{a, b\}$, and $a_+ = a \vee 0$. 
Finally, we let ${\cal G}_K = \{G: G = \sum_{j = 1}^K \pi_j \delta_{\btheta_j},~
\btheta_j \in \Theta, \pi_j \ge 0, \sum_{j=1}^K \pi_j = 1\}$
be the class of mixing measures with at most $K$ components.

\noindent
\textbf{Figures.} All the numerical and algorithmic details of the illustrative figures 
throughout this paper 
are given in Section 
\ref{sec:simulation} and Supplement D. 

\section{The Group-Sort-Fuse (GSF) Method}
\label{sec:GSF}
Let $\bY_1, \dots, \bY_n$ be a  
random sample arising from $p_{G_0}$, where $G_0 \in \calG_{K_0}$ 
is the true mixing measure with unknown order $K_0$. 
Assume an upper bound $K$ on 
$K_0$ is known---further discussion on the choice of $K$ is given 
in Section~\ref{remark}.  
The log-likelihood function of a mixing measure $G$ with $K > K_0$ 
atoms  
is said to be overfitted, and is defined by
\begin{equation}
\label{loglik}
l_n(G) = \sum_{i=1}^n \log p_G(\bY_i). 
\end{equation} 
The overfitted maximum likelihood estimator (MLE) of $G$ is given by 
\begin{equation}
\label{mle}
\bar G_n = \sum_{j=1}^{K} \bar{\pi}_j \delta_{\bar{\btheta}_j}
                    = \argmax_{G \in \calG_K} l_n(G).
 \end{equation}
As discussed in the Introduction, 
though the overfitted MLE is consistent in estimating $G_0$ under suitable metrics, 
it suffers from slow rates of convergence, and 
there may exist atoms of $\bar G_n$ whose corresponding  
mixing probabilities
vanish, and do not converge to any atoms of $G_0$.
Furthermore, from a model selection standpoint, $\bar G_n$ typically 
has order greater than $K_0$. 
In practice, $\bar G_n$ therefore overfits
the data in the following two ways which we will refer to below: 
(a) certain fitted mixing probabilities $\bar\pi_j$ may be near-zero, and
(b) some of the estimated atoms $\bar \btheta_j$  
may be in high proximity to each other.
In this section, we propose a penalized maximum likelihood approach which 
 circumvents both types of overfitting, thus leading to a consistent estimator of $K_0$.
               
Overfitting (a) can readily be addressed by imposing a lower bound on the mixing probabilities,
as was considered by \cite{HATHAWAY1986}. This lower bound, however, 
could be particularly challenging to specify in overfitted mixture models.
An alternative approach is to penalize against near-zero mixing
probabilities \citep{CHEN1996}.  
Thus, we begin by considering  
the following preliminary  
penalized log-likelihood function 
\begin{equation}
\label{tildeG} 
l_n(G) - \varphi(\pi_1, \dots, \pi_K), \quad \ G \in \calG_K,
\end{equation}
where $\varphi\equiv \varphi_n$ is a nonnegative penalty function  such that 
$\inf_{n\geq 1}\varphi_n(\pi_1, \dots, \pi_K) \to \infty$ as $\min_{1 \le j \le K} \pi_j $ $\to 0$.
We further require that $\varphi$ is invariant to relabeling
of its arguments, i.e.
$\varphi(\pi_1, \dots, \pi_K) = \varphi(\pi_{\tau(1)}, \dots, \pi_{\tau(K)})$,
for any permutation $\tau \in S_K$. 
Examples of $\varphi$ are given at the end of this section.  
The presence of this penalty ensures that the 
maximizer of \eqref{tildeG} has mixing probabilities which stay bounded away from zero.
Consequently, as shown in Theorem \ref{paramConsistency} below, 
this preliminary estimator is consistent in estimating the atoms of $G_0$, 
unlike the overfitted MLE in \eqref{mle}.
It does not, however, consistently estimate the order $K_0$ of $G_0$, 
as it does not address overfitting (b). 

Our approach is to introduce a second penalty 
which has the effect of merging
fitted atoms that are in high proximity. We achieve this by applying a 
sparsity-inducing penalty $r_{\lambda_n}$ to the distances between appropriately chosen
pairs of atoms of the overfitted mixture model with order $K$. 
It is worth noting that one could naively apply $r_{\lambda_n}$ to 
all ${K \choose 2}$ pairwise atom distances. Our simulations, however, suggest that
such an exhaustive form of penalization increases the sensitivity of the estimator 
to the upper bound $K$, as shown in Figure \ref{fig:comparisonFigure}.
Instead, given a carefully chosen sorting of the atoms in $\bbR^d$, 
our method merely penalizes their $K-1$
consecutive distances. This results in the double penalized 
log-likelihood $L_n(G)$ in \eqref{penloglik}, which we now describe 
using the following definitions. 
 
\begin{definition}
\label{cluster-part}
Let $\bt_1, \dots, \bt_K \in \Theta \subseteq \bbR^d$, 
and let $\calP = \{ \calC_1, \dots, \calC_H \}$ be a partition of 
$\{ \bt_1, \dots, \bt_K \}$, 
for some integer $1 \leq H \leq K$. Suppose
\begin{equation}
\label{clusterProperty}
\displaystyle 
\max_{\bt_i, \bt_j \in \calC_h} \norm{\bt_i - \bt_j} <
\min_{\substack{\bt_i \in \calC_h \\ \bt_l \not\in \calC_h}} \norm{\bt_i - \bt_l}, 
\quad  h = 1, \dots, H.
\end{equation}
Then, each set $\calC_h$ is said to be an atom cluster, and $\calP$ 
is said to be a cluster partition.
\end{definition}
   
According to Definition \ref{cluster-part}, 
a partition is said to be a cluster
partition if the within-cluster distances
between atoms are always smaller than the between-cluster distances.
The penalization in \eqref{tildeG} (asymptotically) 
induces a cluster partition $\{\calC_1, \dots, \calC_{K_0}\}$ of the estimated atoms.
Heuristically, the estimated atoms falling within each atom cluster $\calC_h$ approximate
some true atom $\btheta_{0j}$, and the goal of the \name is to merge these estimates, 
as illustrated in Figure \ref{fig:clusterPlot}. To do so, the \name hinges on  
the notion of \textit{cluster ordering}---a generalization
of the natural ordering on the real line, which we now define.

\begin{definition}
Let $\bt = (\bt_1, \dots, \bt_K) \in \Theta^K$. A cluster ordering is a 
permutation $\alpha_{\bt}\in S_K$
such that the following two properties hold.
\begin{enumerate}
\item[(i)] \textit{Symmetry.} For any permutation $\tau \in S_K$, if
$\bt' = (\bt_{\tau(1)}, \dots, \bt_{\tau(K)})$, then 
$\alpha_{\bt'} = \alpha_{\bt}$.
\item[(ii)] \textit{Atom Ordering.} For any integer 
$1 \leq H \leq K$ and for any cluster partition $\calP = \{ \calC_1, \dots, \calC_H \}$ 
of $\{\bt_1, \dots, \bt_K \}$, $\alpha_{\bt}^{-1}(\{j: \bt_j \in \calC_h \})$ is a 
set of consecutive integers for all $h=1, \dots, H$. 
\end{enumerate}
\label{cluster-order}
\end{definition}

\begin{figure}[H]
\begin{center}
  \includegraphics[width=4in]{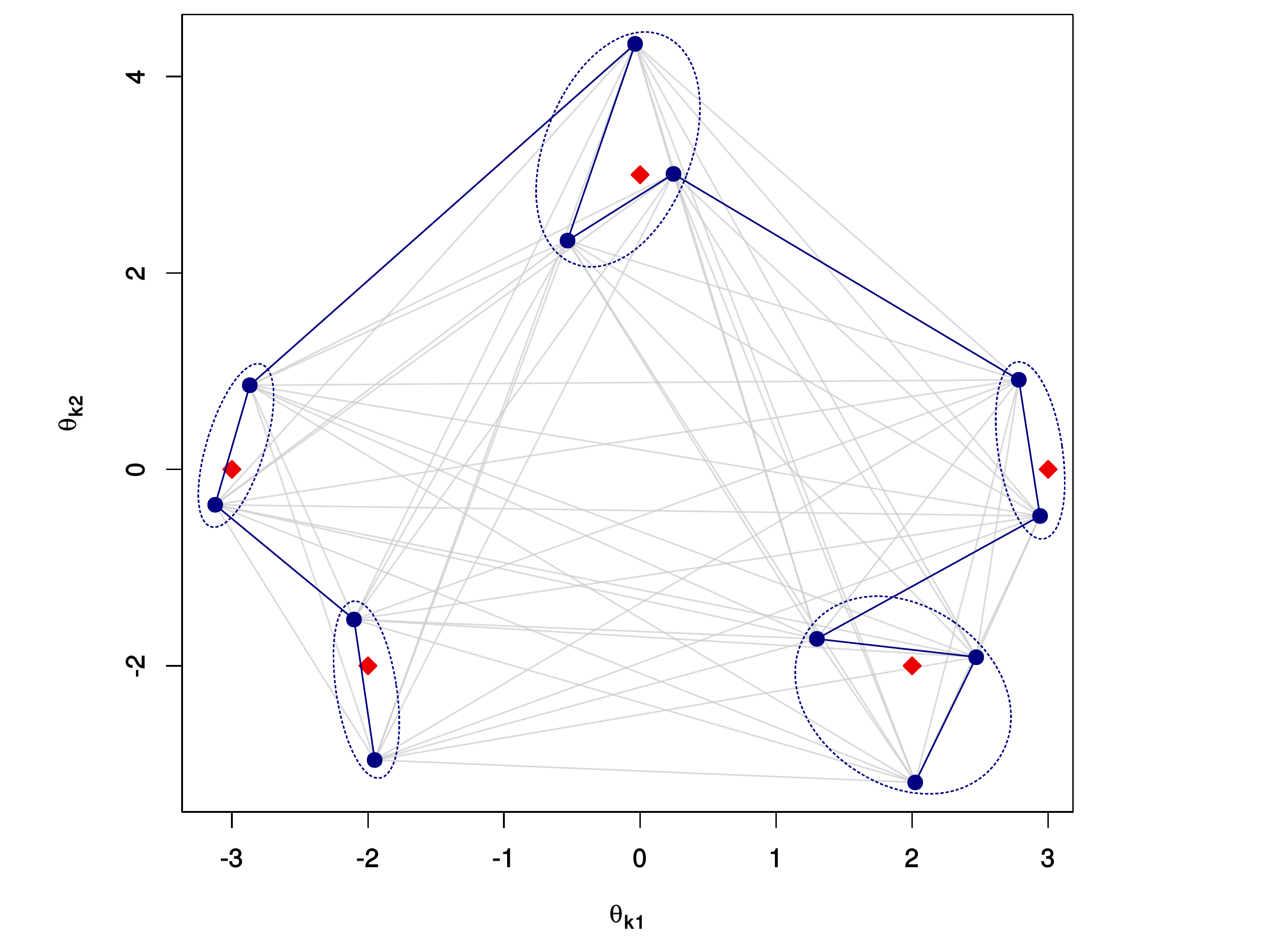} 
\end{center}
  \caption{Illustration of a cluster partition $\calP$ and a cluster ordering $\alpha_{\tilde\btheta}$ 
  with $K=12$, 
  based on the  simulated sample used in 
  Figure \ref{fig:coeff-fig1}, with
  true atoms $\btheta_{01}, \ldots, \btheta_{05}$
  denoted by lozenges ({\color{red}\ding{117}}),
  and 
  atoms $\tbtheta=(\tbtheta_1, \dots, \tbtheta_{12})$, obtained
  by maximizing the penalized log-likelihood \eqref{tildeG}, 
  denoted by disks ({\color{navyblue}\Large $\bullet$}).
  The ellipses ({\color{navyblue} \dots}) represent a choice of $\calP$    
  with $K_0=5$ atom clusters. The blue line (\textbf{\color{navyblue}---}) represents a cluster 
  ordering $\alpha_{\tbtheta}$, in the sense that $\alpha_{\tbtheta}(1)$ is the index of 
  the bottommost point, $\alpha_{\tbtheta}(2)$ 
  is the index of the following point on the line, etc. The grey lines
  ({\color{grey}---}) represent all the pairwise distances penalized by the naive method 
  defined in Figure \ref{fig:comparisonFigure}. }
\label{fig:clusterPlot}	
\end{figure}       
If $t_1, \dots, t_K \in \Theta \subseteq \bbR$ and $\bt=(t_1, \dots, t_K)$, 
then the permutation $\alpha_{\bt} \in S_K$ which induces the natural ordering
$t_{\alpha_{\bt}(1)} \leq \dots \leq t_{\alpha_{\bt}(K)}$ 
is a cluster ordering. 
When $\Theta \subseteq \bbR^d$,   
property (ii) is satisfied for any permutation $\alpha_{\bt} \in S_K$ such that
\begin{equation}
\label{alpha-example}
\alpha_{\bt}(k) = \displaystyle\argmin_{\substack{
		1 \leq j \leq K \\ 
		j \not\in \{\alpha_\bt(i):1 \leq i \leq k - 1 \}}} 
	\norm{\bt_j - \bt_{\alpha_\bt(k-1)}}, \quad k = 2,  \dots, K.
\end{equation}
$\alpha_{\bt}$ further satisfies property (i) provided $\alpha_{\bt}(1)$ is invariant to relabeling of 
the components of $\bt$. Any such 
choice of $\alpha_{\bt}$ is therefore 
a cluster ordering in $\bbR^d$, and an example is shown in 
Figure \ref{fig:clusterPlot} based on a simulated sample.

Given a mixing measure $G = \sum_{j=1}^K \pi_j \delta_{\btheta_j}$ with 
$\btheta=(\btheta_1, \dots, \btheta_K)$,
let $\alpha_{\btheta}$ be a cluster ordering. 
For ease of notation, in what follows we write $\alpha \equiv \alpha_{\btheta}$. 
Let $\bfeta_j = \btheta_{\alpha(j+1)} - \btheta_{\alpha(j)}$, 
for all $j=1, \dots, K-1$. We define the penalized log-likelihood function 
\begin{equation}
L_n(G) =
 l_n(G) - \varphi(\pi_1, \dots, \pi_K)
 - n \sum_{j=1}^{K-1} r_{\lambda_n}(\norm{\bfeta_j}; \omega_j),
\label{penloglik}
\end{equation}
where the penalty $r_{\lambda_n}(\eta; \omega)$ is a non-smooth function at $\eta= 0$ for all $\omega > 0$, 
satisfying conditions (P1)--(P3) discussed in Section \ref{sec:asymptotic}. 
In particular, $\lambda_n \geq 0$ 
is a regularization
parameter, and $\omega_j\equiv \omega_j(G) > 0$ are possibly random weights  
as defined in Section \ref{sec:asymptotic}.  
Property (i) in Definition \ref{cluster-order}, and the invariance of $\varphi$
to relabelling of its arguments, 
guarantee that $L_n(G)$ is well-defined 
in the sense that it does not change upon relabelling the atoms of $G$.
Finally, the Maximum Penalized Likelihood Estimator (MPLE) of  
$G$ is given by
\begin{equation}
\label{mple1}
\hat G_n = \sum_{j=1}^{K} \hpi_j \delta_{\hat{\btheta}_j}
           = \argmax_{G \in \calG_K} L_n(G).
 \end{equation}
To summarize,
the penalty $\varphi$ ensures  
 the asymptotic existence of a cluster partition 
$\{\calC_1, \dots, \calC_{K_0}\}$ of $\{\hbtheta_1, \dots, \hbtheta_K\}$. 
 Heuristically, the estimated atoms in each $\calC_h$ approximate one of the 
 atoms of $G_0$, and the goal of the GSF is to merge their values to be equal.
  To achieve this, Property (ii) of Definition 2 implies that 
any cluster ordering $\alpha$ is amongst the permutations in $S_K$ which maximize
the number of indices $j$ such that $\btheta_{\alpha(j)},\btheta_{\alpha(j+1)} \in \calC_h$, 
and minimize the number of indices $l$ such that $\btheta_{\alpha(l)} \in \calC_h$ and 
$\btheta_{\alpha(l+1)} \not\in \calC_h$, for all $h=1, \dots, K_0$.
Thus our choice of $\alpha$ maximizes the number of penalty 
terms $r_{\lambda_n}(\norm{\bfeta_j}; \omega_j)$
acting on distances between atoms of the same atom cluster $\calC_h$.  
The non-differentiability of $r_{\lambda_n}$ at zero ensures that, asymptotically,   
$\hbfeta_j = \boldsymbol 0$ or equivalently $\hbtheta_{\alpha(j)} = \hbtheta_{\alpha(j+1)}$
 for certain indices $j$, and thus the effective order of $\hat G_n$ 
becomes strictly less than the postulated upper bound $K$.
This is how the GSF simultaneously estimates both the mixture order 
and the mixing measure. The choice of the tuning parameter $\lambda_n$
determines the size of the penalty $r_{\lambda_n}$ and thus the estimated mixture order. 
In Section \ref{sec:asymptotic}, under certain regularity conditions, 
we prove the existence of a sequence $\lambda_n$
for which $\hat G_n$ has order $K_0$ 
with probability tending to one, 
and in Section \ref{sec:simulation} we discuss data-driven 
choices of $\lambda_n$.
\begin{figure}[H]
\begin{center}
  \includegraphics[width=5in,height=3.5in]{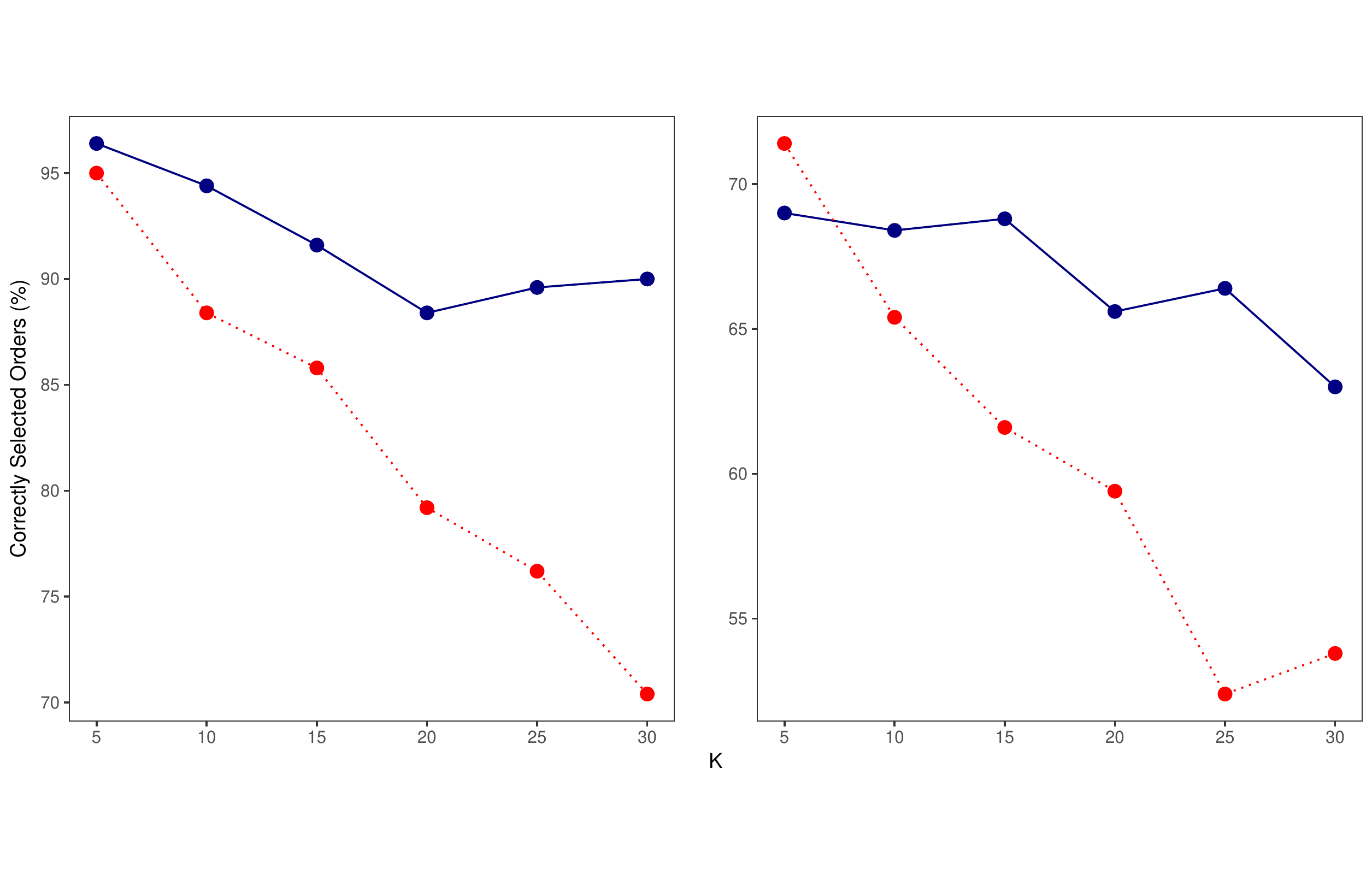}
  \vspace{-.5in}
  \caption{A comparison of the GSF ({\color{blue}---}), 
  and the naive alternative ({\color{red} \dots}) given by 
$\argmax_{G \in \calG_K} \left\{ l_n(G) - \varphi(\pi_1, \dots, \pi_K) - 
n \sum_{j\neq k} r_{\lambda_n}(\norm{\btheta_j - \btheta_k}; \omega_{jk})\right\}$.
 The results are based on 500 simulated samples of size $n=200$ from 
the bivariate Gaussian mixture Models F.1 (left, $K_0=2$) and 
F.2 (right, $K_0=3$) given in Supplement F. 
Each point represents the percentage of times that a method with 
varying upper bounds $K$ correctly estimated  $K_0$. 
} 
\label{fig:comparisonFigure}	
\end{center}
\end{figure}

{\bf Examples of the penalties $\varphi$ and $r_{\lambda_n}$.}
We now discuss some examples of penalty functions $\varphi$ and $r_{\lambda_n}$. 
The functions $\varphi(\pi_1, \dots, \pi_K) \propto -\sum_{j=1}^K \log\pi_j$ 
and $\varphi(\pi_1, \ldots, \pi_K) \propto \sum_{j=1}^K \pi^{-\iota}_j$ 
(for some $\iota > 0$) 
were 
used by \cite{CHEN1996} in the context of distance-based methods for 
mixture order estimation. As seen in Supplement D.1, 
the former is computationally convenient for EM-type algorithms,
and we use it in all demonstrative examples 
throughout this paper.
\cite{LI2009} also discuss the function 
$\varphi(\pi_1, \ldots, \pi_K) \propto -\min_{1 \le j \le K} \log \pi_j$ in 
the context of hypothesis testing for the mixture order,
which is more severe (up to a constant) than the former two penalties.

Regarding $r_{\lambda_n}$, satisfying conditions (P1)--(P3) in Section \ref{sec:asymptotic}, 
 we consider the following three penalties.  
For convenience, the first two penalties are written 
in terms of their first derivatives with respect to $\eta$.
\begin{enumerate}
\item The Smoothly Clipped Absolute Deviation 
(SCAD; \citet{FAN2001}),
\[ 
\label{scad}
r'_{\lambda_n}(\eta; \omega) \equiv r'_{\lambda_n}(\eta) = 
  \lambda_n ~I \{ |\eta| \leq \lambda_n \} +
  \frac{(a \lambda_n - |\eta|)_+}{a-1} ~I\{ |\eta| > \lambda_n \}, \qquad a > 2.
\] 
\item The Minimax Concave Penalty (MCP; \citet{ZHANG2010}), 
\[ 
\label{mcp}
r'_{\lambda_n}(\eta; \omega) \equiv r'_{\lambda_n}(\eta) = \left(\lambda_n - \frac{|\eta|}{a} \right)_+,
\qquad a > 1.
\] 

\item The Adaptive Lasso (ALasso; \citet{ZOU2006}), 
\[
r_{\lambda_n}(\eta; \omega) = \lambda_n w |\eta|.
\]
\end{enumerate}
The Lasso penalty $r_{\lambda_n}(\eta; \omega) = \lambda_n|\eta|$
does not satisfy all the conditions (P1)--(P3), and is further 
discussed in Section \ref{sec:asymptotic}.

\section{Asymptotic Study}
\label{sec:asymptotic}
In this section, we study asymptotic properties of the GSF,  
beginning with preliminaries. 
We also introduce more notation in the sequence 
that it will be needed. 
Throughout this section, except where otherwise stated, 
we fix $K \geq K_0$.

\subsection{\bf Preliminaries}
\label{prelim} 
Inspired by \cite{NGUYEN2013}, we 
analyze the convergence of 
mixing measures in $\calG_K$
using the Wasserstein distance. 
Recall that the Wasserstein distance of order $r \geq 1$ between 
two mixing measures $G = \sum_{j=1}^{K} \pi_j\delta_{\btheta_j} $ 
and $G' = \sum_{k=1}^{K'} \pi'_k\delta_{\btheta_k'} $ is given by
\begin{equation}
W_r(G, G') = \left( \inf_{\bq \in \calQ(\bpi, \bpi')} 
\sum_{j=1}^K \sum_{k=1}^{K'} q_{jk} \norm{\btheta_j - \btheta_k'}^r \right)^{\frac 1 r},
\end{equation}
where $\calQ(\bpi, \bpi')$ denotes the set of joint probability distributions 
$\bq = \{q_{jk}: 1 \leq j \leq K, \ 1 \leq k \leq K'\}$ supported on 
$\{1, \dots, K\} \times \{1, \dots, K'\}$, such that 
$\sum_{j=1}^K q_{jk} = \pi_k'$ and $\sum_{k=1}^{K'} q_{jk} = \pi_j$. 
We note that the $\ell_2$-norm of the underlying parameter space $\Theta$ is 
embedded into the definition of $W_r$. The distance between two mixing 
measures is thus largely controlled by that of their atoms. 
The definition of $W_r$ also bypasses the
non-identifiability issues arising from mixture label switching. 
These considerations make the Wasserstein distance a natural 
metric for the space $\calG_K$.  
  
A  condition which arises in likelihood-based asymptotic theory of 
finite mixture models with unknown order, called strong identifiability (in the second-order), 
is defined as follows.

\begin{definition}[Strong Identifiability; \citet{CHEN1995, HO2016strong}]
\label{def:SI}
The family $\mathcal F$ is said to be strongly identifiable (in the second-order)
if $f(\by; \btheta)$ is twice differentiable with respect to $\btheta$ for all $\by \in \calY$, 
and the following assumption holds for all integers $K \geq 1$. 
\begin{enumerate}
\item [(SI)] {Given distinct $\btheta_1, \dots, \btheta_K \in \Theta$, 
if we have $\zeta_j \in \bbR$, $\bbeta_j, \bgamma_j \in \bbR^d$, $j=1,\ldots, K$, such that
\begin{align*}
\esssup_{\by \in \calY} \left| \sum_{j=1}^K  \left\{ 
	\zeta_j f(\by; \btheta_j) + 
	\bbeta_j^\top \frac{\partial f(\by; \btheta_j)}{\partial \btheta} + 
	\bgamma_j^\top \frac{\partial^2 f(\by; \btheta_j)}{\partial \btheta \partial \btheta^\top} \bgamma_j
\right\} \right| &= 0 
\end{align*}
then $\zeta_j = 0$, $\bbeta_j = \bgamma_j = \boldsymbol 0 \in \bbR^d$, for all $j=1, \dots, K$.}
\end{enumerate}
\end{definition}

For strongly identifiable mixture models, the likelihood ratio statistic with
respect to the overfitted MLE $\bar G_n$ is stochastically bounded \citep{DACUNHA1999}.
In addition, under condition (SI), upper bounds relating the Wasserstein distance between 
a mixing measure  $G$ and $G_0$
to the Hellinger distance between the corresponding densities $p_G$ and $p_{G_0}$ have been 
established by \cite{HO2016strong}. 
In particular, there exist $\delta_0, c_0 > 0$ depending on the
true mixing measure $G_0$ such that for any $G \in \calG_K$ satisfying $W_2(G, G_0) < \delta_0$,
\begin{equation}
\label{hoIneq}
h(p_G, p_{G_0}) \geq c_0 W_2^2(G, G_0),
 \end{equation}
where $h$ denotes the Hellinger distance, 
$$h(p_G, p_{G_0}) = \left(\frac 1 2\int \big(\sqrt{p_G} - \sqrt{p_{G_0}} \big)^2 d\nu\right)^{\frac 1 2}.$$
Specific statements and discussion of these results 
are given in Supplement B, 
and are used throughout the proofs of our Theorems 
\ref{paramConsistency}-\ref{orderConsistency}. 
Further discussion of condition (SI) is given in Section 
\ref{remark}. 
We also require  
regularity conditions (A1)--(A4)
on the 
family $\calF$, condition (C) on the cluster ordering $\alpha_{\bt}$, and condition (F) on the penalty $\varphi$, 
 which we state below.  
 
Define the family of mixture densities 
\begin{equation}
\label{famden0}
\calP_K = 
\left\{ 
p_G(\by) = \int_{\bTheta} f(\by; \btheta) dG(\btheta): G \in \calG_K
\right\}.
\end{equation}
Let $p_0 = p_{G_0}$ be the density of the true finite mixture 
model with its corresponding probability distribution $P_0$. 
Furthermore, define the empirical process
\begin{equation}
\label{nuDef}
\nu_n(G) 
= 
\sqrt n\int_{\{p_0 > 0\}} \frac 1 2 \log\bigg\{\frac{p_G + p_0}{2p_0}\bigg \} d(P_n - P_0), \quad G \in \calG_K,
\end{equation}
where $P_n = \frac 1 n \sum_{i=1}^n \delta_{\bY_i}$
denotes the empirical measure.

For any $\btheta = (\theta_1, \dots, \theta_d)^\top \in \Theta$,
$\by\in\calY$, 
and $G \in \calG_K$, let
\begin{eqnarray}
\label{U0}
U(\by;\btheta, G) & = & \frac {1} {p_G(\by)}f(\by; \btheta)  \\
\label{U1}
U_{\kappa_1 \ldots \kappa_M}(\by; \btheta, G) 
& = & 
\frac 1 {p_G(\by)} \frac{\partial^{M} f(\by; \btheta)}{\partial \theta_{\kappa_1}\dots\partial \theta_{\kappa_M}}
\end{eqnarray}
for all $\kappa_1, \ldots, \kappa_M=1, \dots, d$, 
and any integer $M\ge 1$.

\label{section:regularityConditions}
The regularity conditions are given as follows.
\begin{enumerate}
\item [(A1)] {\textit{Uniform Law of Large Numbers.} \ We have,
\[
\sup_{G \in \calG_K} \frac 1 {\sqrt n} \left|\nu_n(G)\right| \overset{a.s.}{\longrightarrow} 0, 
\quad \text{as } n \to \infty.
\]
}

\item [(A2)] 
{\textit{Uniform Lipchitz Condition.} \ The kernel density $f$ is uniformly Lipchitz up to the second order \citep{HO2016strong}. That is, there exists $\delta > 0$ such that for any $\bgamma \in \bbR^d$ and $\btheta_1, \btheta_2 \in \Theta$, there exists $C > 0$ such that for all $\by \in \mathcal Y$
\begin{equation*}
\left| \bgamma^\top \left(
	\frac{\partial^2 f(\by; \btheta_1)}{\partial \btheta \partial \btheta^\top}- 
	\frac{\partial^2 f(\by; \btheta_2)}{\partial \btheta \partial \btheta^\top}
\right) \bgamma \right|
\leq C \norm{\btheta_1 - \btheta_2}_1^\delta \norm{\bgamma}_2^2.
\end{equation*}
}
\label{itm:lipschitz} 

\item [(A3)]  
\textit{Smoothness}.  
There exists $h_1 \in L^1(\nu)$ such that $|\log f(\by; \btheta)| \leq h_1(\by)$ 
$\nu$-almost everywhere. Moreover, the kernel density $f(\by; \btheta)$ 
possesses partial derivatives up to order 5 with respect to $\btheta$. 
For all $M \leq 5$, and all $\kappa_1, \dots, \kappa_M$,
\[
U_{\kappa_1 \ldots \kappa_M}(\cdot; \btheta, G_0) \in L^3 (P_0).
\] 
There also exists $h_2 \in L^3(P_0)$ and 
$\epsilon > 0$ 
such that for all $\by \in \mathcal Y$,
\[
\sup_{\norm{\btheta - \btheta_0} \leq \epsilon} 
\left| 
U_{\kappa_1 \ldots \kappa_5}(\by; \btheta, G_0)
\right|  
\leq h_2(\by).
\] 

\item[(A4)] 
\textit{Uniform Boundedness}. 
There exist $\epsilon_1, \epsilon_2 > 0$, and $q_1, q_2 \in L^2(P_0)$ 
 such that
 for all $\by \in \calY$, $|U(\by; \btheta, G)| \leq q_1(\by)$,
and for every $\kappa_1=1, \dots, d$, $|U_{\kappa_1}(\by; \btheta, G)| \leq q_2(\by)$,
uniformly for all~$G$ such that $W_2(G, G_0) < \epsilon_1$, and for all $\btheta \in \Theta$ such that $\norm{\btheta-\btheta_{0k}} < \epsilon_2$,
for some $k \in \{1, \dots, K_0\}$.
\end{enumerate}
(A1) is a standard condition required to establish consistency of nonparametric maximum likelihood estimators. A sufficient condition for (A1) to hold is that the kernel density $f(\by; \btheta)$
is continuous with respect to $\btheta$ for $\nu$-almost every $\by$ (see Example 4.2.4 of \cite{GEER2000}). 
Under condition (A2) and the Strong Identifiability
condition (SI) in Definition \ref{def:SI}, 
local upper bounds relating the Wasserstein distance over $\calG_K$ to the
Hellinger distance over $\calP_K$ in \eqref{famden0}
have been established by \cite{HO2016strong}---see Theorem \ref{thm:ho_long_results} of 
Supplement B. 
Under conditions (A3) and  (SI),  
\cite{DACUNHA1999} showed that
the likelihood ratio statistic for overfitted mixtures is stochastically bounded---see Theorem \ref{thm:lrs} 
of Supplement B. Condition (A4) is used to perform an order assessment for a score-type quantity
in the proof of the order selection consistency of the GSF
(Theorem \ref{orderConsistency}).

We further assume that the
cluster ordering $\alpha_{\bt}$ satisfies the following 
continuity-type condition.
\begin{itemize}
\item[(C)] {Let $\btheta_0=(\btheta_{01}, \dots, \btheta_{0K_0})$, and 
$\btheta=(\btheta_1, \dots, \btheta_K) \in \Theta^K$. 
Suppose there exists a cluster partition $\calP = \{\calC_1, \dots, \calC_{K_0}\}$ of $\btheta$ of size $K_0$.
Let
$\tau \in S_{K_0}$ be the permutation such that $(\btheta_{\alpha_{\btheta}(1)}, \dots, \btheta_{\alpha_{\btheta}(K)}) = (\calC_{\tau(1)}, \dots, \calC_{\tau(K_0)})$, as implied by the definition of cluster ordering.
Then, there exists $\delta > 0$ such that, if for all 
$k=1, \dots, K_0$ and $\btheta_j \in \calC_k$, we have $\norm{\btheta_j - \btheta_{0k}} < \delta$, then $\tau = \alpha_{\btheta_0}$.
}
\end{itemize} 
An illustration of condition (C) is provided in Figure \ref{conditionC}.
It is easy to verify that the
example of cluster ordering in \eqref{alpha-example} satisfies (C)
whenever the minimizers therein are unique.  
Finally, we assume that the penalty $\varphi \equiv \varphi_n$ satisfies the following condition.
\begin{enumerate}
\item [(F)] 
$\varphi_n = a_n \phi$,
where $0 < a_n = o(n)$, $a_n \not\to 0$, and $\phi:\cup_{j=1}^K (0,1]^j\to \bbR_+$ 
is Lipschitz on any compact subset of $(0,1]^j, 1 \leq j\leq K$. 
Also, for all $\pi_1, \dots, \pi_K \in (0,1]$ and
 $\rho_k \geq \pi_k, 1 \leq k \leq  K_0 \le K$, 
$\phi(\pi_1, \dots, \pi_K) \geq \phi(\rho_1, \dots, \rho_{K_0})$, and 
$\phi(\pi_1, \dots, \pi_K) \to \infty$ as $\min_j \pi_j \to 0$.

\end{enumerate}  
Condition (F) holds for all examples of functions $\varphi$ stated in Section~\ref{sec:GSF}.
When $r_{\lambda_n}(\eta; \omega)$ is constant 
with respect to $\eta$ away from zero, as is the case for 
the SCAD and MCP, condition (P2) below implies that
$a_n$ is constant with respect to $n$. 
For technical purposes, 
we require $a_n$ to diverge when $r_{\lambda_n}$ is the ALasso
penalty, ensuring that $\varphi_n$ and $nr_{\lambda_n}$ are of comparable order.
In practice, however, we notice that the GSF is hardly sensitive to the choice of $a_n$. 

Given $G = \sum_{j=1}^K \pi_j \delta_{\btheta_j} \in \calG_K$, 
we now define a choice of the weights $\omega_j\equiv \omega_j(G)$ for the 
penalty function $r_{\lambda_n}$ in \eqref{penloglik}, which are random 
and depend on $G$.  
It should be noted that the choice of these weights is relevant for 
the ALasso penalty but not for the SCAD and MCP. 
Define the estimator
\begin{equation}
\label{gntilde}
\widetilde G_n = \sum_{j=1}^K \tilde\pi_j \delta_{\tbtheta_j} = 
\argmax_{G \in \calG_K} \left\{l_n(G) - \phi(\pi_1, \ldots, \pi_K)\right\},
\end{equation}
and let $\tbtheta = (\tbtheta_1, \dots, \tbtheta_K)$. Define
$\tbfeta_j = \tbtheta_{\talpha(j+1)} - \tbtheta_{\talpha(j)}$, for all $j=1, \dots, K-1$, 
where $\talpha\equiv\alpha_{\tbtheta}$, 
and recall that 
$\bfeta_j = \btheta_{\alpha(j+1)} - \btheta_{\alpha(j)}$, where $\alpha \equiv \alpha_{\btheta}$. 
Let $u, v \in S_{K-1}$ be the permutations such that
$$\norm{\bfeta_{u(1)}} \geq \dots \geq \norm{\bfeta_{u(K-1)}}, \quad
  \norm{\tbfeta_{v(1)}} \geq \dots \geq \norm{\tbfeta_{v(K-1)}},$$
and set $\psi = v \circ u^{-1}$. 
Inspired by \cite{ZOU2006}, for some $\beta > 1$, we then define
\begin{equation}
\label{weightDef}
\omega_j = \norm{\tbfeta_{\psi(j)}}^{-\beta}, \quad j=1, \dots, K-1.
\end{equation}
Finally, we define the Voronoi diagram of 
the atoms $\{\hbtheta_1, \dots, \hbtheta_K\}$ of $\hat G_n$ in \eqref{mple1} by
$\{\hat\calV_k: 1 \leq k \le K_0\}$, where for all $k=1, \dots, K_0$,  
\begin{equation}
\label{voronoi-diag}
\hat\calV_k = \left\{\hbtheta_j: \big\lVert \hbtheta_j - \btheta_{0k} \big\rVert 
< \big\lVert \hbtheta_j - \btheta_{0l} 
\big\rVert \ , \forall l \neq k, \ 1 \le j \le K \right\},
\end{equation}
are called Voronoi cells with 
corresponding index sets   
$\hat \calI_k = \{1 \le j \le K: \hat{\btheta}_j \in \hat\calV_k \}$.

\begin{center}
\begin{figure}[t]
\centering
   \includegraphics[width=0.46\textwidth]{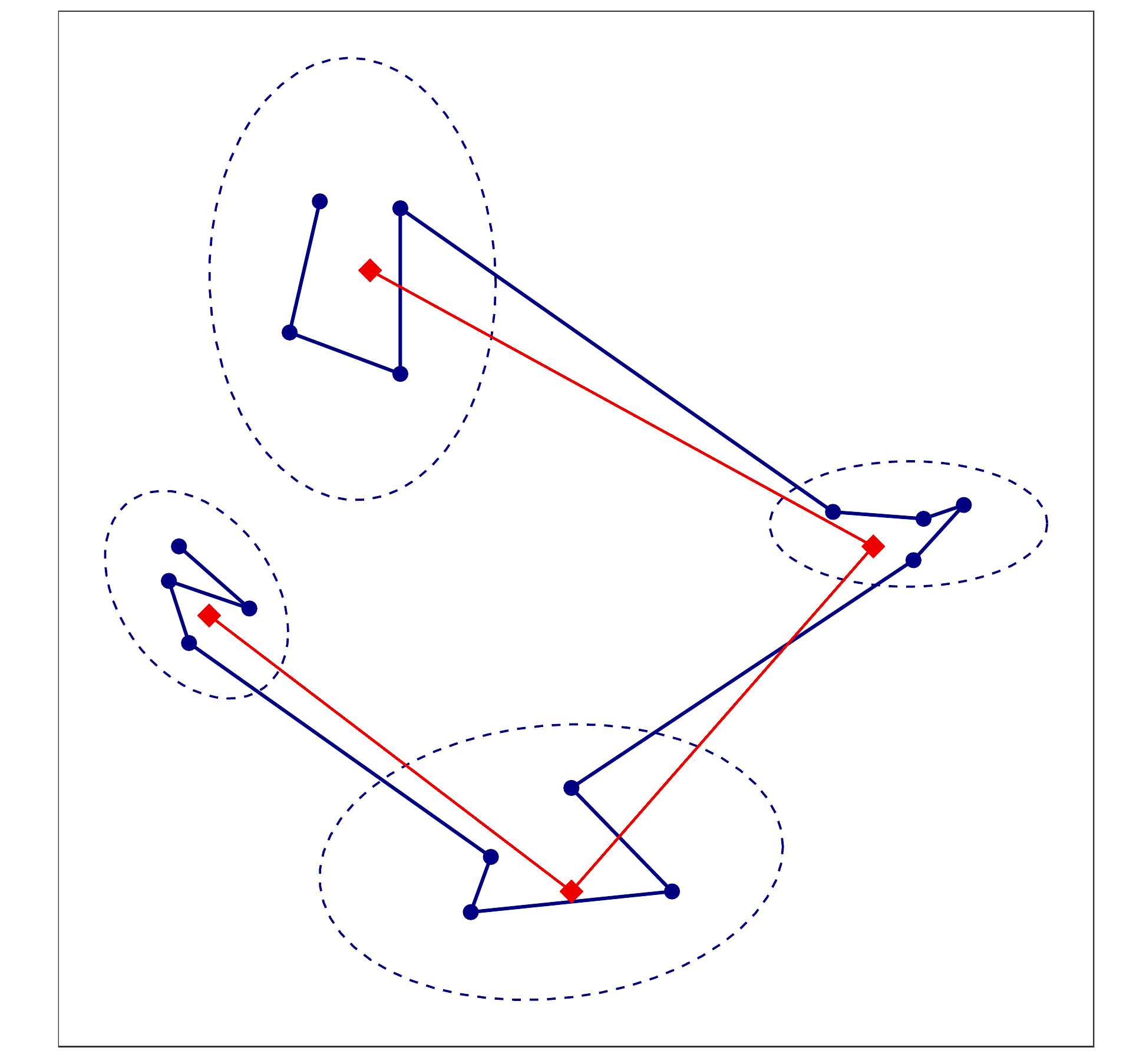}
   \includegraphics[width=0.46\textwidth]{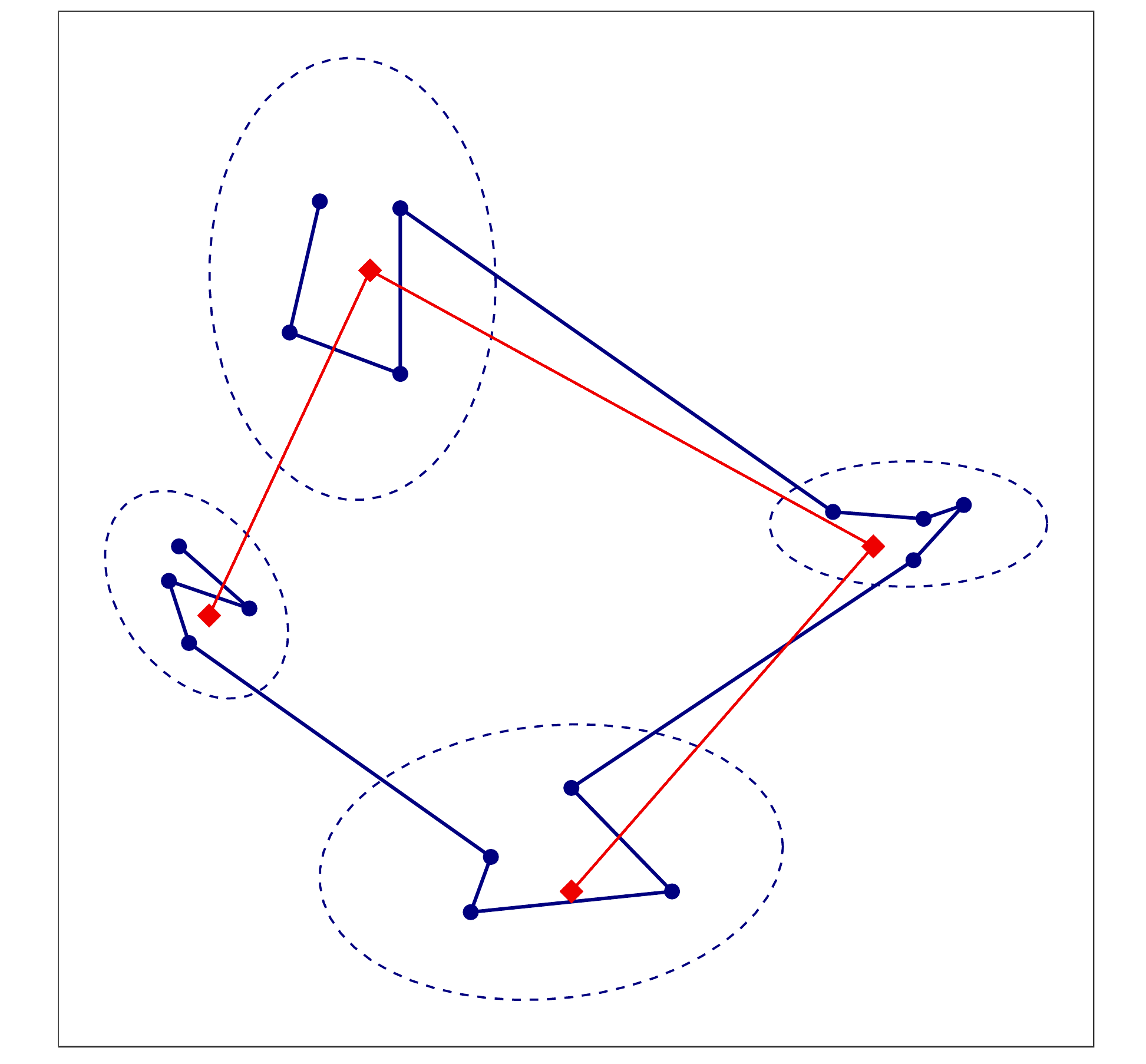}
  \caption{
  Illustration of condition (C). The points of $\btheta$ are depicted in blue ({\color{navyblue} {\Large $\bullet$}})
  and the points of $\btheta_0$ are depicted in red ({\color{red} \ding{117}}). 
  The blue solid lines ({\color{navyblue}---})
  denote the permutation $\alpha_{\btheta}$, while the red solid
  lines ({\color{red}---}) denote the permutation $\alpha_{\btheta_0}$. The ellipses ({\color{navyblue} - - -}) 
  represent a choice of cluster partition of $\btheta$. The choice of cluster ordering in the left plot satisfies 
  condition (C), while that of the right plot does not. }
  \label{conditionC}
\end{figure}
\end{center}

\subsection{\bf Main Results}
\label{main-resul}
  We are now ready to state our main results.
Theorem \ref{paramConsistency} below shows that $\{\hat\calV_k: 1 \leq k \leq K_0\}$ 
asymptotically forms a cluster partition of $\{\hbtheta_1, \dots, \hbtheta_K\}$. 
This result, together with the rate of convergence established in Theorem \ref{densityConsistency}, 
leads to the consistency of the GSF 
in estimating $K_0$, as stated in Theorem \ref{orderConsistency}. 

 \begin{theorem}
\label{paramConsistency}
Assume conditions (SI), (A1)--(A2) and (F) hold, 
and let the penalty function 
$r_{\lambda_n}$ satisfy the following condition,
\begin{enumerate}
\item [(P1)]  
$r_{\lambda_n}(\eta; \omega)\geq 0$ is a  nondecreasing function of $\eta \in \bbR_+$
which satisfies 
$r_{\lambda_n}(0; \omega) = 0$ and
$\lim_{n \to \infty} r_{\lambda_n}(\eta; \omega) = 0$, for all $\eta,\omega \in \bbR_+$.
Furthermore, for any fixed compact sets $I_1,I_2 \subseteq (0,\infty)$,
$r_{\lambda_n}(\cdot; \omega)$ is convex over $I_1$  for large $n$, 
and $\mathrm{diam}(nr_{\lambda_n}(I_1;I_2)) =O(a_n)$.
\end{enumerate}

Then, as $n \to \infty$, 
\begin{enumerate}
\item[(i)] 
 $W_r(\hat{G}_n, G_0) \to 0$, almost surely, for all $r \geq 1$. 
 \end{enumerate}
 Assume further that condition (A3) holds. Then, \begin{enumerate}
\item [(ii)]$\phi(\hpi_1, \dots, \hpi_K) = O_p(1)$. In particular, 
for every $k=1, \dots, K_0$, $\sum_{j \in \hat\calI_k} \hpi_j = \pi_{0k} + o_p(1)$. 
\item[(iii)]
For every $1 \le l \le K$, there exists a unique $1 \le k \le K_0$, such that 
\(
\big\lVert \hbtheta_l - \btheta_{0k} \big\rVert = o_p(1),
\)
 thus $\{\hat\calV_k: 1 \le k \le K_0\}$ is a cluster partition of $\{\hbtheta_1, \dots, \hbtheta_K\}$, 
with probability tending to one.
\end{enumerate}
\end{theorem}

Theorem 1.(i) establishes the consistency of $\hat G_n$ 
under the Wasserstein distance---a property shared by 
the overfitted MLE $\bar G_n$ \citep{HO2016strong}. 
This is due to the fact that, by conditions (F) and (P1), the log-likelihood
function is the dominant term in $L_n$, in \eqref{penloglik}.
Theorem 1.(ii) implies that the estimated mixing proportions $\hpi_j$ 
are stochastically bounded away from 0, which then results 
in Theorem 1.(iii) showing that every atom of $\hat G_n$ is
consistent in estimating an atom of $G_0$.
A straightforward investigation of the proof shows that
this property also holds for $\widetilde G_n$ in~\eqref{gntilde}, but not
for the overfitted MLE $\bar G_n$, which may have a subset of atoms whose limit
points are not amongst those of $G_0$.

When $K > K_0$, the result of Theorem \ref{paramConsistency} does not imply the 
consistency of $\hat G_n$ in estimating $K_0$. The latter is achieved if
the number of distinct elements of each Voronoi cell $\hat\calV_k$ is equal to one with probability
tending to one, which is shown in Theorem \ref{orderConsistency} below. To establish
this result, we require an upper bound on the rate of convergence of $\hat G_n$
under the Wasserstein distance. We obtain 
this bound 
by studying the rate of convergence of the density 
$p_{\hat G_n}$ to $p_{0}$, with respect to the Hellinger distance, and
appeal to inequality \eqref{hoIneq}.
\cite{GEER2000} (see also \cite{WONG1995}) established convergence rates for nonparametric
maximum likelihood estimators under the Hellinger distance in terms of the bracket entropy 
integral 
$$\calJ_B\left(\gamma, \bar\calP_K^{\frac 1 2}(\gamma), \nu\right) = 
 \int_0^{\gamma} \sqrt{H_B \left(u, \bar\calP^{\frac 1 2}_K(u), \nu \right)}du , 
\quad \gamma > 0
,$$
where $H_B\big(u, \bar{\calP}_K^{\frac 1 2}(u), \nu\big)$ denotes the $u$-bracket entropy 
with respect
to the $L^2(\nu)$ metric of the density family 
\[
\bar\calP_K^{\frac 1 2}(u) = \left\{\sqrt{\frac {p_G+p_{0}} 2}: G \in \calG_K, \ h\left(\frac{p_G + p_{0}}{2}, p_{0}\right) \leq u\right\},
\quad u > 0.
\]
In our work, however, the main difficulty in bounding $h(p_{\hat G_n}, p_{0})$ 
is the presence of the penalty~$r_{\lambda_n}$. The following Theorem
shows that, as $n \to \infty$, if the growth rate of~$r_{\lambda_n}$ 
 away from zero, as a function of $\eta$, is carefully controlled, 
 then $p_{\hat G_n}$ achieves the same rate of 
 convergence as the MLE~$p_{\bar G_n}$.  
\begin{theorem}
\label{densityConsistency}
Assume the same conditions as Theorem \ref{paramConsistency}, and that the cluster ordering $\alpha_\bt$ satisfies condition (C). 
For a universal constant $J > 0$, assume 
there exists a sequence of real numbers $\gamma_n \gtrsim (\log n/n)^{1/2}$ such that for all $\gamma \geq \gamma_n$,
\begin{equation}
\label{bracketEntrop}
\calJ_B\left(\gamma, \bar\calP_K^{\frac 1 2}(\gamma), \nu\right) \leq J \sqrt n \gamma^2.
\end{equation}
Furthermore, assume $r_{\lambda_n}$ satisfies the following condition,
\begin{enumerate}
\item [(P2)] The restriction of $r_{\lambda_n}$
to any compact subset of $\{(\eta, \omega) \subseteq \bbR^2: \eta, \omega > 0\}$ is Lipschitz continuous
in both $\eta$ and $\omega$, with Lipschitz constant
$\ell_n = O(\gamma_n^{3/2}/\log n)$, and $a_n\asymp n\ell_n\vee 1$.
\end{enumerate}
Then, $h(p_{\hat G_n}, p_{0}) = O_p(\gamma_n).$
\end{theorem}
Gaussian mixture models are known to 
satisfy condition \eqref{bracketEntrop} for $\gamma_n \asymp (\log n/n)^{\frac 1 2}$, 
under certain boundedness assumptions on $\Theta$ \citep{GHOSAL2001, GENOVESE2000}. 
Lemma 3.2.1 of \cite{HO2017} shows that 
\eqref{bracketEntrop} also holds for this choice of $\gamma_n$ for
many of the strongly identifiable density families which we discuss below. 
For these density families, $p_{\hat G_n}$ achieves the parametric 
rate of convergence up to polylogarithmic factors.

Let $\hat K_n$ be the order of $\hat G_n$, namely the number of distinct components $\hbtheta_j$ of $\hat G_n$
with non-zero mixing proportions.  We now 
prove the consistency of $\hat K_n$ in estimating $K_0$.
\begin{theorem}
\label{orderConsistency}
Assume the same conditions as Theorem~\ref{densityConsistency}, and assume that the 
family $\calF$ satisfies condition (A4). Suppose further
that the penalty $r_{\lambda_n}$ satisfies the following condition,
\begin{enumerate}
\item[(P3)]{$r_{\lambda_n}(\cdot; \omega)$ is differentiable for all $\omega > 0$, and 
$$\lim_{n \to \infty} \inf\left\{\gamma_n^{-1} \ \frac{\partial r_{\lambda_n}(\eta; \omega)}{\partial \eta} :
 0 < \eta \leq \gamma_n^{\frac 1 2} \log n, \ \omega \geq \left(\gamma_n^{\frac {\beta} 2} \log n \right)^{-1} \right\}= \infty,$$
 where $\gamma_n$ is the sequence defined in Theorem 2, and $\beta > 1$ is the constant in \eqref{weightDef}.}
\end{enumerate}
Then, as $n \to \infty$, 
\begin{enumerate}
\item[(i)] $\bbP(\hat K_n = K_0) \to 1.$ 
In particular, $\bbP \left( \bigcap_{k=1}^{K_0} \{|\hat \calV_k| = 1\} \right) \to 1.$
\item[(ii)] $W_1(\hat G_n, G_0) = O_p(\gamma_n).$
\end{enumerate}
\end{theorem}
Condition (P3) ensures that as $n \to \infty$, $r_{\lambda_n}$ grows sufficiently
fast in a vanishing neighborhood of $\eta=0$
 to prevent any mixing measure of order greater than $K_0$
from maximizing $L_n$. In addition to being model selection consistent, 
Theorem 3 shows that for most strongly identifiable parametric families $\calF$,
 $\hat G_n$ is a $(\log n/n)^{1/2}$-consistent estimator
of $G_0$. Thus, $\hat G_n$ improves on the $(\log n/n)^{1/4}$ rate of convergence 
of the overfitted MLE $\bar G_n$. This fact combined with Theorem 1.(iii) 
implies that the fitted atoms $\hbtheta_j$ 
are also $(\log n/n)^{1/2}$-consistent in estimating
the true atoms $\btheta_{0k}$, up to relabeling.

\subsection{\bf Remarks}\label{remark}
We now discuss several aspects of the GSF in regards to the 
(SI) condition, penalty $r_{\lambda}$, upper bound $K$, and its relation 
to existing approaches in Bayesian mixture modeling.

(I) \textbf{The Strong Identifiability (SI) Condition.} 
A wide range of univariate parametric families  
are known to be strongly identifiable, including 
most exponential families \citep{CHEN1995,CHEN2004},
and circular distributions \citep{HOLZMANN2004}. Strongly
identifiable families with multidimensional parameter space
include multivariate Gaussian distributions in location or scale, 
certain classes of Student-$t$ distributions, as well as von Mises, Weibull, 
logistic and Generalized Gumbel distributions 
\citep{HO2016strong}. In this paper, we also consider finite mixture of multinomial 
distributions. To establish conditions under which 
this family satisfies condition (SI), we begin with the following result.
\begin{proposition} 
\label{binomialSI}
Consider the binomial family with known number of trials $M\geq 1$,
\begin{equation}
\label{multinomSI}
\calF = \left\{f(y;\theta)={{M}\choose y} \theta^y (1-\theta)^{M-y}:
 \theta \in (0,1), \ y \in \{0, \dots, M\}\right\}.
\end{equation}
Given any integer $r \geq 1$, the condition $(r+1)K - 1 \leq M$ is
necessary and sufficient for $\calF$ to be strongly identifiable in 
the $r$-th order \citep{HEINRICH2018}.
That is, for any $K$ distinct points $\theta_1, \dots, \theta_K \in (0, 1)$, 
and $\beta_{jl} \in \bbR$, $j=1, \dots, K$, $l=0, \dots, r$, if 
\[
\sup_{y \in \{0, \dots, M\}} \left|\sum_{j=1}^K \sum_{l=0}^r 
\beta_{jl} \frac{\partial^l f(y; \theta_j)}{\partial\theta^l}\right| = 0,
\]
then $\beta_{jl} = 0$ for every $j=1, \dots, K$ and $l=0, \dots, r$.
\end{proposition}
The inequality $(r+1)K-1 \leq M$ is comparable to the classical 
identifiability result of \cite{TEICHER1963}, which states that binomial mixture models
are identifiable with respect to their mixing measure if and only if $2K -1\leq M$. 
Using Proposition \ref{binomialSI}, we can readily establish the following result.
\begin{corollary}
\label{corollaryMultinomial}
A sufficient condition for the multinomial family
\begin{equation}
\label{multinomSI}
\textstyle \calF = \left\{{{M}\choose{y_1, \dots, y_d}} \prod_{j=1}^{d} \theta_j^{y_j}:
 \theta_j \in (0,1), \  0\leq  y_j\leq M ,\ \sum_j^d \theta_j = 1, \ \sum_j^d y_j=M\right\}
\end{equation}
with known number of trials $M \geq 1$,
to satisfy condition (SI) is $3K-1 \leq M$.
\end{corollary}

(II) \textbf{The Penalty Function $r_{\lambda_n}$.} 
Condition (P1) 
is standard and is satisfied by most well-known
regularization functions, including the Lasso, ALasso, SCAD and MCP,
 as long as $\lambda_n \to 0$, for large enough $a_n$, as $n \to \infty$. 
Conditions (P2) and (P3) are satisfied by SCAD and MCP  
when $\lambda_n \asymp \gamma_n^{\frac 1 2} \log n$.  
When $\gamma_n \asymp (\log n / n)^{1/2}$, it follows that $\lambda_n$
decays slower than the $n^{-1/4}$ rate,
contrasting the typical rate $\lambda_n \asymp n^{-1/2}$ 
encountered in variable selection problems for parametric regression 
(see for instance \cite{FAN2001}). 

We now consider the ALasso with the weights $\omega_j$ 
in \eqref{weightDef}, which are similar to 
those  
proposed by \cite{ZOU2006} in the context of 
variable selection in regression.
Condition (P2)  
implies $\lambda_n \gamma_n^{-\frac 3 2} \log n \to 0$, while
 condition 
(P3) implies $\lambda_n \gamma^{-\frac{\beta+2}{2}}_n \to \infty$, 
where $\beta$ is the parameter in the weights. Thus, both conditions 
(P2) and (P3) are satisfied by the ALasso with the weights 
in \eqref{weightDef} only when $\beta > 1$ 
and by choosing $\lambda_n \asymp \gamma_n^{3/2} / \log n$. 
In particular, the value $\beta=1$ is invalid. 
When $\gamma_n \asymp (\log n / n)^{1/2}$, it follows that 
$\lambda_n \asymp n^{-3/4} (\log n)^{-1/4}$ which decays 
much faster than the sequence $\lambda_n$ required for the 
SCAD and MCP discussed above. 
This discrepancy can be anticipated from the fact
the weights $\omega_j$ corresponding to
nearby atoms of $\widetilde G_n$ diverge.
It is worth noting that the typical
tuning parameter for the 
ALasso in parametric regression 
is required to satisfy
$\sqrt n \lambda_n \to 0$ and $n^{\frac{1+\beta}{2}}\lambda_n \to \infty$,
for any $\beta > 0$.

Finally, we note that the Lasso penalty $r_{\lambda_n}(\eta; \omega) = \lambda_n |\eta|$ 
cannot simultaneously 
satisfy conditions (P2) and (P3), since they would require opposing choices of $\lambda_n$.
Furthermore, for this penalty, when $\Theta \subseteq \bbR$ and $\alpha$ is the 
natural ordering on the real line, that is $\theta_{\alpha(1)} \leq \dots \leq \theta_{\alpha(K)}$,
we obtain the telescoping sum 
\[
\lambda_n \sum_{j=1}^{K-1} |\eta_j| = 
\lambda_n \sum_{j=1}^{K-1} (\theta_{\alpha (j+1)} - \theta_{\alpha (j)}) = \lambda_n(\theta_{\alpha (K)} - \theta_{\alpha(1)})
\]
which fails to penalize the vast majority of the overfitted components.

(III) \textbf{Choice of the Upper Bound $K$.}
By Theorem \ref{orderConsistency}, as long as the upper bound on the mixture order 
satisfies $K \ge K_0$, the GSF provides a consistent 
estimator of $K_0$.
The following result shows 
the behaviour of the GSF for a misspecified bound $K < K_0$. 
\begin{proposition}
\label{prop:misspec_K} 
Assume that the family $\calF$ satisfies condition (A3), and 
that the mixture family $\{p_G: G \in \calG_{K}\}$ 
is identifiable,
Then, for any $K < K_0$, as $n \to \infty$, 
the GSF order estimator $\hat K_n$ satisfies:
\(
\bbP(\hat K_n = K) \to 1.
\)
\end{proposition}

Guided by the above result, if the GSF chooses the prespecified 
upper bound $K$ as the estimated order, the bound is likely misspecified 
and larger values should also be examined. 
This provides a natural heuristic for  
choosing an upper bound $K$ for the GSF in practice, which we further elaborate upon 
in Section \ref{sensit} of the simulation study.

(IV) \textbf{Connections between the GSF and Existing Bayesian Approaches.} 
When 
$\varphi(\pi_1, \dots, \pi_K)= (1-\gamma) \sum_{j=1}^K \log\pi_j$, 
for some $\gamma > 1$, the estimator  $\widetilde G_n$ in \eqref{gntilde} 
can be viewed as the posterior mode of 
the overfitted Bayesian mixture model 
\begin{align}
\label{eq:overfitted_prior_theta} 
\btheta_1, \dots, \btheta_K ~&\overset{\text{iid}}{\sim}~ H, \\
\label{eq:overfitted_prior_pi}
(\pi_1, \dots, \pi_K) ~&\sim ~\text{Dirichlet}\left(\gamma, \dots,\gamma \right) ,~
\bY_i | G = \sum_{j=1}^K \pi_j \delta_{\btheta_j}
~~\overset{\text{iid}}{\sim} ~p_G,~
i=1, \dots, n,
\end{align}
where $H$ is 
a uniform prior on the (compact) set $\Theta \subseteq \bbR^d$. 
Under this setting, \cite{ROUSMENG2011} 
showed that 
when $\gamma < d/2$, the posterior
distribution has the effect of asymptotically emptying out redundant
components of the overfitted mixture model, such that  
the posterior expectation of the mixing probabilities of the $(K-K_0+1)$ extra 
components decay at the rate $n^{-1/2}$, up to polylogarithmic factors. 
On the other hand, if $\gamma > d/2$, 
two or more of the posterior atoms with non-negligible  
mixing probabilities will have the tendency to approach each other. 
The authors discuss that the former case results in more 
stable behaviour of the posterior distribution. 
In contrast, under our setting with the choice $\gamma > 1$, 
Theorem 
\ref{paramConsistency}.(i) implies that all the mixing probabilities of $\widetilde G_n$ 
are bounded away from zero with probability tending to one. 
This behaviour matches their above setting $\gamma > d/2$, though with
a generally different cut-off for $\gamma$.
We argue that the GSF does not suffer from the instability
described by \cite{ROUSMENG2011} in this setting, as it
proposes a simple procedure for merging nearby atoms using the second 
penalty $r_\lambda$ 
in \eqref{penloglik}, hinging upon the notion of cluster ordering. 
From a Bayesian standpoint, this penalty can be viewed as replacing the iid prior $H$ in 
\eqref{eq:overfitted_prior_theta} by the following exchangeable and non-iid prior
\begin{equation}
\label{eq:gsf_prior}
(\btheta_1, \dots, \btheta_K) \sim p_{\btheta}(\btheta_1, \dots, \btheta_K) \propto
    \prod_{j=1}^{K-1} \exp\Big\{-r_{\lambda}\big(\norm{\btheta_{\alpha_{\btheta}(j+1)}
- \btheta_{\alpha_{\btheta}(j)}}; \omega_j\big)\Big\}
\end{equation}
up to rescaling of $r_\lambda$, which  places high-probability mass on 
nearly-overlapping atoms. 
On the other hand, 
\cite{PETRALIA2012}, \cite{xie2020}
replace $H$ by 
so-called repulsive priors, which favour diverse atoms, 
and are typically used with
$\gamma < d/2$. For example, \cite{PETRALIA2012}
study the prior 
\begin{equation}
\label{eq:repulsive_prior}
(\btheta_1, \dots, \btheta_K) \sim p_{\btheta}(\btheta_1, \dots, \btheta_K) \propto
    \prod_{j < k}^K \exp\left\{-\tau \norm{\btheta_{j} - \btheta_{k}}^{-1}
    \right\},\quad \tau > 0.
\end{equation}
In contrast to the GSF, 
the choice $\gamma < d/2$ 
ensures vanishing posterior mixing probabilities
corresponding to redundant components, which is further encouraged 
by the repulsive prior \eqref{eq:repulsive_prior}. 
Without a post-processing step which thresholds
these mixing probabilities, however, this methods do not
yield consistent order selection. 
It turns out that by further placing a prior on $K$,
order consistency can be obtained 
\citep{nobile1994,miller2018}.
 
A distinct line of work in nonparametric Bayesian mixture modeling
places a prior, 
such as a Dirichlet process,
directly on the mixing measure $G$. 
Though the resulting 
posterior typically has infinitely-many atoms,
consistent estimators of $K_0<\infty$ 
can be obtained using post-processing techniques, such as the Merge-Truncate-Merge (MTM) 
method of \cite{guha2019}. 
Both the GSF and MTM aim at reducing the overfitted mixture order
by merging nearby atoms. 
Unlike the GSF, however, the Dirichlet
process mixture's posterior may have vanishing 
mixing probabilities, hence a single merging stage of its atoms 
is insufficient to obtain an asymptotically 
correct order. The MTM thus also truncates 
such redundant components,
and performs a second merging of their mixing probabilities 
to recover  a proper mixing measure.
Both the truncation and merging stages use hard-thresholding rules.  
We compare the two methods in our simulation study,
Section \ref{sec:sim-merging}. 
  
\section{Simulation Study}
\label{sec:simulation}
We conduct a simulation study to assess the finite-sample performance of the GSF. 
We  develop a modification of the EM algorithm 
to obtain an approximate solution to the optimization 
problem in \eqref{mple1}. The main ingredients
are the Local Linear Approximation algorithm of 
\cite{ZOU2008} 
for nonconcave penalized likelihood models, 
and the proximal gradient method
\citep{NESTEROV2004}.
Details of our numerical solution are given in 
Supplement D.1. The algorithm is implemented in our R package \package.
 
In the GSF, 
the tuning parameter $\lambda$ regulates the order of the fitted
model. Figure \ref{fig:coeff-fig1} (see also Figure
\ref{fig:tuning} in Supplement E.7)
 shows the evolution of 
the parameter estimates $\hbtheta_j(\lambda)$ 
for a simulated dataset, over a grid of $\lambda$-values. 
These qualitative representations can provide insight about the order of the mixture model, 
for purposes of exploratory data analysis. 
For instance, as seen in the figures, 
when small values of $\lambda$ lead to a significant reduction in the postulated order $K$,  
a tighter bound on $K_0$ can often be obtained. 
In applications where a specific choice of $\lambda$ is required, common techniques
include $v$-fold Cross Validation and the BIC, applied
directly to the MPLE for varying values of $\lambda$ \citep{Zhangetal2010}.
In our simulation,  
we use the BIC due to its low computational burden. 

\textbf{Default Choices of Penalties, Tuning Parameters, and Cluster Ordering.}
Throughout 
all simulations and real data analyses in this paper, 
including those contained in 
Figures~\ref{fig:coeff-fig1}-\ref{fig:comparisonFigure}, 
the following choices
were used by default unless otherwise specified. 
We used the penalty $\varphi(\pi_1, \dots, \pi_K) = (1-\gamma) \sum_{j=1}^K \log\pi_j$,
with the constant $1-\gamma \approx -\log 20$ 
following the suggestion of \cite{CHEN1996}. 
The penalty $r_{\lambda}$ is taken to be the SCAD by default,
though we also consider simulations below which employ
the MCP and ALasso penalties. For the ALasso, 
the weights $\omega_{j}$ are specified as in \eqref{weightDef}.
The tuning parameter $\lambda$ is selected using the BIC as described above. 
The cluster ordering $\alpha_{\btheta}$ is chosen
as in \eqref{alpha-example}. We recall that
this choice does not constrain $\alpha_{\btheta}(1)$---in our
simulations, we chose this value using a heuristic which ensures
that $\alpha_{\btheta}$ reduces to the natural ordering on $\bbR$ 
in the case $d=1$. Further numerical details are given in Supplement~D.2. 

\subsection{\bf Parameter Settings and Order Selection Results}
\label{sim-setting}
Our simulations are based on multinomial and multivariate location-Gaussian mixture models.
We compare the GSF under the SCAD (GSF-SCAD), MCP (GSF-MCP) 
and ALasso (GSF-ALasso) penalties to the AIC, BIC, and ICL 
\citep{BIERNACKI2000},  
as implemented in the R packages \texttt{mixtools} \citep{BENAGLIA2009} 
and \texttt{mclust} \citep{FRALEY1999}. 
ICL performed similarly to the BIC in our multinomial simulations, 
but generally underperformed in our Gaussian simulations. Therefore, below
we only discuss the performance of AIC and BIC.
 
We report the proportion of times that each method selected the correct order $K_0$, out of 500 replications, 
based on the models described below. For each simulation, we also report
detailed tables in Supplement E with the number of times each method 
incorrectly selected orders other than $K_0$. 
 We fix the upper bound $K=12$ throughout this section. 
For this choice, 
the effective number of parameters of the mixture models hereafter is less than
the smallest sample sizes considered.

{\bf Multinomial Mixture Models.} 
The density function of multinomial mixture
model of order $K$ is given by
\begin{equation}
p_G(\by) = \sum_{j=1}^K \pi_j {{M}\choose{y_1, \dots, y_d}} \prod_{l=1}^d \theta_{jl}^{y_l}
\end{equation}
with $\btheta_j=(\theta_{j1}, \dots, \theta_{jd})^\top \in (0,1)^d$, 
$\by=(y_1, \dots, y_d)^\top \in \{1, \dots, M\}^d$, where
$\sum_{l=1}^d \theta_{jl} = 1$, $\sum_l^d y_l = M$. 
We consider 7 models with true orders $K_0=2, 3, ..., 8$, dimensions $d=3,4,5$, 
and $M=35, 50$ to satisfy the strong identifiability condition $3K-1 \leq M$ described 
in Corollary \ref{corollaryMultinomial}. 
The parameter settings are given in Table \ref{tab:multinomial_models}.
The results for $M=50$ are reported in Figure \ref{fig:multinomialResults} below. 
Those for $M=35$ are similar, and are relegated 
to Supplement E.1. 
The simulation results are based on the sample sizes $n=100, 200, 400$.  
\begin{table}[H]
\centering 
\begin{tabular}{l*{4}{c}r}
\hline 
	Model & 1 & 2 & 3  \\
	\hline 
	$\pi_1, \btheta_1$ & $.2, (.2, .2, .2, .2, .2)$ & $\frac 1 3, (.2, .2, .2, .2, .2)$ &$.25, (.2, .2, .6)$   \\
	$\pi_2, \btheta_2$ & $.8, (.1, .3, .2, .1, .3)$ & $\frac 1 3, (.1, .3, .2, .1, .3)$ &$.25, (.2, .6, .2)$  \\
	$\pi_3, \btheta_3$ & &$\frac 1 3, (.3, .1, .2, .3, .1)$ &$.25,(.6, .2, .2)$  \\
	$\pi_4, \btheta_4$ & & &$.25,(.45, .1, .45)$ \\[0.1in]
\end{tabular}
\begin{tabular}{l*{5}{c}r}
 \hline
	Model & 4 & 5 & 6 & 7 \\
	\hline  
	$\pi_1, \btheta_1$ & $.2, (.2, .2, .6)$ & $\frac 1 6, (.2, .2, .6)$   & $\frac 1 7, (.2, .2, .6)$   & $.125, (.2, .2, .2, .4)$ \\
	$\pi_2, \btheta_2$ &$.2, (.6, .2, .2)$ & $\frac 1 6, (.2, .6, .2)$   & $\frac 1 7, (.2, .6, ,2)$   & $.125, (.2, .2, .4, .2)$ \\
	$\pi_3, \btheta_3$ &$.2, (.45, .1, .45)$ & $\frac 1 6, (.6, .2, .2)$   & $\frac 1 7, (.6, .2, .2)$   & $.125, (.2, .4, .2, .2)$ \\
	$\pi_4, \btheta_4$  &$.2, (.2, .7, .1)$ & $\frac 1 6, (.45, .1, .45)$ & $\frac 1 7, (.45, .1, .45)$ & $.125, (.4, .2, .2, .2)$ \\
	$\pi_5, \btheta_5$  & $.2, (.1, .7, .2)$ & $\frac 1 6, (.2, .7, .1)$   & $\frac 1 7, (.1, .7, .2)$   & $.125, (.1, .3, .1, .5)$ \\
	$\pi_6, \btheta_6$ &  & $\frac 1 6, (.1, .7, .2)$   & $\frac 1 7, (.7, .2, .1)$   & $.125, (.1, .3, .5, .1)$ \\
	$\pi_7, \btheta_7$ & &                             & $\frac 1 7, (.1, .2, .7)$   & $.125, (.1, .5, .3, .1)$ \\ 
 	$\pi_8, \btheta_8$ & &                             &                             & $.125, (.5, .1, .3, .1)$ \\
	\hline
\end{tabular}
\caption{\label{tab:multinomial_models} Parameter settings for the multinomial mixture Models 1--7.}
\end{table}

Under Model 1, all five methods selected the correct order most often,
and exhibited similar performance across all the sample sizes---the results
are reported in Table \ref{tab:supp_multinomial_1}  of Supplement E.1. 
The results for Models 2-7 with orders $K_0=2, 3, 4, 5$, are plotted by percentage of correctly selected 
orders in Figure \ref{fig:multinomialResults}. 
Under Model 2, the correct order is selected most frequently by the BIC and GSF-ALasso, 
for all the sample sizes.
Under Models 3 and 4, the GSF with all three penalties, in particular the GSF-ALasso, 
outperforms AIC and BIC. 
Under Models 5-7, all methods selected the correct order for $n=100$ fewer than 55\% of the time. For $n=200$, the GSF-SCAD and GSF-MCP 
select the correct number of components more than 55\% of the time, unlike AIC and BIC.  All three GSF penalties continue to outperform the other methods when $n=400$.

\begin{figure}[t]
  \includegraphics[width=\linewidth]{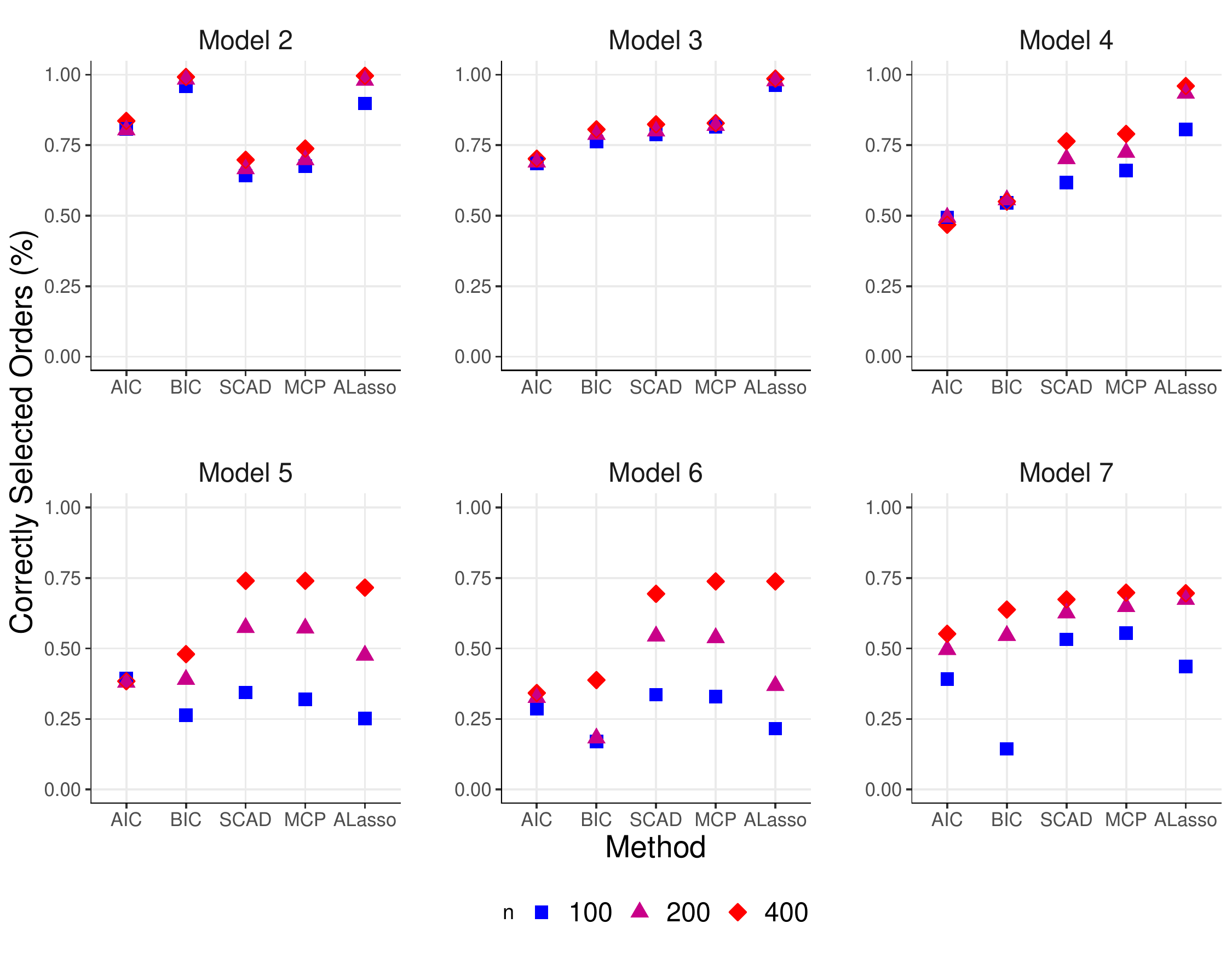}  
    \captionof{figure}{ Percentage of correctly selected orders for the multinomial mixture models.} 
    \label{fig:multinomialResults}
\end{figure}

{\bf Multivariate Location-Gaussian Mixtures 
with Unknown Covariance Matrix.}  
The density function of a multivariate Gaussian mixture model in mean,
of order $K$, is given by 
\begin{table}[t]
\centering
\resizebox{.9\textwidth}{!}{ 
\begin{tabular}{l*{7}{c}r}
\hline
	Model  & $\sigma_{ij}$ & $\pi_1, \bmu_1$ & $\pi_2, \bmu_2$ & $\pi_3, \bmu_3$ & $\pi_4, \bmu_4$ & $\pi_5, \bmu_5$ &  &   \\
	\hline \\
 1.a & $I(i=j)$ & $.5, (0, 0)^\top$ & .5, $(2, 2)^\top$ \\ 
 1.b & $(0.5)^{|i-j|}$ & $.5, (0, 0)^\top$ & .5, $(2, 2)^\top$ \\ \\
 2.a & $I(i=j)$ & $.25, (0, 0)^\top$ & $.25, (2, 2)^\top$ & $.25, (4, 4)^\top$ & $.25, (6, 6)^\top$ \\ 
 2.b & $(0.5)^{|i-j|}$ & $.25, (0, 0)^\top$ & $.25, (2, 2)^\top$ & $.25, (4, 4)^\top$ & $.25, (6, 6)^\top$ \\ \\
 $\begin{matrix} 
 \text{3.a} \\ \text{3.b}\end{matrix}$ & $\begin{matrix}I(i=j)\\ (0.5)^{|i-j|}\end{matrix}$ & $\frac 1 3, \begin{pmatrix} 0\\ 0\\ 0\\ 0 \end{pmatrix}$ & $\frac 1 3, \begin{pmatrix} 2.5\\ 1.5 \\ 2 \\ 1.5 \end{pmatrix}$ 
 																							  & $\frac 1 3, \begin{pmatrix} 1.5 \\ 3 \\ 2.75 \\ 2 \end{pmatrix}$  \\ \\
$\begin{matrix}  \text{4.a} \\ \text{4.b}\end{matrix}$ & $\begin{matrix}I(i=j)\\ (0.5)^{|i-j|}\end{matrix}$ & $\frac 1 5, \begin{pmatrix} 0 \\ 0 \\ 0 \\ 0 \\ 0 \\ 0 \end{pmatrix}$ 
 			    & $\frac 1 5, \begin{pmatrix} -1.5 \\ 2.25 \\ -1 \\ 0 \\ .5 \\ .75 \end{pmatrix}$
 			    & $\frac 1 5, \begin{pmatrix} .25 \\ 1.5 \\ .75 \\ .25 \\ -.5 \\ -1\end{pmatrix}$ 
 			    & $\frac 1 5, \begin{pmatrix} -.25 \\ .5 \\ -2.5 \\ 1.25 \\ .75 \\ 1.5 \end{pmatrix}$  		 			    	   
 			    & $\frac 1 5, \begin{pmatrix} -1 \\ -1.5 \\ -.25 \\ 1.75 \\ -.5 \\ 2\end{pmatrix}$
			    \\ \\	

$\begin{matrix}  \text{5.a} \\ \text{5.b}\end{matrix}$ & $\begin{matrix}I(i=j)\\ (0.5)^{|i-j|}\end{matrix}$ &  $\frac 1 5, \begin{pmatrix} 0 \\ 0 \\ 0 \\ 0 \\ 0 \\ 0 \\ 0 \\ 0\end{pmatrix}$ 
 			    & $\frac 1 5, \begin{pmatrix} 1 \\ 1.5 \\ 0.75 \\ 2 \\ 1.5 \\ 1.75 \\ 0.5 \\ 2.5\end{pmatrix}$
 			    & $\frac 1 5, \begin{pmatrix} 2 \\ 0.75 \\ 1.5 \\ 1 \\ 1.75 \\ 0.5 \\ 2.5 \\ 1.5 \end{pmatrix}$ 
 			    & $\frac 1 5, \begin{pmatrix} 1.5 \\ 2 \\ 1 \\ 0.75 \\ 2.5 \\ 1.5 \\ 1.75 \\ 0.5 \end{pmatrix}$  		 			    	   
 			    & $\frac 1 5, \begin{pmatrix} 0.75 \\ 1 \\ 2 \\ 1.5 \\ 0.5 \\ 2.5 \\ 1.5 \\ 1.75 \end{pmatrix}$ 	 \\	
			    \hline		     
\end{tabular}
}
\caption{\label{tab:normalModels} 
Parameter settings for the multivariate Gaussian mixture models.  
}
\end{table} 
\[
p_G(\by) = \sum_{j=1}^K \pi_j \frac{1}{\sqrt{(2 \pi)^d |\bSigma}|} 
	\exp \left\{- \frac 1 2 (\by- \bmu_j)^\top \bSigma^{-1} (\by- \bmu_j) \right\},
\]
where $\bmu_j \in \bbR^d, j=1, \dots, K$, and $\bSigma = \{\sigma_{ij}: i, j= 1, \ldots, d\}$ is a positive definite $d \times d$ covariance matrix.
We consider the 10 mixture models in Table \ref{tab:normalModels} 
with true orders $K_0 = 2, 3, 4, 5$, and with dimension $d=2, 4, 6, 8$. 
For each model, we consider both an identity and  non-identity covariance 
matrix $\bSigma$,  which is estimated as an unknown parameter. 
The simulation results are based on the sample sizes $n=200, 400, 600, 800$. 

The results for Models 1.a, 1.b, 3.a, 3.b, 4.a, 4.b are plotted by 
percentage of correctly selected orders in Figure \ref{fig:normalRes} below. 
Detailed results for the more 
challenging Models 2.a, 2.b, 5.a and 5.b are reported by percentage of selected orders 
between $1, \dots, K(=12)$ in Tables \ref{tab:supp_normal_2} and \ref{tab:supp_normal_5}
of Supplement E.2. 
 
\begin{figure}[t]
\begin{center}
  \includegraphics[width=0.95\linewidth]{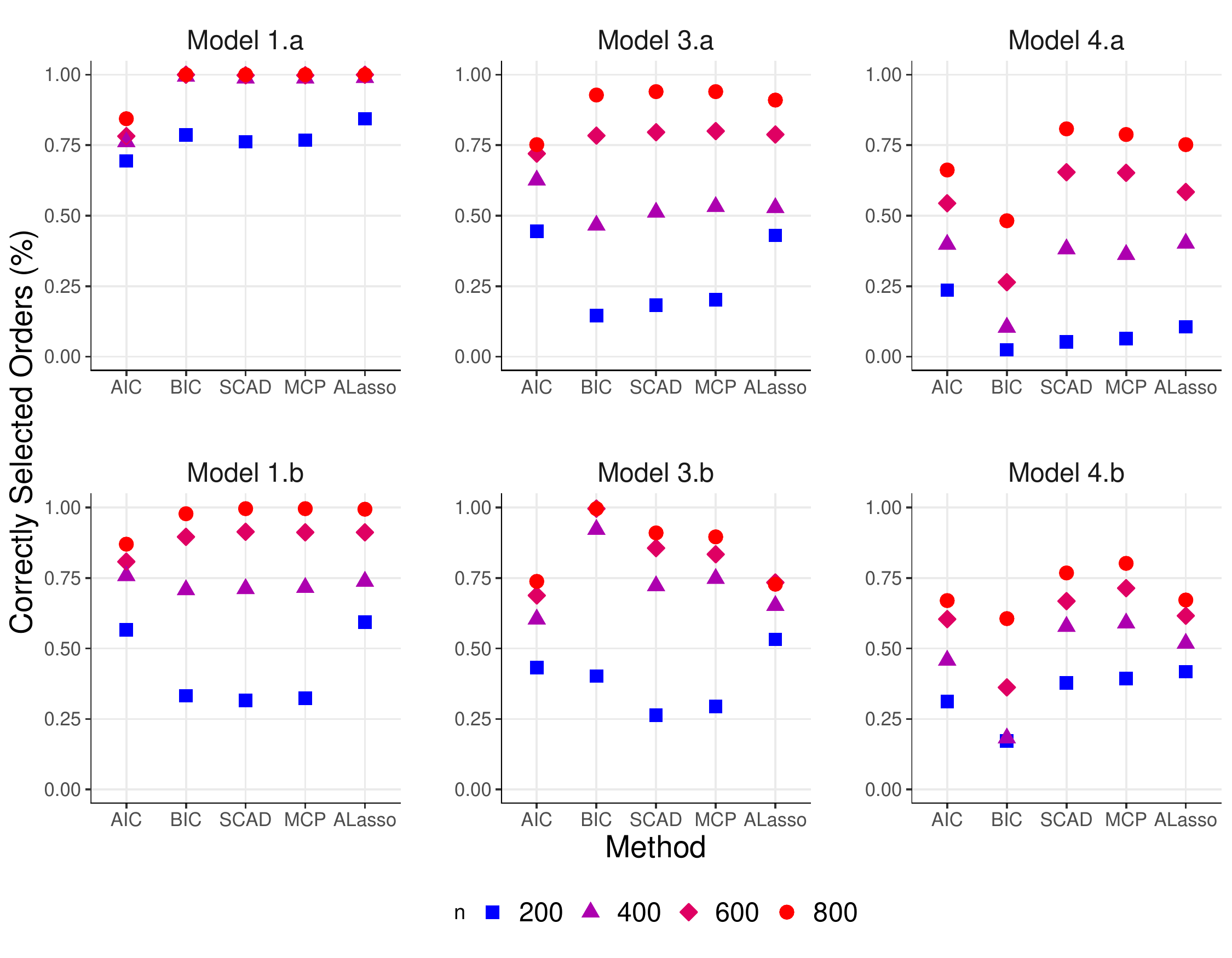}
  \caption{
  Percentage of correctly selected orders for the  multivariate Gaussian mixture models.} 
\label{fig:normalRes}	
\end{center}
\end{figure}

In Figure \ref{fig:normalRes}, 
under  
Models 1.a and 1.b with $d=2$, all the methods selected the correct number of components most frequently for $n=400, 600, 800$; however, the performance of all methods deteriorates in Model 1.b with non-identity covariance matrix when $n=200$.
Under Model 3.a with $d=4$, all methods perform similarly for $n=400, 600, 800$, but the GSF-ALasso and the AIC outperformed the other methods for $n=200$. Under Model 3.b, the BIC outperformed the other methods for $n=400, 600, 800$, but the GSF-ALasso again performed the best for $n=200$. In Models 4.a and 4.b with $d=6$, the GSF with the three penalties outperformed AIC and BIC across all sample sizes. 

From Table \ref{tab:supp_normal_2}, under Model 2.a with $d=2$ and 
identity covariance matrix, the BIC and the GSF with the three penalties 
underestimate and the AIC overestimates the true order, for sample 
sizes $n=200, 400$. The three GSF penalties significantly outperform 
the AIC and BIC, when $n=600, 800$. For the more difficult Model 2.b 
with non-identity covariance matrix, all methods underestimate across all 
sample sizes considered, but the AIC selects the correct order most frequently. 
From Table \ref{tab:supp_normal_5}, under Model 5.a, all methods apart from AIC underestimated $K_0$ for $n=200, 400, 600$, and the three GSF penalties outperformed the other methods when $n=800$. Interestingly, the performance of all methods improves for Model 5.b with non-identity covariance matrix. Though all methods performed well for $n=400, 600, 800$, the BIC did so the best, while the GSF-ALasso exhibited the best performance when $n=200$.

In summary, depending on the models and sample sizes considered here, 
in some cases AIC or BIC exhibit the best performance, while in others the GSF 
based on at least one of the penalties 
(ALasso, SCAD, or MCP) outperforms. 
The universality of information criteria 
in almost any model selection problem is in part
due to their ease of use on the investigator's 
part, while many other methods require specification of multiple tuning parameters.
Though we defined the GSF 
in its most general form, 
our empirical investigation suggests
that, other than
$\lambda$ and $K$, its tuning
parameters ($\alpha_{\bt},\varphi,\omega_j$, and choices therein) 
may not need to be tuned beyond their default choices used here.
We have shown that off-the-shelf
data-driven methods for selecting $\lambda$ 
yield reasonable performance.  
We next discuss
the choice of the bound $K$. 

\subsection{\bf Sensitivity Analysis for the Upper Bound $K$}
\label{sensit}

In this section, we assess the sensitivity of the GSF with
respect to the choice of upper bound $K$ via simulation. 
Specifically, we show the behaviour of the GSF
for a range of $K$-values which are both misspecified ($K < K_0$)
and well-specified ($K \geq K_0$).
In the former case,  
by Proposition \ref{prop:misspec_K},
the GSF is expected to select the order $K$, whereas in the
latter case, by Theorem \ref{orderConsistency}, 
the GSF selects the correct $K_0$ with high probability. 

We consider the multinomial 
Models 3 ($K_0= 4$) and 5 ($K_0=6$) with sample size $n=400$, 
and the Gaussian Models 3.a ($K_0= 3$) 
and 4.a ($K_0= 5$) with sample size $n=600$. 
The results are based on 
$80$ simulated samples from each model. 
For each sample, we apply the GSF-SCAD 
with $K= 2, \ldots, 25$,
and then report the most frequently estimated 
order $\hat K$, as well as the average 
estimated order over the $80$ samples. 
The results are given in 
Figure~\ref{fig:sensitivity_simulation}.
Detailed results are reported 
by percentage of selected orders with respect to the bounds $K= 2, \dots, 25$, 
in Tables \ref{tab:supp_multiResults1}-\ref{tab:supp_normalResults2} 
of Supplement E.3.

For all four models, 
it can be seen that the GSF estimates the order $K$
most frequently when $K < K_0$. In fact, it does so
on every replication for $K=1,2$ (resp. $K=1,2,3$)
under multinomial Model 3 (resp. Model 5).
When $K \geq K_0$, the GSF correctly
estimates the order $K_0$ most frequently
for all four models. 
Although the average selected order is seen to
slightly deviate from $K_0$ as $K$ increases 
(as was already noted in Figure \ref{fig:comparisonFigure}),
the overall behaviour of the GSF is remarkably stable 
with respect to the choice of $K$. 
The resulting elbow shape of the solid red lines in Figure 
\ref{fig:sensitivity_simulation} is anticipated 
by Theorem \ref{orderConsistency} and Proposition \ref{prop:misspec_K}.

Guided by the above results, in applications where finite mixture models ($K_0 < \infty$) 
 have meaningful interpretations in capturing population heterogeneity, we suggest to examine 
 the GSF over a range of small to large values of $K$.
 This range may be chosen with consideration of the resulting number of mixture parameters, 
 with respect to the sample size $n$.  
 An elbow-shaped scatter 
 plot of $(K, \hat K)$ can shed light on a safe choice of the bound $K$ 
 and the selected order $\hat K$. 
We illustrate such a strategy through the real
data analysis in Section \ref{sec:data}. 
\begin{figure}[H]
\begin{center}
  \includegraphics[width=.8\linewidth]{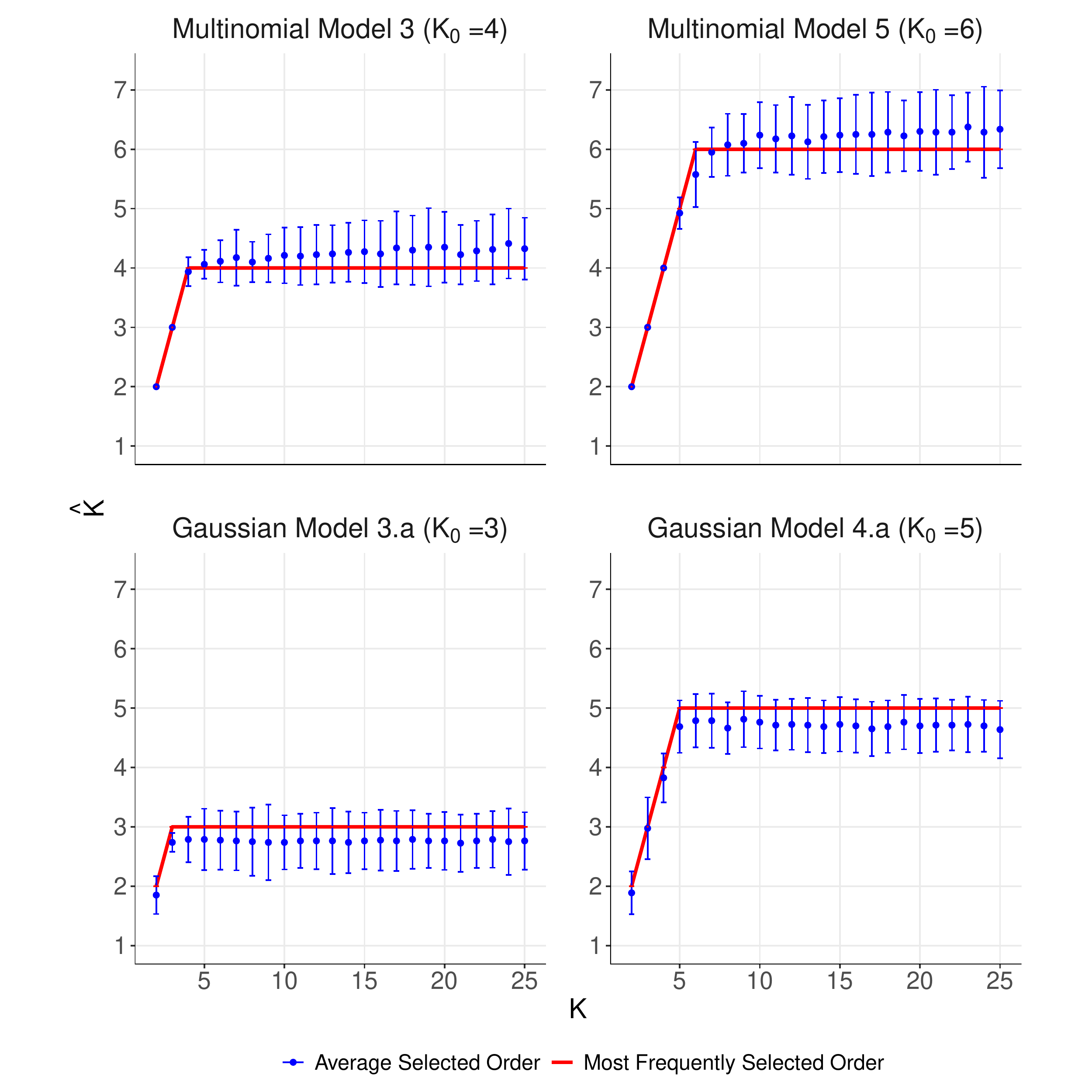}
  \caption{\label{sensitivity_simulation}
   Sensitivity analysis of the GSF with respect to the upper bound $K$.
   Error bars represent one standard deviation of the fitted
   order.} 
\label{fig:sensitivity_simulation}	
\end{center}
\end{figure}

\subsection{\bf Comparison of Merging-Based Methods}
\label{sec:sim-merging}

We now compare the GSF to alternate order
selection methods which are also 
based on merging the components of an overfitted mixture. 
Our simulations are based on location-Gaussian 
mixture models, though unlike Section \ref{sim-setting}, we now treat
the common covariance $\bSigma$ as known. 
In addition to the GSF, and to the AIC/BIC which are included  
as benchmarks, 
we consider the following two methods.
\begin{enumerate}
\item[--] The Merge-Truncate-Merge (MTM) procedure 
\citep{guha2019} described in Section 
\ref{remark}(IV), applied to posterior samples from a Dirichlet Process mixture (DPM).
\item[--] A hard-thresholding analogue of the GSF, denote
by GSF-Hard, which is obtained by first computing the estimator
$\widetilde G_n$ in \eqref{tildeG}, 
 and then merging the atoms of $\widetilde G_n$
which fall within a sufficiently small distance $\lambda > 0$
of each other
(see Algorithm \ref{alg:gsf-hard} in Supplement D.2 
for a precise description).
The GSF-Hard thus replaces the penalty $r_{\lambda}$ in the GSF 
with a post-hoc merging rule. 
By a straightforward simplification of our asymptotic
theory, the GSF-Hard estimator satisfies 
the same properties as $\hat G_n$ in Theorems
\ref{paramConsistency}--\ref{orderConsistency}.
 \end{enumerate} 
 
We fit the MTM procedure using the same algorithm and 
parameter settings as  described in Section 5 of \cite{guha2019}. 
The truncation and (second) merging stages of the MTM require a 
tuning parameter $c > 0$,
which plays a similar role as $\lambda$ in the GSF-Hard.
 The authors recommend considering various choices of $c$ in 
practice, though we are not aware of a method for tuning $c$. 
We therefore follow them by reporting
the performance
of the MTM for a range of 
$c$-values. 
For the GSF-Hard, we tune $\lambda$
using the BIC. 
Further implementation details are provided in 
Supplement~D.2.

We report the proportion of times that each method
selected the correct order
under Gaussian Models 1.b and 2.a
in Figure \ref{fig:merging_results}, 
based on $n=50, 100, 200, 400$.
More detailed results can  
be found in Supplement
E.4, including those for $n=600,800$.
For each sample size, we perform 80 replications
due to the computational burden associated with fitting
Dirichlet Process mixture models. 
The MTM results are based on the posterior mode.

The AIC, BIC, and GSF under all three penalties
exhibit improved performance under the current setting
with fixed $\bSigma$, compared to that
of Section \ref{sim-setting}. 
The GSF-Hard performs reasonably under
Model 1.b but  markedly
underperforms in Model 2.a.
Regarding the MTM, we  
report the results under four consecutive $c$-values
which were most favourable from a range 
of 16 candidate values. 
Under Model 1.b, the MTM under all four $c$-values
estimates $K_0$ most of the time, 
under most sample sizes, but underperforms compared to the remaining
methods. In contrast, under Model 2.a,
there exists a value of $c$ for which the MTM
remarkably estimates  $K_0$ on nearly all replications.
However, 
the sensitivity to $c$ is also seen to increase,
which can be problematic in the absence of a data-driven 
tuning procedure. Finally, we recall that 
the MTM is based on a nonparametric Bayes procedure, 
while the other methods are parametric and might
generally require smaller sample sizes to achieve reasonable accuracy.
\begin{figure}[t]
\begin{center}
  \includegraphics[width=.85\linewidth]{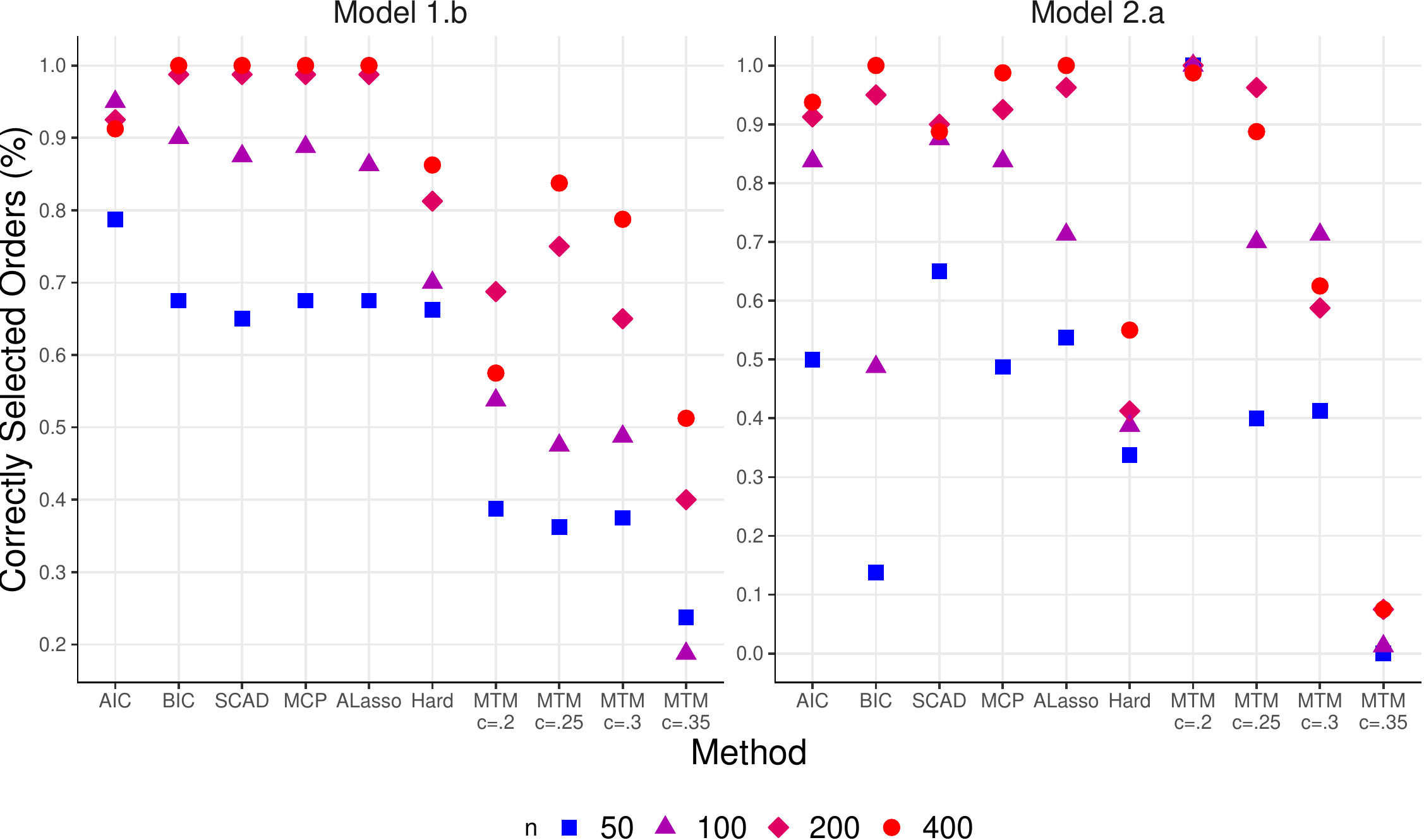}
  \caption{\label{fig:merging_results}
    Percentage of correctly selected
    orders for the multivariate Gaussian models
    with common and known covariance
    matrix.
    } 
 \end{center}
\end{figure}

We emphasize that MTM and GSF-Hard are both post-hoc
procedures for reducing the order of an 
overfitted mixing measure $G_n$, 
which is respectively equal to 
a sample from the DPM posterior, or to the estimator 
$\widetilde G_n$. 
This contrasts the GSF, which uses 
continuous penalties of 
the parameters to 
simultaneously perform order selection and mixing measure estimation,
and does not vary discretely with the tuning parameter $\lambda$. 
On the other hand, these two post-hoc procedures 
have the  practical
advantage of being computationally inexpensive
wrappers on top of the well-studied estimators $G_n$,
for which standard implementations are available. 
To illustrate this point, in Table 
\ref{tab:merging_time} we report the computational
time associated with the results from Figure \ref{fig:merging_results}, 
including also the sample sizes $n=600,800$.
It can be seen that GSF-Hard is typically computable with an
order of magnitude fewer seconds than the GSF under any of the three
penalties. The computational times for the MTM
are largely dominated by the time required to sample
the DPM posterior with the implementation we used---the
post-processing procedure itself accounts for a negligible
fraction of this time.

\begin{table}[H]\setlength\tabcolsep{2pt} 
\begin{tabular}{c || m{0.85cm}| m{1.2cm}m{1.2cm} m{1.2cm} | m{1cm} | m{1.2cm}||
 m{0.85cm}| m{1.2cm}m{1.2cm} m{1.2cm} | m{1cm} | m{1.2cm}m{0.95cm} m{0.95cm}m{0.95cm}m{0.95cm}m{0.95cm}m{0.95cm}m{0.95cm}m{0.95cm}m{0.95cm}m{0.95cm}}
\firsthline
  & \multicolumn{6}{c||}{Model 1.b}
 & \multicolumn{6}{c}{Model 2.a}\\
\hline
$n$ & AIC/ BIC  & GSF-SCAD & GSF-MCP & GSF-ALasso & GSF-Hard & MTM  
 & AIC/ BIC  & GSF-SCAD & GSF-MCP & GSF-ALasso & GSF-Hard & MTM\\
\hline
50  & 23.6 &  1.30  & 1.2  & 5.3  & 3.8 & 2830.0  &  21.1 & 2.6   & 1.8  & 5.3  & 3.9 & 2502.6 \\
\hline                                              
100 & 29.8 & 2.7  & 2.0  & 9.9   & 5.2 & 7148.2  &  25.6 & 6.3   & 3.8  & 9.9  & 5.4 & 5607.2 \\
\hline                                               
200 & 38.7 & 6.7  & 4.5  & 19.6  & 6.8 & 25428.3 &  34.9 & 17.2  & 8.6   & 19.6  & 7.0 & 21008.0 \\
\hline                                               
400 & 47.6  & 12.4 & 8.5  & 35.8 & 7.6 & 34911.9 &  45.8 & 43.5   & 16.4 & 35.8 & 9.0 & 20151.2 \\ 
\hline                                              
600 &  54.4 & 24.2 & 15.3 & 49.3 & 8.8 & 51131.0 &  51.3 & 57.8  & 21.8 & 49.3 & 10.0 & 37535.7 \\ 
\hline                                                       
800 & 60.0 & 32.2 & 22.8 & 67.1  & 9.9 & 74185.0 &  56.7 & 103.6 & 39.9 & 67.1  & 10.3 & 57469.7  \\ 
\end{tabular}%
\centering
\captionsetup{justification=centering}
\caption{\label{tab:merging_time} Average computational time (in seconds) 
per replication 
for the multivariate Gaussian models
with common and known covariance matrix. 
}
\end{table}

\section{Real Data Example}
\label{sec:data}
We consider the data analyzed by \cite{MOSIMANN1962}, arising from the study of 
the Bellas Artes pollen core from the Valley of Mexico, in view of reconstructing 
surrounding vegetation changes from the past. The data consists of $M=100$ counts on the 
frequency of occurrence of $d= 4$ kinds of fossil pollen grains, at $n=73$ different 
levels of a pollen core. A simple multinomial model provides a poor fit to this data, 
due to over-dispersion caused by clumped sampling. \cite{MOSIMANN1962} 
modelled this extra variation using a Dirichlet-multinomial distribution, and 
\cite{MOREL1993} fitted a 3-component multinomial mixture model. 

We applied the GSF-SCAD 
with upper bounds $K=2, \ldots, 25$. For each $K$, we 
fitted the GSF based on five 
different initial values for the modified EM algorithm, 
and selected the model with optimal tuning parameter
value. For $K=2$, the estimated order was $2$
and for $K\ge 3$, the most frequently selected order was $\hat K=3$. 
Given the similarity of the sample size
and dimension with those considered in the simulations, 
below we report the fitted model corresponding to the upper bound $K=12$.  
\begin{center}
\begin{figure}[H]
\centering
  \includegraphics[width=0.6\textwidth]{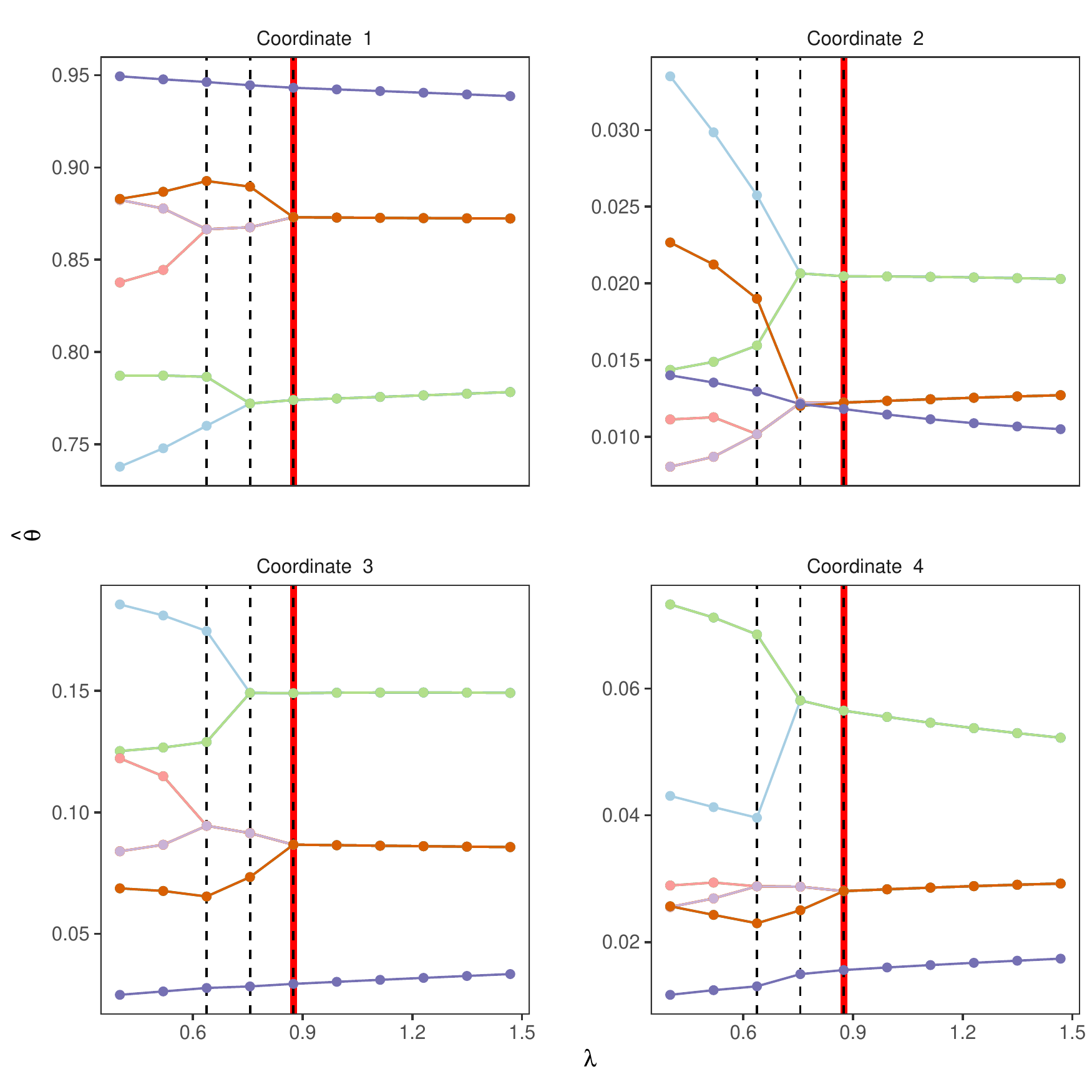}
  \caption{Coefficient plots for the GSF-SCAD on the pollen data. The vertical red lines indicate
  the selected tuning parameter.} 
  \label{fig:pollenCoef}
\end{figure}
\vspace{-.5in}
\end{center}

The models obtained by the GSF with the three 
penalties are similar---for instance, the fitted model obtained by the GSF-SCAD is
$$.15 ~ \text{Mult}(\hbtheta_1) +
  .25 ~ \text{Mult}(\hbtheta_2) +
  .60 ~ \text{Mult}(\hbtheta_3).$$
where $\text{Mult}(\btheta)$ denotes the multinomial distribution
with 100 trials and probabilities $\btheta$, 
$\hbtheta_1 = (.94, .01, .03, .02)^\top$, 
$\hbtheta_2 = (.77, .02, .15, .06)^\top$ and
$\hbtheta_3 = (.87, .01, .09, .03)^\top$.
The log-likelihood
value for this estimate is -499.87. The coefficient plots produced by the tuning parameter 
selector for GSF-SCAD are 
shown in Figure \ref{fig:pollenCoef}. 
Interestingly, the fitted order  
equals 3, for all $\lambda > 0.9$ in the range considered, coinciding with the final selected order,
and with the aforementioned sensitivity analysis
on $K$.

We also ran the AIC, BIC and ICL on this data. 
The AIC selected six components, while the BIC and ICL selected three components. 
The fitted model under the latter two methods is given by
$$
  .17 ~ \text{Mult}(\hbtheta_1) +
  .22 ~ \text{Mult}(\hbtheta_2) +
  .61 ~ \text{Mult}(\hbtheta_3).$$
where $\hbtheta_1 = (.95, .02, .03, .01)^\top$, 
$\hbtheta_2 = (.77, .02, .15, .07)^\top$ and
$\hbtheta_3 = (.87, .01, .09, .03)^\top$, with entries rounded to the nearest
hundredths.
The log-likelihood
value for this estimate is -496.39.
  
\section{Conclusion and Discussion}
\label{sec:discussion}
In this paper, we developed the
Group-Sort-Fuse (GSF) method for estimating the order of
finite mixture models with a multidimensional parameter space. 
By starting with 
a conservative upper bound $K$ on the mixture order, the GSF estimates the true 
order by applying two penalties to the overfitted log-likelihood, 
 which group and fuse redundant 
mixture components.
Under certain regularity conditions, 
the GSF is consistent in estimating the true order and it further 
provides a $\sqrt n$-consistent estimator for the true 
mixing measure (up to polylogarithmic factors). 
We examined its finite sample performance 
via thorough simulations, and illustrated its application to
 two real datasets, one of which is relegated to Supplement E.6.

We suggested the use of off-the-shelf methods, such as $v$-fold cross
validation or the BIC, for selecting the tuning parameter
$\lambda_n$ involved in the penalty $r_{\lambda_n}$. 
Properties of such choices with respect to our theoretical guidelines, 
or alternative methods specialized to the GSF, 
require further investigation.  

The methodology developed in this paper may be applicable
to mixtures which satisfy weaker notions of strong identifiability \citep{HO2016weak}.
Extending
our proof techniques to such models is, however, nontrivial. In particular, bounding the log-likelihood 
ratio statistic for the overfitted MLE $\bar G_n$ \citep{DACUNHA1999}, and the penalized log-likelihood ratio for the MPLE $\hat G_n$,
would require new insights in the absence of (second-order) strong identifiability.
Empirically, we illustrated in Section \ref{sim-setting} the promising finite sample performance
of the GSF under location-Gaussian mixtures with an unknown but common covariance matrix, which
themselves violate condition (SI). 

We have shown that the GSF achieves
a near-parametric rate of convergence
under the Wasserstein distance, 
but this rate only
holds pointwise in the true mixing measure
$G_0$. Our work leaves open the 
behaviour of the GSF when the true mixing
measure is permitted to vary with the sample
size $n$---indeed, the minimax
risk is known to scale at a rate markedly slower than 
parametric
\citep{HEINRICH2018,wu2020optimal}.

We established in
Proposition \ref{prop:misspec_K}
the asymptotic behaviour of the GSF
when the upper bound $K$ is 
underspecified.
However, our work provides no guarantees
when other aspects of the mixture model
$\calP_K=\{p_G: G \in \calG_K\}$ are misspecified,
such as the kernel density family
$\calF$. We note that the recent
work of \cite{guha2019} establishes
the asymptotic behaviour of various
Bayesian procedures under
such misspecification, in terms of
a suitable Kullback-Leibler
projection of the true mixture distribution.
While we expect the GSF to obey similar
asymptotics, we are not aware of a general
theory for maximum likelihood
estimation under 
misspecification in non-convex models
such as $\calP_K$.
We leave a careful investigation of such properties to future work.

We believe that the framework
developed in this paper
paves the way to a new class of methods for order selection problems 
in other latent-variable models, such as
mixture of regressions and Markov-switching autoregressive models
\citep{FRUE2006}. 
Results of the type developed by \cite{DACUNHA1999} in understanding 
large sample behaviour of likelihood ratio statistics for these models, and the
recent work of \cite{HO2019} in characterizing
rates of convergence for parameter estimation in over-specified Gaussian mixtures 
of experts, may provide first steps toward such extensions. 
We also mention applications of the GSF procedure to non-model-based clustering 
methods, such as the $K$-means algorithm.
While the notion of order, or true number of clusters, is generally elusive in the absence 
of a model,  extensions of the GSF may provide a natural heuristic for choosing the 
number of clusters in such methods.

{\bf Acknowledgements.}  
We would like to thank the editor, an associate editor, and two referees for 
their insightful comments and suggestions which significantly improved the quality of this paper.
We thank Jiahua Chen for discussions
related to the proof of Proposition~\ref{binomialSI}, 
Russell Steele for bringing to our attention the multinomial dataset  
analyzed in Section~\ref{sec:data},
and Aritra Guha for sharing an  
 implementation of the Merge-Truncate-Merge procedure. 
We also thank Sivaraman Balakrishnan and Larry Wasserman for useful discussions.
Tudor Manole was supported by the Natural Sciences and Engineering Research
Council of Canada  
and also by the Fonds de
recherche du Qu\'ebec--Nature et technologies. 
Abbas Khalili was supported by the Natural Sciences and Engineering Research Council
of Canada through Discovery Grant (\textsc{nserc rgpin}-2015-03805 and \textsc{nserc rgpin}-2020-05011). 
 
\clearpage 
 
{\noindent \LARGE \bf Supplementary Material}
\renewcommand{\theequation}{S.\arabic{equation}}
 
This Supplementary Material contains six sections.  
Supplement 
\hyperref[section:regularityConditions]{A} contains notation
which will be 
used throughout the sequel. 
Supplement \hyperref[section:otherPapers]{B}
states several results from other papers which are needed for our subsequent proofs.
Supplement \hyperref[section:regularityConditions]{C}
contains all proofs of the results stated in the paper, and includes the statements
and proofs of several auxiliary results.
Supplement \hyperref[section:regularityConditions]{D} 
outlines our numerical solution, and Supplement \hyperref[section:regularityConditions]{E} 
reports several
figures and tables cited in the paper. Finally, Supplement 
\hyperref[section:regularityConditions]{F} 
reports the implementation
and complete numerical results of the simulation in 
Figure \ref{fig:comparisonFigure} of the paper.

\section*{Supplement A: Notation}
Recall that $\calF = \{ f(\by; \btheta) : \btheta= (\theta_1, \theta_2, \ldots, \theta_d)^{\top} 
\in \Theta \subseteq \bbR^d, \ \by \in \mathcal{Y} \subseteq \bbR^N \}$
is a parametric density family with respect to a $\sigma$-finite measure $\nu$. Let 
\begin{equation}
\label{famden}
\calP_K = \left\{p_G(\by) = \int_{\Theta} f(\by; \btheta)dG(\btheta): G \in \calG_K\right\},
\end{equation}
where, recall, that $\calG_K$ is the set of finite mixing measures with order at most $K \geq K_0$.
Let $p_0 = p_{G_0}$ be the density of  the true finite mixture model with its 
corresponding probability distribution $P_0$. Let $\hat p_n = p_{\hat G_n}$ be the estimated mixture density
based on the MPLE $\hat G_n$, and
define the empirical measure $P_n = \frac 1 n \sum_{i=1}^n \delta_{\bY_i}$.

For any $p_G \in \calP_K$, let $\bar p_G = \frac{p_G + p_0}{2}$, and $\bar \calP_K^{\frac 1 2} = 
\left\{\bar p_G^{\frac 1 2} : p_G \in \calP_K\right\}$. For any $\delta > 0$, recall that
$$\bar\calP_K^{\frac 1 2}(\delta) = \left\{\bar p_G^{\frac 1 2} \in \bar\calP_K^{\frac 1 2}: h(\bar p_G, p_0) \leq \delta\right\}.$$
Furthermore, define the empirical process
\begin{equation}
\label{nuDef}
\nu_n(G) = \sqrt n\int_{\{p_0 > 0\}} \frac 1 2 \log\bigg\{\frac{p_G + p_0}{2p_0}\bigg \} d(P_n - P_0), \quad G \in \calG_K.
\end{equation}
We also define the following two collections of mixing measures, 
for some $0 < b_0 < 1$,
\begin{align}
\label{tildeGSet}
\calG_K(b_0) &= \left\{G \in \calG_K \setminus \calG_{K_0-1}: G = \sum_{j=1}^K \pi_j \delta_{\btheta_j}, \ \pi_j \geq b_0\right\}, \\
\calG_K(b_0; \gamma) &= \left\{G \in \calG_K(b_0): h(p_G, p_{G_0}) \leq \gamma\right\}, \quad \forall \gamma > 0.
\end{align}
Also, 
let 
\begin{equation}
\label{KL-div}
\KL(p,q) = \int \log\left(\frac{p}{q}\right)~ p ~d\nu
\end{equation}
denote
the Kullback-Leibler divergence between any two densities $p$ and $q$ 
dominated by the measure~$\nu$.

In the proofs of our main results, we will frequently work with differences of the form 
$L_n(G) - L_n(G_0)$, for $G \in \calG_K$. We therefore introduce the following constructions.
Given a generic mixing measure $G = \sum_{j=1}^K \pi_j \delta_{\btheta_j} \in \calG_K$ and
the true mixing measure 
$G_0 = \sum_{k=1}^{K_0} \pi_{0k} \delta_{\btheta_{0k}}$, 
define $\btheta = (\btheta_1, \dots, \btheta_K)$ and $\btheta_0 = (\btheta_{01}, \dots, \btheta_{0K_0})$,
and let $\bpi = (\pi_1, \dots, \pi_K)^\top$ and $\bpi_0 = (\pi_{01}, \dots, \pi_{0K_0})^\top$. 
For simplicity in notation, in what follows we write $\varphi(\pi_1, \ldots, \pi_K) = \varphi(\bpi)$.

Define the following difference between the first penalty functions
\begin{equation}
\label{zetaDef}
\zeta_n(G) = \frac 1 n \left\{\varphi(\bpi_0) - \varphi(\bpi) \right\}, \quad \forall G \in \calG_K.
\end{equation}
Furthermore, 
recall that $\alpha_{\bt}$ is a cluster ordering, and let $\alpha = \alpha_{\btheta}$ and $\alpha_0 = \alpha_{\btheta_0}$.
Recall that $\bfeta_j = \btheta_{\alpha(j+1)} - \btheta_{\alpha(j)}$, $j=1, \dots, K-1$,
and let $\bfeta_{0k} = \btheta_{0\alpha_0(k+1)} - \btheta_{0\alpha_0(k)}$, $k=1, \dots, K_0-1$. 
Likewise, given a mixing measure  $\wtG = \sum_{j=1}^K \tilde\pi_j \delta_{\tbtheta_j} \in \calG_K$, 
let $\tbtheta=(\tbtheta_1, \dots, \tbtheta_K)$. Define
$\tbfeta_j = \tbtheta_{\talpha(j+1)} - \tbtheta_{\talpha(j)}$ for all $j=1, \dots, K-1$, where
$\talpha = \alpha_{\tbtheta}$.
Let $u, v \in S_{K-1}$ be the permutations such that
$$\norm{\bfeta_{u(1)}} \geq \dots \geq \norm{\bfeta_{u(K-1)}}, \quad
  \norm{\tbfeta_{v(1)}} \geq \dots \geq \norm{\tbfeta_{v(K-1)}},$$
and similarly, let $u_0 \in S_{K_0-1}$ be such that
\begin{equation}
\label{omegaDef}
\norm{\bfeta_{0u_0(1)}} \geq \dots \geq \norm{\bfeta_{0u_0(K_0-1)}}. 
\end{equation}
Let $\psi = v \circ u^{-1}$ and $\psi_0 = v \circ u_0^{-1}$. 
Then, as in Section \ref{sec:asymptotic} of the paper, we define the weights
\begin{equation}
\label{weightDefn}
\omega_j \equiv \omega_j(\btheta, \tbtheta) = \norm{\tbfeta_{\psi(j)}}^{-\beta}, \quad 
  \omega_{0k} \equiv \omega_{0k}(\btheta_0, \tbtheta) = \norm{\tbfeta_{\psi_0(k)}}^{-\beta},
\end{equation}
 for some $\beta > 1$, and for all $j=1, \dots, K-1, \ k=1, \dots, K_0-1$.
We then set
\begin{equation}
\label{xiDef}
\xi_n(G; \wtG) = \sum_{k=1}^{K_0-1} r_{\lambda_n}(\norm{\bfeta_{0k}}; \omega_{0k}) - 
\sum_{j=1}^{K-1} r_{\lambda_n}(\norm{\bfeta_j}; \omega_j).
\end{equation}
It is worth noting that $\xi_n(G; \wtG)$ is well-defined
due to Property (i) in Definition \ref{cluster-order} of cluster orderings.
Finally, throughout the sequel, we let  
\begin{equation}
\label{two-modified}
\wtG_n =  
\argmax_{G \in \calG_K} \left\{l_n(G) - \phi(\bpi)\right\}.
\end{equation}
With this notation, the penalized log-likelihood difference $L_n(G) - L_n(G_0)$
may be written as follows for the choice of weights described in Section \ref{sec:asymptotic}
of the paper, 
$$
L_n(G) - L_n(G_0)
 = \left\{ l_n(G) - l_n(G_0)\right\} + n\zeta_n(G) + n\xi_n(G; \wtG_n),
 $$
for any $G \in \calG_K$.
 
Finally, for any matrix $\bM = (m_{ij})_{1\le i \le d_1,  1 \le j \le d_2}$,  
we write the Frobenius norm as $\norm{\bM}_{\text F} = \left( \sum_{i=1}^{d_1} \sum_{j=1}^{d_2} m_{ij}^2 \right)^{\frac 1 2}$. For any real symmetric matrix $\bM$,  $\varrho_{\text{min}}(\bM)$ and
$\varrho_{\text{max}}(\bM)$ denote its respective minimum and maximum eigenvalues.

\section*{Supplement B: Results from Other Papers}
\label{section:otherPapers}

In this section, we state several existing
results which are needed to prove our main Theorems
\ref{paramConsistency}-\ref{orderConsistency}. We begin with a simplified
statement of Theorem 
3.2 of \cite{DACUNHA1999}), which describes the behaviour
of the likelihood ratio statistic of strongly identifiable mixture models.

\begin{customthm}{B.1}[\cite{DACUNHA1999}] 
\label{thm:lrs}
Under conditions (SI) and (A3),  
the log-likelihood ratio statistic 
over the class $\calG_K$ satisfies
\[
\sup_{G \in \calG_K}l_n(G) - l_n(G_0) = O_p(1).
\]
\end{customthm}
Next, we summarize two results of \cite{HO2016strong}, 
relating the Wasserstein distance between two mixing measures to the 
Hellinger distance between their corresponding mixture densities.  
We note that these results were originally
proven in the special case where the dominating measure $\nu$ of the 
parametric family $\calF$
is the Lebesgue measure. A careful verification of Ho and Nguyen's
proof technique readily shows that $\nu$ can be 
any $\sigma$-finite measure.
The assumptions made in our statement below are stronger than necessary for part (i), 
but kept for convenience. 

\begin{customthm}{B.2}[\cite{HO2016strong}]
\label{thm:ho_long_results}
Suppose that $\calF$ satisfies conditions (SI) and (A2)  
Then, there exist $\delta_0, c_0 > 0$ depending only on $G_0, \Theta$ and $\calF$ such that
the following two statements hold.
\begin{enumerate}
\item[(i)] For all mixing measures $G$ with exactly $K_0$ atoms 
satisfying $W_1(G, G_0) < \delta_0$, we have $h(p_G, p_0) \geq c_0 W_1(G, G_0)$.
\item[(ii)] For all mixing measures $G \in \calG_{K}$ satisfying 
$W_2(G, G_0) < \delta_0$, we have $h(p_G, p_0) \geq c_0 W_2^2(G, G_0)$.
\end{enumerate}
\end{customthm}
The following result \citep[Lemma 3.1]{HO2016singularity} relates the 
convergence of a mixing measure in Wasserstein distance to
the convergence of its atoms and mixing proportions.
\begin{customlem}{B.3}[\cite{HO2016singularity}]
\label{lem:ho}
For any mixing measure $G = \sum_{j=1}^K \pi_k\delta_{\btheta_j} \in \calG_K(b_0)$, 
for some $b_0 > 0$, let
$\calI_k = \left\{j: \norm{\btheta_j - \btheta_{0k}} \leq \norm{\btheta_j - \btheta_{0l}}, \ \forall l \neq k\right\}$, for 
all $k=1, \dots, K_0$.
Then, for any $r \geq 1$,
$$W_r^r(G, G_0) \asymp \sum_{k=1}^{K_0} \sum_{j \in \calI_k} \pi_j \norm{\btheta_j - \btheta_{0k}}^r +
                       \sum_{k=1}^{K_0} \Bigg| \pi_{0k} - \sum_{j \in \calI_k} \pi_j \Bigg|,$$
as $W_r(G, G_0) \downarrow 0$.
\end{customlem}                       
The following theorem from empirical process theory is a special case of Theorem 5.11 
from \cite{GEER2000}, and will be invoked in the proof of Theorem \ref{densityConsistency}.
\begin{customthm}{B.4}[\cite{GEER2000}]
\label{thm:incrementGeer}
Let $R > 0$ be given and let
\begin{equation}
\label{calJDef}
N(R) = \left\{G \in \calG_K: h(\bar p_G, p_0) \leq R\right\},
\end{equation}
where $\bar p_G = \frac{p_G + p_0}{2}$.
Given a universal constant $C > 0$, let $a, C_1 > 0$ be chosen such that
\begin{equation}
\label{thm511_2}
a \leq C_1\sqrt n R^2 \wedge 8\sqrt n R,
\end{equation}
and,
\begin{equation}
\label{thm511_5}
a \geq\sqrt{C^2(C_1+1)} \left(\int_0^R \sqrt{H_B 
\left(\frac u {\sqrt 2}, \left\{p_G:G \in N(R)\right\}, \nu \right)} du \vee  R\right) , 
\end{equation}
Then, 
\[
\bbP\left\{\sup_{G \in N(R)} |\nu_n(G)| \geq a\right\}
 \leq C \exp\left(-\frac{a^2}{C^2(C_1+1)R^2}\right),
 \]
where $\nu_n(G)$ is defined in equation \eqref{nuDef}.
\end{customthm}

The following result shows the behavior of the likelihood ratio statistic 
for underfitted finite mixture models, and is used in the proof of 
Proposition \ref{prop:misspec_K} .

\begin{customthm}{B.5}
\label{thm:KER2000LER1992}
Let $K < K_0$. Suppose that $\calF$ satisfies condition (A3), and that $\calP_K$ is identifiable. 
\begin{enumerate}
\item[(i)]\citep{LEROUX1992}
For all $1 \leq k \leq K$, there exists a mixing measure $G^*_k \in \calG_k$ for which 
$\KL(p_{G^*_k}, p_{G_0})  = \inf_{G \in \calG_k} \KL(p_G, p_{G_0})$,
where $\KL$ is the Kullback-Leibler divergence in \eqref{KL-div}.
Furthermore, 
\begin{equation}
\label{eq:leroux}
\KL(p_{G^*_k}, p_{G_0}) > \KL(p_{G^*_{k+1}}, p_{G_0}).
\end{equation}
\item[(ii)]\citep{KERIBIN2000}
Consider the MLE 
\(
\bar G_n^{(k)} = \argmax_{G \in \calG_k} l_n(G).
\) 
For all $k =1, \dots, K$, as~$n~\to~\infty$,    
\begin{equation}
\label{eq:keribin}
\frac 1 n \Big\{ l_n(\bar G_n^{(k)}) - l_n(G_0)\Big\} 
\overset{a.s.}{\longrightarrow}
-\KL(p_{G^*_k}, p_{G_0}) < 0.
\end{equation}
\end{enumerate}
\end{customthm}

\section*{Supplement C: Proofs}
\label{section:proofs}
\subsection*{C.1. Proof of Theorem \ref{paramConsistency}}
We begin with the following Lemma, which generalizes Lemma 4.1 of \cite{GEER2000}.
\noindent 
\begin{lemma}
\label{basicInequalities}
The MPLE $\hat G_n$ satisfies
\begin{equation}
\label{basicInequality} 
h^2\left(\frac{\hp_n + p_0} 2, p_0\right) - 
\frac 1 4 \left[\zeta_n(\hat G_n) 
+
 \xi_n(\hat G_n; \wtG_n)\right] \leq \frac 1 {\sqrt n}\nu_n(\hat G_n), 
\end{equation}
for all $n\geq 1$, where $\wtG_n$ is given in \eqref{two-modified}. 
\end{lemma}
\begin{proof}
By concavity of the log function, we have
\begin{equation}
\label{logConcave}
\log \frac{\hp_n + p_0}{2p_0}I\left\{p_0 > 0\right\} \geq \frac 1 2\log\frac{\hp_n}{p_0}I\{p_0 > 0\}.
\end{equation}
Now, note that
\begin{align*}
0 \leq \frac 1 n \left\{L_n(\hat G_n) - L_n(G_0)\right\}
 &= \int \log \frac{\hp_n}{p_0} dP_n + \zeta_n(\hat G_n) + \xi_n(\hat G_n;\wtG_n). 
\end{align*}
Thus, by \eqref{logConcave},
\begin{align*}
-\frac 1 4 &\left[\xi_n(\hat G_n;\widetilde G_n) + \zeta_n(\hat G_n)\right] \\
 &\leq \int_{\{p_0 > 0\}} \frac 1 4 \log \frac{\hp_n}{p_0} dP_n \\
 &\leq \int_{\{p_0 > 0\}} \frac 1 2 \log \frac{\hp_n + p_0}{2p_0} d(P_n-P_0) + 
       \int_{\{p_0 > 0\}} \frac 1 2 \log \frac{\hp_n + p_0}{2p_0} dP_0 \\
 &= \int_{\{p_0 > 0\}} \frac 1 2 \log \frac{\hp_n + p_0}{2p_0} d(P_n-P_0) - 
    \frac 1 2 \text{KL}\left(\frac{\hp_n + p_0}{2}, p_0 \right) \\
 &\leq  \frac 1 {\sqrt n}\nu_n(\hat G_n) - 
    h^2\left(\frac{\hp_n + p_0}{2}, p_0 \right),
\end{align*}
where  we have used the well-known inequality $h^2(q, q') \leq \frac 1 2 \text{KL}(q, q')$, 
for any densities $q$ and $q'$
with respect to the measure $\nu$. 
The claim follows. 
\end{proof}

As noted in \cite{GEER2000}, for all $G \in \calG_K$ we have
\begin{equation}
\label{simple_hellinger_bounds}
h^2(\bar p_G, p_0) \leq \frac 1 2 h^2(p_G, p_0), \quad h^2(p_G, p_0) \leq 16 h^2(\bar p_G, p_0).
\end{equation}
Combining the second of these inequalities with Lemma \ref{basicInequalities} immediately
yields an upper bound on $h(\hat p_n, p_0)$. This fact combined with the local relationship
$W_2^2 \lesssim h$ in Theorem \ref{thm:ho_long_results} leads
to the proof of Theorem \ref{paramConsistency}.
\begin{proof}[Proof (Of Theorem \ref{paramConsistency})]
We begin with Part (i). A combination of \eqref{simple_hellinger_bounds} 
and Lemma~\ref{basicInequalities} yields
\begin{align*}
h^2(\hat p_n, p_0)
 &\lesssim h^2\left(\frac{\hat p_n + p_0}{2}, p_0\right) \\
 &\leq \frac 1 4 \left[ \zeta_n(\hat G_n) + \xi_n(\hat G_n; \wtG_n)\right] + \frac 1 {\sqrt n}\nu_n(\hat G_n) \\
 &\leq \frac {\varphi(\bpi_0)} {4n} +  
        \frac 1 4 \sum_{k=1}^{K_0- 1} r_{\lambda_n}(\norm{\bfeta_{0k}}; \omega_{0k})  +
        \sup_{G \in \calG_K} \frac 1 {\sqrt n}|\nu_n(G)|.
\end{align*}
Since the elements of $\bpi_0$ are bounded away from zero, 
$\frac{\varphi(\bpi_0)}{4n} = o(1)$ by condition (F).
Furthermore, under assumption (P1) on $r_{\lambda_n}$, and using the fact
that $\wtG_{n}$ is consistent under $W_2$, and hence has at least $K_0$ atoms as $n \to \infty$ almost surely, 
we have $r_{\lambda_n}(\norm{\bfeta_{0k}}; \omega_{0k}) \overset{a.s.}{\longrightarrow} 0 $ for all $k=1, \dots, K_0-1$. 
Finally, assumption (A1)
implies that 
$\sup_{G \in \calG_K} \frac 1 {\sqrt n}|\nu_n(G)| \overset{a.s.}{\longrightarrow} 0.$
We deduce that $h(\hat p_n, p_0) \overset{a.s.}{\longrightarrow} 0.$
 
Furthermore, for any $r \geq 1$, using the interpolation equations (7.3) and (7.4) of \cite{VILLANI2003},
and Part (ii) of Theorem \ref{thm:ho_long_results} above, we have
\begin{align*}
W_r^r(\hat G_n, G_0) 
 \leq \Big(\text{diam}^{r-2}(\Theta) \vee 1\Big) W_2^2(\hat G_n, G_0) \lesssim h(\hat p_n, p_0) \overset{a.s.}{\longrightarrow} 0,
\end{align*}
due to the compactness assumption on $\Theta$. The result follows.

We now turn to Part (ii).
As a result of Part (i), the MPLE $\hat{G}_n$ has at least as many atoms as $G_0$ 
with probability tending to one. 
This implies that for every $k=1, 2, \dots, K_0$, there exists an index $1 \le j \le K$ such that 
$\big\lVert\hat \btheta_j - \btheta_{0k} \big\rVert \overset{p}{\longrightarrow} 0$, as $n \to \infty$.  
Therefore, since $\alpha$ is a bijection, there exists a set $S \subseteq \{1, \dots, K\}$
with cardinality at least $K_0-1$
such that for all $j \in S$, $\norm{\hbfeta_j}\geq \delta_0
= (1/2)\min_{1 \leq j < k \leq K_0} \norm{\btheta_{0j}-\btheta_{0k}}$, with probability tending to one.  
By compactness of $\Theta$, there must therefore exist $D > 0$ such that
$\norm{\hbfeta_k} \in [\delta_0, D]$ for all $k \in S$ in probability. 
Furthermore, notice that part (i) of this result also holds for the mixing measure
$\widetilde G_n$, thus it is also the case that for at least $K_0-1$ indices $k \in \{1, \dots, K\}$,  
$\norm{\tilde \bfeta_k}\geq \delta_0$ with probability tending to one.  
Due the definition of $\psi$ in the construction of weights $\omega_j$, we deduce that 
$\omega_k  \in [D^{-1}, \delta_0^{-1}]$ for all $k \in S$, for large $n$ in probability. 
Thus, by condition (P1), since $r_{\lambda_n} \ge 0$, and $\bar{G}_n$ is the MLE of $G$ over $\calG_K$, 
we have with probability tending to one,
\begin{align*}
0 
 &\leq L_n(\hG_n) - L_n(G_0)  \\
 &= \left\{ l_n(\hG_n) - l_n(G_0) \right\}-
	 \left\{ \varphi(\hat\bpi) - \varphi(\bpi_0) \right\}
  + n \left\{ \sum_{k=1}^{K_0-1} r_{\lambda_n} (\norm{\bfeta_{0k}}; \omega_{0k}) -
	   \sum_{j=1}^{K-1} r_{\lambda_n}(\norm{\hbfeta_j};\omega_j) \right\}  \\
 &\leq \left\{ l_n(\hG_n) - l_n(G_0) \right\}-
	 \left\{ \varphi(\hat\bpi) - \varphi(\bpi_0) \right\}
  + n \left\{ \sum_{k=1}^{K_0-1} r_{\lambda_n} (\norm{\bfeta_{0k}}; \omega_{0k}) -
	   \sum_{j\in S} r_{\lambda_n}(\norm{\hbfeta_j};\omega_j) \right\}  \\
 &\leq \left\{ l_n(\hG_n) - l_n(G_0) \right\}-
	 \left\{ \varphi(\hat\bpi) - \varphi(\bpi_0) \right\}
  + n (K_0-1) \text{diam}\big(r_{\lambda_n} ([\delta_0, D]; [D^{-1}, \delta_0^{-1}])\big) \\
 &\leq \left\{ l_n(\bar{G}_n) - l_n(G_0) \right\}- a_n \phi(\hat\bpi)  + O(a_n).
\end{align*}	
Under condition (SI) and regularity condition (A3)  
it now follows from Theorem \ref{thm:lrs} that
$$
0 \leq \phi(\hat \bpi) \leq \frac 1 {a_n} \left\{ l_n(\bar{G}_n) - l_n(G_0) \right\} +  O(1) = O_p(a_n^{-1}) + O(1) = O_p(1),$$
where we used condition (F) to ensure that $a_n \not\to 0$.
By definition of $\phi$, the estimated mixing proportions $\hat \pi_j$
are thus strictly positive in probability, as $n \to \infty$. 
It must then follow from Lemma \ref{lem:ho} that, for all 
$k= 1, 2, \ldots, K_0$, $\sum_{j \in \calI_k} \hpi_j = \pi_{0k} + o_p(1)$ up to relabelling, 
and for every $l=1, \dots, K$, there exists $k= 1, \ldots, K_0$ such that 
$\big\lVert \hbtheta_l - \btheta_{0k} \big\rVert \overset{p}{\longrightarrow} 0$, 
or equivalently
\( \max_{j \in \calI_k}\big\lVert \hbtheta_j - \btheta_{0k} \big\rVert \overset{p}{\longrightarrow} 0\)
, as $n \to \infty$. 
\end{proof}

\subsection*{C.2. Proof of Theorem \ref{densityConsistency}}
Inspired by \cite{GEER2000}, our starting point for proving
Theorem \ref{densityConsistency} is the Basic Inequality in Lemma \ref{basicInequalities}. 
To make use of this inequality, we must control the penalty
differences $\zeta_n(G)$ and $\xi_n(G; \wtG)$ for all $(G; \wtG)$ 
in an appropriate neighborhood of $G_0$. We do so by first establishing a rate of 
convergence of the estimator $\wtG_n$. 
In what follows, we write $\tilde p_n = p_{\wtG_n}$.
\begin{lemma}
\label{rate_tildeG}
For a universal constant $J > 0$, assume there exists a sequence 
of real numbers $\gamma_n \gtrsim(\log n / n)^{1/2}$
such that for all $\gamma \geq \gamma_n$,
$$\calJ_B\left(\gamma, \bar\calP_K^{\frac 1 2}(\gamma), \nu \right) \leq J \sqrt n \gamma^2.$$
Then 
$h(\tilde p_n, p_0)= O_p(\gamma_n).$    
In particular, it follows that $W_2(\wtG_n, G_0) = O_p(\gamma_n^{\frac 1 2}).$
\end{lemma}

\begin{proof}
The proof follows by the same argument as that of Theorem 7.4 in \cite{GEER2000}. 
In view of Lemma \ref{basicInequalities} with $\lambda_n = 0$, we have
\begin{align*}
\bbP\left\{h(\tilde p_n, p_0) > \gamma_n\right\} 
 &\leq 
\bbP\left\{
\sup_{\substack{G \in \calG_K \\ h(\bar p_G, p_0) > \gamma_n/4}} n^{-\frac 1 2} \nu_n(G) + \frac 1 {4n}[\phi(\bpi_0) - \phi(\bpi)] - h^2(\bar p_G, p_0)\geq 0 \right\} \\
 &\leq  
\bbP\left\{
\sup_{\substack{G \in \calG_K \\ h(\bar p_G, p_0) > \gamma_n/4}} n^{-\frac 1 2} \nu_n(G) + \frac H n - h^2(\bar p_G, p_0)\geq 0 \right\},
\end{align*}
where $H := \phi(\bpi_0)/4$.
Let $\calS = \min\{s: 2^{s+1}\gamma_n/4 > 1\}$. We have
\begin{align*}
\bbP&\left\{\sup_{\substack{G \in \calG_K \\ h(\bar p_G, p_0) > \gamma_n/4}} n^{-\frac 1 2} \nu_n(G) + \frac H n  - h^2(\bar p_G, p_0)\geq 0 \right\} \\
 &\leq \sum_{s=0}^{\calS} \bbP\left\{\sup_{\substack{G \in \calG_K \\ h(\bar p_G, p_0) \leq (2^{s+1})\gamma_n/4}} \nu_n(G)  \geq \sqrt n 2^{2s}\left(\frac{\gamma_n}{4}\right)^2 - \frac H {\sqrt n}   \right\}.
\end{align*}
We may now invoke Theorem \ref{thm:incrementGeer}. Let $R=2^{s+1}\gamma_n$, $C_1=15$, and
$$a = \sqrt n 2^{2s}\left(\frac{\gamma_n}{4}\right)^2 - \frac H {\sqrt n}.$$
To show that condition \eqref{thm511_5} holds, note that
\begin{align*}
4C &\left(\int_0^{2^{s+1}\gamma_n} \sqrt{H_B\left(\frac u {\sqrt 2}, \bar\calP_K^{\frac 1 2}\left(2^{s+1}\frac{\gamma_n}{4}\right), \nu\right)} du \vee 2^{s+1}\gamma_n\right) \\
 &\leq 4C \left(\sqrt 2 \int_0^{2^{s+\frac 1 2}\gamma_n} \sqrt{H_B\left(u, \bar\calP_K^{\frac 1 2}\left(2^{s+\frac 1 2}\gamma_n\right), \nu\right)} du \vee 2^{s+1}\gamma_n\right) \\
 &\leq  4C \left(J \sqrt 2 \sqrt n 2^{2s+\frac 1 2} \gamma_n^2 \vee 2^{s+1}\gamma_n\right) \\
 &=  4C \left(J \sqrt n 2^{2s+ 1} \gamma_n^2 \vee 2^{s+1}\gamma_n\right).
\end{align*} 
There exists $N > 0$ depending on $H$ (and hence on $G_0$)
such that the above quantity is bounded above by $a$ for all $n \geq N$, for a
universal constant $J > 0$. Invoking Theorem 1, we therefore have
\begin{align*}
\sum_{s=0}^{\calS} \bbP&\left\{\sup_{\substack{G \in \calG_K \\h(\bar p_G, p_0) \leq (2^{s+1})\gamma_n/4}} \nu_n(G) + \frac H{\sqrt n} \geq \sqrt n 2^{2s}\left(\frac{\gamma_n}{4}\right)^2 \right\} \\
 &\leq C\sum_{s=0}^{\infty} \exp\left\{-\frac 1 {16C^2 2^{2s+2}\gamma_n^2}\left[\sqrt n 2^{2s}\left(\frac {\gamma_n} 4\right)^2 - \frac H {\sqrt n}\right]^2\right\} \\
 &\leq C\sum_{s=0}^{\infty} \exp\left\{-\frac 1 {16C^2 2^{2s+2}\gamma_n^2}\left[ 
      \frac{n 2^{4s} \gamma_n^4}{(16)^2} - \frac{2^{2s+1}\gamma_n^2 H}{16}\right] \right\} \\ 
 &= C\exp\left\{\frac H {2^9 C^2}\right\} \sum_{s=0}^{\infty} \exp\left\{-\frac {n2^{2s-2}\gamma_n^2} {(16)^3 C^2}
 \right\} \\ 
 &=o(1).
\end{align*}
The claim of the first part follows. The second part follows by Theorem \ref{thm:ho_long_results}.
\end{proof}

In view of Theorem \ref{paramConsistency} and Lemma \ref{rate_tildeG}, for every $\epsilon \in (0,1)$, 
there exists $b_0 > 0$ 
such that $\wtG_n \in \calG_K(b_0; \gamma_n)$ for large enough $n$, with probability at least $1-\epsilon$.
This fact, combined with the following key proposition, 
will lead to the proof of Theorem \ref{densityConsistency}.
\begin{customprop}{C.1}
\label{xiBound}
Let $\kappa_n \geq \gamma_n \gtrsim (\log n/n)^{1/2}$. Let $0 < b_0 < \min_{1 \leq k \leq K_0} \pi_{0k}$. 
Under penalty conditions (P1) and (P2), there exists constants $c,M > 0$ 
depending on $G_0$ such that, if $\kappa_n \leq M$, then
\begin{align*}\sup & \bigg\{\xi_n(G; \wtG): 
                             G \in \calG_K(b_0; \kappa_n), \
                           \wtG \in \calG_K(b_0; \gamma_n)\bigg\}
\leq c \gamma_n^{\frac 3 2}\left(\kappa_n^{\frac 1 2} + \gamma_n^{\frac 1 2} \right) / \log n,
 \end{align*} 
and, \ \ ~$ \displaystyle \sup \bigg\{\zeta_n(G) :   G \in \calG_K(b_0; \kappa_n)\bigg\} 
\leq \frac{c\gamma_n^{\frac 3 2}\kappa_n}{ \log n}.$
 \end{customprop}
\begin{proof}
We prove the  first claim in six steps.

\noindent
\textbf{Step 0: Setup.} Let $G \in \calG_K(b_0; \kappa_n)$ and $\wtG \in \calG_K(b_0; \gamma_n)$.
The dependence of $G$ and $\wtG$ on $n$ is omitted from the notation for simplicity. 
It will suffice to prove that there exist $c, M > 0$
which do not depend on $G$ and $\wtG$, such that if $\kappa_n, \gamma_n \leq M$, then
$$ \xi_n(G; \wtG) \leq 
    c \gamma_n^{\frac 3 2}\left(\gamma_n^{\frac 1 2} + \kappa_n^{\frac 1 2}\right)/\log n .$$
Writing $G = \sum_{j=1}^K \pi_j \delta_{\btheta_j}$, define the Voronoi diagram
$$\calV_k = \left\{\btheta_j: \norm{\btheta_j - \btheta_{0k}} \leq \norm{\btheta_j - \btheta_{0l}}, \ \forall l \neq k\right\}, \quad k=1, \dots, K_0,$$
with corresponding index sets $\calI_k = \{1 \leq j \leq K: \btheta_j \in \calV_k\}$.
It follows from Lemma \ref{lem:ho} that there exists
a small enough choice of constants $M_1,c_1 > 0$ (depending on $G_0$ but not on $G$)
such that if $\kappa_n \leq M_1$, then	
\begin{equation}
\label{eq:w2bound_prop2}
W_2^2(G, G_0) > c_1 \left\{ \sum_{k=1}^{K_0} \sum_{j \in \calI_k} \pi_j \norm{\btheta_j - \btheta_{0k}}^2 +
                       \sum_{k=1}^{K_0} \left| \pi_{0k} - \sum_{j \in \calI_k} \pi_j \right|\right\}.
\end{equation}
Thus, using the fact that $\pi_j \geq b_0$ for all $j$, we have,
\begin{equation}
\label{voronoiDiam}
\norm{\btheta_j - \btheta_{0k}} < \frac {W_2(G,G_0)} {\sqrt{c_1 b_0}} \leq \frac{\kappa_n^{\frac 1 2}}{c_0\sqrt{c_1 b_0}}, \quad
\forall j \in \calI_k, \ k=1, \dots, K_0,
\end{equation}
where $c_0$ is the constant in Theorem \ref{thm:ho_long_results}. Let $c_2 = \frac{1}{c_0\sqrt{c_1 b_0}}$, and
 $\epsilon_0 = \inf\{\norm{\btheta_{0j} - \btheta_{0k}}: 1 \leq j<k \leq K_0\}$.
Choose $M_2 = \left(\frac{\epsilon_0\wedge \delta_0}{4c_2}\right)^2 \wedge \delta^2$, where $\delta$
is the constant in condition (C) on the cluster ordering $\alpha_{\bt}$. 
Fix $M = M_1 \wedge M_2$ for the rest of the proof, and assume $\kappa_n \leq M$.
In particular, we then obtain
$\norm{\btheta_j - \btheta_{0k}} < \epsilon_0/4$ for all $j \in \calI_k$, $k=1, \dots, K_0$.
It follows that for all $j,l \in \calI_k$, $k=1, \dots, K_0$, 
$$\norm{\btheta_j - \btheta_l} \leq \norm{\btheta_j - \btheta_{0k}} + \norm{\btheta_{0k} - \btheta_l} < \frac{\epsilon_0}{2},$$
and for all $k, k' =1, \dots, K_0$, $k\neq k'$, if $j\in \calI_k$ and $i \in \calI_{k'}$, 
$$\epsilon_0 \leq \norm{\btheta_{0k} - \btheta_{0k'}} \leq \norm{\btheta_{0k} - \btheta_j} + \norm{\btheta_j - \btheta_i} + \norm{\btheta_i - \btheta_{0k'}}
< \frac{\epsilon_0}{2} + \norm{\btheta_j - \btheta_i}.$$
Therefore, 
$$\max_{j,l \in \calI_k} \norm{\btheta_j - \btheta_l} < \min_{\substack{j\in \calI_k \\ i \in \calI_{k'}}} \norm{\btheta_j - \btheta_i},$$
for all $k \neq k'$, 
which implies that $\{\calV_1, \dots, \calV_{K_0}\}$ is a cluster partition, and
condition $(C)$ can be invoked on any cluster ordering over this partition.

As outlined at the beginning of Supplement \hyperref[section:proofs]{C}, recall that
 $\btheta = (\btheta_1, \dots, \btheta_K)$, $\btheta_0 = (\btheta_{01}, \dots, \btheta_{0K_0})$, 
  $\alpha =\alpha_{\btheta}$, and $\alpha_0 =\alpha_{\btheta_0}$.
For every $j=1, \dots, K$, let $k_j \in \{1, \dots, K_0\}$ be the unique
integer such that $j \in \calI_{k_j}$. Let
$$S_k = \{j: k_{\alpha(j)} =k, ~ k_{\alpha(j+1)} \neq k\} = \{j: \alpha(j) \in \calI_k, \ \alpha(j+1) \not\in \calI_k\},$$
for all $k=1, \dots, K_0,$ and let $S = \bigcup_{k=1}^{K_0} S_k$, which colloquially denotes
the set of indices for which the permutation $\alpha$
moves between Voronoi cells. 
We complete the proof in the following 5 steps.

\noindent\textbf{Step 1: The Cardinality of $S$.}
We claim that $|S| = K_0 - 1$.
Since $\alpha$ is a permutation, we must have $S_{k_{\alpha(K)}} = \emptyset$
and $S_k \neq \emptyset$ for all 
$k \neq \alpha(K)$, $k=1, \dots, K_0$. It follows that $|S| \geq K_0 - 1$. 
 
 By way of a contradiction, suppose that $|S| > K_0 - 1$. Then, 
 by the Pigeonhole Principle, there exists some $1 \leq k \leq K_0$ such that 
 for distinct $j, l \in \{1, 2, \dots, K-1\}$, 
$$\alpha(j), \alpha(l) \in \calI_k, \quad \alpha(j+1), \alpha(l+1) \not\in \calI_k,$$
which implies that $\alpha^{-1}(\calI_k)$ is not a consecutive set of integers, 
and contradicts the fact that $\alpha$ is a cluster ordering. Thus, $|S| = K_0 - 1$ as claimed. 

\noindent\textbf{Step 2:  
Bounding the distance between the atom differences of $G$ and $G_0$.
}
Using the previous step, we may write $S = \{j_1, j_2, \dots, j_{K_0-1}\}$, 
where $1\leq j_1 < \dots < j_{K_0-1} \leq K$. 
Recall that $\{\calV_1, \dots, \calV_{K_0}\}$ is a cluster partition of $\btheta$. 
Thus, it follows from the definition of cluster ordering that there exists $\tau \in S_{K_0}$ such that
$$(\btheta_{\alpha(1)}, \dots, \btheta_{\alpha(K)}) = (\calV_{\tau(1)}, \dots, \calV_{\tau(K_0)}),$$
where the right-hand side of the above display uses block matrix notation. 
Since condition (C) applies, we have $\tau = \alpha_0$. 
Colloquially, this means that the path taken by $\alpha_0$ \textit{between} the Voronoi cells
is the same as that of $\alpha$. 
Combining this fact with \eqref{voronoiDiam}, we have
$\norm{\btheta_{\alpha(j_k)} - \btheta_{0\alpha_0(k)}} \leq c_2 \kappa_n^{\frac 1 2}$ 
for all $k=1, \dots, K_0,$ and so, for all $k=1, \dots, K_0-1$, 
\begin{equation}
\label{etaDifBound}
\norm{\bfeta_{j_k} - \bfeta_{0k}} \leq 
 \norm{\btheta_{\alpha(j_k)} - \btheta_{0\alpha_0(k)}} + \norm{\btheta_{\alpha(j_k+1)} - \btheta_{0\alpha_0(k+1)}}
 \leq 2c_2\kappa_n^{\frac 1 2}.
\end{equation}

\noindent\textbf{Step 3: Bounding the distance between the atom differences of $\wtG$ and $G_0$.}
Let $\tbtheta = (\tbtheta_1, \dots, \tbtheta_K)$ and recall that $\tilde\alpha =\alpha_{\tbtheta}$, and
$\tilde\bfeta_j = \tbtheta_{\tilde\alpha(j+1)} - \tbtheta_{\tilde\alpha(j)}$, for all $j=1, \dots, K-1$.
Similarly as before, let
$$\tilde{\calV}_k = \left\{\tbtheta_j: \big\lVert\tbtheta_j - \btheta_{0k}\big\rVert \leq \big\lVert\tbtheta_j - \btheta_{0l}\big\rVert, \ \forall l \neq k\right\}, \quad k=1, \dots, K_0,$$
with corresponding index sets $\tilde\calI_k = \{1 \leq j \leq K: \tbtheta_j \in \calV_k\}$.
Using the same argument as in Step 0, and using the fact that $\kappa_n \geq \gamma_n$,
it can be shown that $\{\tilde\calV_1, \dots, \tilde\calV_{K_0}\}$ is a cluster
partition. Furthermore, 
$$\big\lVert\tbtheta_j - \btheta_{0k} \big\rVert \leq c_2 \gamma_n^{\frac 1 2}, \quad \forall j \in \tilde\calI_k, \ k=1, \dots, K_0.$$
Now, define $\tilde S = \bigcup_{k=1}^{K_0} \{j: \talpha(j) \in \tilde\calI_k, \ \talpha(j+1) \not\in \tilde\calI_k\}.$
Using the same argument as in Step 1,
we have $|\tilde S| = K_0-1$, and we may write $\tilde S = \{l_1, \dots, l_{K_0-1}\}$, 
where $l_1 < l_2 < \dots < l_{K_0-1}$. Using condition (C) on the cluster ordering $\alpha$, 
we then have as before
\begin{equation}
\label{tildeEtaDifBound}
\norm{\tilde\bfeta_{l_k} - \bfeta_{0k}} \leq 
 \big\|\tilde\btheta_{\talpha(l_k)} - 
 \btheta_{0\alpha_0(k)}\big\| + \norm{\btheta_{\talpha(l_k+1)} - \btheta_{0\alpha_0(k+1)}}
 \leq 2c_2 \gamma_n^{\frac 1 2}.
\end{equation}

On the other hand, recall that we defined $\omega_j = \norm{\tbfeta_{\psi(j)}}^{-\beta}$, 
where $\psi = v \circ u^{-1}$, and 
$u,v \in S_{K-1}$ are such that
$$\norm{\bfeta_{u(1)}} \geq \norm{\bfeta_{u(2)}} \geq \dots \geq \norm{\bfeta_{u(K-1)}}, \quad \text{and }
  \norm{\tbfeta_{v(1)}} \geq \dots \geq \norm{\tbfeta_{v(K-1)}}.$$
Since $\{\calV_1, \dots, \calV_{K_0}\}$ and $\{\tilde\calV_1, \dots, \tilde\calV_{K_0}\}$
are cluster partitions, and $|S| = |\tilde S| = K_0-1$, it is now a simple observation that $\norm{\bfeta_{u(1)}}, \dots, \norm{\bfeta_{u(K_0-1)}}$
and $\norm{\tbfeta_{v(1)}}, \dots, \norm{\tbfeta_{v(K_0-1)}}$ are the norms of the atom
differences \textit{between} Voronoi cells, which are bounded away from zero, 
and are resepectively in a $\kappa_n^{\frac 1 2}$- and $\gamma_n^{\frac 1 2}$- neighborhood of $\norm{\bfeta_{01}}, \dots, \norm{\bfeta_{0(K_0-1)}}$ up to 
reordering (also, the remaining 
$\norm{\bfeta_{u(K_0)}}, \dots, \norm{\bfeta_{u(K-1)}}$ and $\norm{\tbfeta_{v(K_0)}}, \dots, \norm{\tbfeta_{v(K-1)}}$
are precisely the norms of the atom differences \textit{within} Voronoi cells, and are therefore respectively
in a $\kappa_n^{\frac 1 2}$- and $\gamma_n^{\frac 1 2}$-neighborhood of zero). We therefore have,
$u^{-1}(j_k) = v^{-1}(l_k)$ for all $k=1, \dots, K_0-1$, and,
$$\psi(j_k) = (v \circ u^{-1})(j_k) = (v\circ v^{-1})(l_k) = l_k, \quad k=1, \dots, K_0-1.$$
Comparing this fact with \eqref{tildeEtaDifBound}, we arrive at,
\begin{equation}
\label{thm2_step2_2}
\norm{\tbfeta_{\psi(j_k)} - \bfeta_{0k}} \leq 2c_2 \gamma_n^{\frac 1 2}.
\end{equation}

\noindent\textbf{Step 4: Bounding the Weight Differences.}
The arguments of Step 3 can be repeated  to obtain 
\[
\label{psi0EtaBound}
\norm{\tbfeta_{\psi_0(k)} - \bfeta_{0k}} \leq 2c_2 \gamma_n^{\frac 1 2}, \quad k=1, \dots, K_0.
\]
In particular, by \eqref{thm2_step2_2},
\begin{equation}
\label{weightEtaDiff}
\big|\norm{\tbfeta_{\psi(j_k)}} - \norm{\tbfeta_{\psi_0(k)}}\big| \leq 4c_2 \gamma_n^{\frac 1 2}, \quad k=1, \dots, K_0.
\end{equation}
This is the key property which motivates our definition of $\omega_k$.
Now, since $\gamma_n \leq M \leq \left(\frac{\epsilon_0}{4c_2}\right)^2,$
we have
\begin{equation}
\label{etalb1}
\norm{\tilde\bfeta_{\psi(j_k)}} \geq \norm{\bfeta_{0k}} - \norm{\tilde\bfeta_{\psi(j_k)} - \bfeta_{0k}}
\geq \norm{\bfeta_{0k}} - 2c_2 \gamma_n^{\frac 1 2} \geq \frac{\epsilon_0}{2} > 0,
\end{equation} 
implying together with \eqref{weightEtaDiff} that 
\begin{align} 
\label{weightBound}
|\omega_{j_k} - \omega_{0k}| \lesssim \gamma_n^{\frac 1 2},\quad k=1, \dots, K_0-1.
\end{align}
  
\noindent\textbf{Step 5: Upper Bounding $\xi_n(G; \wtG)$.} 
We now use \eqref{etaDifBound}, \eqref{etalb1}, and \eqref{weightBound} to bound $\xi_n(G; \wtG)$. Since $r_{\lambda_n} \ge 0$
by condition (P1), we have
\begin{align}
\label{prop3interm}
\xi_n(G; \wtG)
 &= \sum_{k=1}^{K_0-1} r_{\lambda_n}(\norm{\bfeta_{0k}}; \omega_{0k}) - 
    \sum_{j \in S} r_{\lambda_n}(\norm{\bfeta_j}; \omega_{j})
  - \sum_{j \not\in S} r_{\lambda_n}(\norm{\bfeta_j}; \omega_{j}) \nonumber \\
 &\leq
  \sum_{k=1}^{K_0 - 1} r_{\lambda_n}(\norm{\bfeta_{0k}}; \omega_{0k}) - 
        \sum_{j \in S} r_{\lambda_n}(\norm{\bfeta_j}; \omega_j)  \\
\nonumber &= \sum_{k=1}^{K_0-1}\left\{
         r_{\lambda_n}(\norm{\bfeta_{0k}}; \omega_{0k}) - 
         r_{\lambda_n}(\norm{\bfeta_{j_k}}; \omega_{j_k})  
       \right\}\\
\nonumber &= \sum_{k=1}^{K_0-1}\left\{
         r_{\lambda_n}(\norm{\bfeta_{0k}}; \omega_{0k}) - 
         r_{\lambda_n}(\norm{\bfeta_{0k}}; \omega_{j_k})
       \right\}  \\ \nonumber
&  \qquad + \sum_{k=1}^{K_0 - 1} \left\{
         r_{\lambda_n}(\norm{\bfeta_{0k}}; \omega_{j_k}) - 
         r_{\lambda_n}(\norm{\bfeta_{j_k}}; \omega_{j_k})
       \right\}.
\end{align}
Now, by similar calculations as in equation \eqref{etalb1}, 
it follows that $\norm{\bfeta_{0k}}$, $\norm{\bfeta_{j_k}}$, $\omega_{0k}$, and $\omega_{j_k}$ lie
in a compact set away from zero which is constant with respect to $n$. 
Therefore, by penalty condition (P2), 
$$\xi_n(G;\wtG ) \lesssim 
\left(\gamma_n^{\frac 3 2}/\log n\right) \sum_{k=1}^{K_0-1} \left\{\norm{\bfeta_{j_k} - \bfeta_{0k}} + 
   |w_{j_k} - w_{0k}| \right\}.$$
Finally, invoking \eqref{etaDifBound} and \eqref{weightBound}, there exists $c > 0$ depending only on $G_0$ such that, 
\begin{equation}
\xi_n(G;\wtG) \leq c \gamma_n^{\frac 3 2}\left(\kappa_n^{\frac 1 2} + \gamma_n^{\frac 1 2} \right) / \log n.
\end{equation}
Since $c$ does not depend on $(G; \wtG)$, it is clear that this entire
calculation holds uniformly in the $(G; \wtG)$ under consideration,
which leads to the first claim. 
 
To prove the second claim, let $\rho_k = \sum_{j \in \calI_k} \pi_j$
for all $k=1, \dots, K_0$.
For all $G \in \calG_K(b_0;\kappa_n)$, we have $|\pi_{0k}-\rho_k| \lesssim \kappa_n$
by equation \eqref{eq:w2bound_prop2}, hence by conditions (F) and (P2),
\begin{align*}
\zeta_n(G)
 &= \frac 1 n [\varphi(\bpi_0) - \varphi(\bpi)] 
 \leq \frac 1 n [\varphi(\bpi_0) - \varphi(\rho_1, \dots, \rho_{K_0})] 
 \leq \ell_n \sum_{k=1}^{K_0} |\pi_k - \rho_k| \lesssim \gamma_n^{\frac 3 2}\kappa_n / \log n.
\end{align*}
The claim follows.  
\end{proof}
We are now in a position to prove Theorem \ref{densityConsistency}. 
\begin{proof}[Proof (Of Theorem \ref{densityConsistency})]
Let $\epsilon > 0$. By Theorem \ref{paramConsistency} and Lemma \ref{rate_tildeG}, 
there exists $b_0 > 0$ and an integer
$N_1 > 0$ such that for every $n \geq N_1$, 
$$\bbP\big(\hat G_n \in \calG_K(b_0)\big) > 1 - \frac{\epsilon}{2},  \quad 
  \bbP\big(\wtG_n \in \calG_K(b_0; \gamma_n)\big) > 1 - \frac{\epsilon}{2}.
  $$
Let $M > 0$ be the constant in the statement
of Proposition \ref{xiBound}, and let $N_2 > 0$ be a sufficiently large integer
such that $\gamma_n \leq M$ for all $n \geq N_2$. For the remainder of the proof, let $n \geq N_1 \vee N_2$.
The consistency of $\hat p_n$ with respect to the Hellinger distance
was already established in Theorem \ref{paramConsistency}, so it suffices to prove that 
$\bbP(\gamma_n < h(\hp_n, p_0) < M) \to 0$ as $n \to \infty$.
We have
\begin{align}
\bbP&\left(\gamma_n < h(\hp_n, p_0) < M\right) \nonumber \\
 &= 
  \bbP\left(\gamma_n < h(\hp_n, p_0) < M, 
    \{\hat G_n \not\in \calG_K(b_0)\} \cup \{\wtG_n\not\in \calG_K(b_0;\gamma_n)\} \right) \nonumber \\
    & \qquad + 
   \bbP\left(\gamma_n < h(\hp_n, p_0) < M, \hat G_n \in \calG_K(b_0), \wtG_n \in \calG_K(b_0; \gamma_n) \right) \nonumber \\
 &\leq \bbP\left(\hat G_n \not\in \calG_K(b_0)\right) +
                \bbP\left(\wtG_n \not\in \calG_K(b_0; \gamma_n)\right) \nonumber \\ 
 & 
 \qquad + \bbP\left(\gamma_n < h(\hp_n, p_0) < M, \hat G_n \in \calG_K(b_0), \wtG_n \in \calG_K(b_0; \gamma_n) \right)
\nonumber \\
 & \leq \frac{\epsilon}{2} + \frac{\epsilon}{2} 
 + \bbP\left(\gamma_n < h(\hp_n,p_0) < M, \hat G_n \in \calG_K(b_0),  \wtG_n \in \calG_K(b_0; \gamma_n)\right)    \nonumber \\
 & \leq \epsilon 
 + \bbP\left(\gamma_n/4 < h\left(\frac{\hp_n+p_0}{2},p_0\right) < M/\sqrt 2, \hat G_n \in \calG_K(b_0),\wtG_n \in \calG_K(b_0; \gamma_n)\right) 
    \label{thm2_presup} \\ 
 & \leq 
 \epsilon 
 + 
\bbP\left\{
\sup_{\substack{G \in \calG_K(b_0) \\ \gamma_n/4 < h(\bar p_G, p_0) < M/\sqrt 2 \\ \wtG \in \calG_K(b_0; \gamma_n)}} n^{-\frac 1 2} \nu_n(G) 
+ \frac 1 4 [\zeta_n(G) + \xi_n(G; \wtG)] - h^2(\bar p_G, p_0)\geq 0 \ \right\},
\label{thm2_sup}
\end{align}
where in \eqref{thm2_presup} we used the inequalities
in \eqref{simple_hellinger_bounds}, and in \eqref{thm2_sup} we used Lemma \ref{basicInequalities}.
It therefore suffices to prove that the right-hand side term of \eqref{thm2_sup} tends to zero.
To this end, let $\calS_n = \min\{s: 2^{s+1}\gamma_n > M/\sqrt 2 \}$. Then,
\begin{align}
\label{thm2_step}
\nonumber\bbP&\left\{
\sup_{\substack{G \in \calG_K(b_0) \\ \gamma_n/4 < h(\bar p_G, p_0) < M/\sqrt 2 \\ \wtG \in \calG_K(b_0; \gamma_n)}} n^{-\frac 1 2} \nu_n(G) + 
   \frac 1 4  [\zeta_n(G) + \xi_n(G; \wtG)] - h^2(\bar p_G, p_0)\geq 0 \ \right\} \\
 &\leq \sum_{s=0}^{\calS_n} \bbP\left\{\sup_{\substack{G \in \calG_K(b_0; (2^{s+1})\gamma_n/4) \\ \wtG \in \calG_K(b_0; \gamma_n)}} \nu_n(G) + 
 \frac{\sqrt n}{4}[\zeta_n(G) + \xi_n(G; \wtG)] \geq \sqrt n 2^{2s}\left(\frac{\gamma_n}{4}\right)^2  \right\}.
\end{align}
Thus, using Proposition \ref{xiBound} we have
\begin{align}
\label{thm2_step3}
\eqref{thm2_step}\leq \sum_{s=0}^{\calS_n} \bbP\left\{\sup_{G \in \calG_K(b_0; (2^{s+1})\gamma_n/4)} \nu_n(G) \geq \sqrt n 2^{2s}\left(\frac{\gamma_n}{4}\right)^2 -  \frac{c\sqrt n \gamma_n^2}{4\log n} \left(1+2^{\frac{s-1}{2}}  + \gamma_n^{\frac 1 2}2^{s-1} \right) \right\}.
\end{align}
We may now invoke Theorem \ref{thm:incrementGeer}. Let
$$a = \sqrt n 2^{2s}\left(\frac{\gamma_n}{4}\right)^2 - \frac{c\sqrt n \gamma_n^2}{4 \log n} \left(1+2^{\frac{s-1}{2}}
+ \gamma_n^{\frac 1 2}2^{s-1} \right).$$
We may set $R=2^{s+1}\gamma_n $ and $C_1=15$. 
It is easy to see that \eqref{thm511_2} is then satisfied. 
To show that condition \eqref{thm511_5} holds, note that
\begin{align*}
4C &\left(\int_0^{2^{s+1}\gamma_n} \sqrt{H_B\left(\frac u {\sqrt 2}, \bar\calP_K^{\frac 1 2}\left(2^{s+1}\frac{\gamma_n}{4}\right), \nu\right)} du \vee 2^{s+1}\gamma_n\right) \\
 &\leq 4C \left(\sqrt 2 \int_0^{2^{s+\frac 1 2}\gamma_n} \sqrt{H_B\left(u, \bar\calP_K^{\frac 1 2}\left(2^{s+\frac 1 2}\gamma_n\right), \nu\right)} du \vee 2^{s+1}\gamma_n\right) \\
 &\leq  4C \left(J \sqrt 2 \sqrt n 2^{2s+\frac 1 2} \gamma_n^2 \vee 2^{s+1}\gamma_n\right) \\
 &=  4C \left(J \sqrt n 2^{2s+ 1} \gamma_n^2 \vee 2^{s+1}\gamma_n\right). 
\end{align*} 
It is clear that $a \geq 2^{s+1}\gamma_n$ for sufficiently large $n$, and 
\begin{align*}
4C &J \sqrt n 2^{2s+1}\gamma_n^2 \\
 &= a - \sqrt n 2^{2s}\left(\frac{\gamma_n} 4\right)^2 + 4CJ \sqrt n 2^{2s+1} \gamma_n^2 + 
    \frac{c\sqrt n \gamma_n^2}{4\log n}  \left(1+2^{\frac{s-1}{2}} + \gamma_n^{\frac 1 2}2^{s-1} \right) \\
 &= a +  \sqrt n 2^{s} \gamma_n^2\left\{\left(8C J - \frac 1 {16}\right)2^{s} 
   + \frac c {4\log n}\left(2^{-s} + 2^{-\frac{s+1}{2}} + \gamma_n^{\frac 1 2}/2\right)\right\}. 
\end{align*}
Now, choose $J$ such that $8C J < \frac 1 {16}$. Then, for large enough $n$, since $\gamma_n \gtrsim (\log n/n)^{1/2}$, 
it is clear that
the right-hand term of the above quantity is negative, so condition \eqref{thm511_2} is satisfied.
Invoking Theorem \ref{thm:incrementGeer}, we have
\begin{align*}
\eqref{thm2_step3}
 \leq C \sum_{s=0}^{\calS_n} \exp\left\{-\frac{a^2}{16C^2R^2}\right\}.
\end{align*}
Now, a simple order assesment shows that $a$ is dominated by its first term.
Therefore, there exists $c_1 > 0$ such that $a^2 \geq c_1 n 2^{4s} \gamma_n^4$ for large enough $n$.  
Let $c_n = \frac{c_1}{64C^2}n \gamma_n^2$. Then, 
\begin{align*}
\eqref{thm2_step3}
& \lesssim  \sum_{s=0}^{\infty} \exp(-c_n2^{2s})
 \leq \exp(-c_n) - 1 + \sum_{s=0}^{\infty} \exp(-c_n s)\\
 &  =\exp(-c_n) - 1 + \frac 1 {1-\exp(-c_n)} \to 0,
\end{align*}
as $n \to \infty$, where we have used the fact that $c_n \to \infty$ because $\gamma_n \gtrsim (\log n/n)^{1/2}$.
The claim follows. 
\end{proof}

\subsection*{C.3. Proof of Theorem \ref{orderConsistency}} 
We now provide the proof of Theorem \ref{orderConsistency}.

\begin{proof}[Proof (Of Theorem \ref{orderConsistency})]
We begin with Part (i). According to Theorem \ref{paramConsistency}, the MPLE 
$\hat G_n$
of $G_0$ obtained by maximizing the penalized 
log-likelihood function $L_n$ is a consistent estimator of $G_0$ with respect to $W_r$, and therefore
has at least $K_0$ components with probability tending to one. 
It will thus suffice to prove that $\bbP(\hat K_n > K_0) = o(1)$. Furthermore, given $\epsilon > 0$, it follows from Theorems \ref{paramConsistency} 
and \ref{densityConsistency} that there exist $b_0, N > 0$ 
such that
$$\bbP(\hat G_n \in \calG_K(b_0;\gamma_n)) \geq 1 - \epsilon, \quad \forall n \geq N.$$
These facts imply
\begin{align*}
\bbP(\hat K_n > K_0) 
 &= \bbP\left\{\hat K_n > K_0, \hat G_n \not\in \calG_K(b_0;\gamma_n)\right\} +
    \bbP\left\{\hat K_n > K_0, \hat G_n \in \calG_K(b_0;\gamma_n)\right\}  \\
 &\leq \epsilon+
   \bbP\left\{\sup_{G \in \calG_K(b_0; \gamma_n) \setminus \calG_{K_0}} L_n(G) \geq \sup_{G \in \calG_{K_0}} L_n(G)\right\}.
\end{align*}
It will thus suffice to prove that the right-hand term in the above display tends to zero.
 To this end, let $G = \sum_{j=1}^K \pi_j \delta_{\btheta_j} \in \calG_K(b_0; \gamma_n) \setminus \calG_{K_0}$. 
Specifically, $G$ is any mixing measure with order 
$K>K_0$, such that $\pi_j \geq b_0$ for all $j=1, \dots, K$ and, by Theorem \ref{thm:ho_long_results},
\begin{equation}
\label{mixingNeighborhood}
W_2(G, G_0) = O(\gamma_n^{\frac 1 2}).
\end{equation} 
The dependence of $G$ on $n$ is omitted from its notation for simplicity.
Define the following Voronoi diagram  with respect to the atoms of $G$,
\[
\calV_k = \{\btheta_j: \norm{\btheta_j - \btheta_{0k}} < \norm{\btheta_j - \btheta_{0l}}, \forall l \not= k, 1 \le j \le K \}
~~,~~k=1, 2, \dots, K_0,
\]
and the corresponding index sets 
$\calI_k = \{1 \le j \le K: \btheta_j \in \calV_k \}$, for all $k=1, 2, \dots, K_0.$ 
Also, let $\rho_k = \sum_{j \in \calI_k} \pi_j$. Since the mixing proportions of $G$ are bounded below, it follows from 
\eqref{mixingNeighborhood} and Lemma \ref{lem:ho} that
\begin{equation}
\label{thetaNeighborhood}
\norm{\btheta_j - \btheta_{0k}} = O(\gamma_n^{\frac 1 2}), \quad \forall j \in \calI_k, \ k=1, \dots, K_0,
\end{equation}
and 
\begin{equation}
\label{piNeighborhood}
|\rho_k - \pi_{0k}| = O(\gamma_n^{\frac 1 2}), \quad k=1, \dots, K_0.
\end{equation}
Let $H_k$ be the following discrete measure, whose atoms are the elements of $\calV_k$,
\begin{equation}
\label{Hk_defn}
\displaystyle H_k = \frac 1 {\rho_k} \sum_{j \in \calI_k} \pi_j \delta_{\btheta_j}, \qquad k=1, 2, \dots, K_0.
\end{equation}
Note that $H_k$ is a mixing measure in its own right, and we may rewrite the mixing measure $G$ as 
\begin{equation}
\label{GDef}
G\ = \sum_{k=1}^{K_0} \rho_k  H_k.
\end{equation}
Furthermore, let $\alpha = \alpha_{\btheta}$ where $\btheta = (\btheta_1, \dots, \btheta_K)$,
and recall that $\bfeta_k = \btheta_{\alpha(k+1)} - \btheta_{\alpha(k)}$, for $k=1, 2, \dots, K-1$. 

On the other hand, 
let $\aG_n = \sum_{k=1}^{K_0} \rho_k \delta_{\abtheta_k}$ be the maximizer of $L_n(G)$ over the set of mixing 
measures in $\calG_{K_0}$ with mixing proportions fixed at $\rho_1, \dots, \rho_{K_0}$. 
 Under condition (A3), the same proof technique as Theorem 2, 
together with Theorem~\ref{thm:ho_long_results}, 
implies that the same rate holds under the Wasserstein distance, 
$$W_{2}(\aG_n, G_0) = O_p(\gamma_n^{\frac 1 2}).$$
Since $\aG_n$ has $K_0$ components, it follows that every atom of $\aG_n$ is in a $O_p(\gamma_n^{\frac 1 2})$-neighborhood
of an atom of $G_0$. Without loss of generality, we assume the atoms of $\aG_n$ are ordered such that
$\norm{\abtheta_{k} - \btheta_{0k}}=O_p(\gamma_n^{\frac 1 2})$. 
Letting $\aalpha = \alpha_{\abtheta}$, where $\abtheta = (\abtheta_1, \dots, \abtheta_K)$,
we define the differences  
$\abfeta_k = \abtheta_{\aalpha(k+1)} - \abtheta_{\aalpha(k)}$,
for $k=1, 2, \dots, K_0-1$. 

Note that
\begin{align*}
\bbP&\left\{\sup_{G \in \calG_K(b_0; \gamma_n) \setminus \calG_{K_0}} L_n(G) \geq \sup_{G \in \calG_{K_0}} L_n(G)\right\}
 \leq \bbP\left\{\sup_{G \in \calG_K(b_0; \gamma_n) \setminus \calG_{K_0}} L_n(G) - L_n(\aG_n) \geq 0\right\}.
\end{align*}
It will therefore suffice to prove that with probability tending to one, $L_n(G) < L_n(\aG_n)$, 
as $n \to \infty$. This implies that with probability tending to one, as $n \to \infty$, the MPLE 
cannot have more than $K_0$ atoms. We proceed as follows. 

Let $\bpi = (\pi_1, \dots, \pi_K)^\top$ and $\brho = (\rho_1, \dots, \rho_{K_0})^\top$. 
Consider the difference.
\begin{align}
\label{penLikDifference}
\nonumber L_n&(G) - L_n(\aG_n)  \\
  &= \left\{ l_n(G) - l_n(\aG_n) \right\} - \left\{\varphi(\bpi) - \varphi(\brho) \right\} 
   - 
   n \left\{ \sum_{k=1}^{K-1} r_{\lambda_n} (\norm{ \bfeta_k}; \omega_k) - 
   			 \sum_{k=1}^{K_0 - 1} r_{\lambda_n} (\norm{\abfeta_k}; \aomega_k) \right\} \nonumber \\
   &\leq \left\{ l_n(G) - l_n(\aG_n) \right\} 
   -
   n \left\{ \sum_{k=1}^{K-1} r_{\lambda_n} (\norm{ \bfeta_k}; \omega_k) - 
             \sum_{k=1}^{K_0 - 1} r_{\lambda_n} (\norm{\abfeta_k}; \aomega_k) \right\},
\end{align}
where the weights $\aomega_k$ are constructed in analogy to 
Section \ref{sec:asymptotic} of the paper, and where
the final inequality is due to condition (F) on $\varphi$.
We show this quantity is negative in three steps.

\noindent 
\textbf{Step 1: Bounding the Second Penalty Difference.}
We use the same decomposition as in  Proposition \ref{xiBound}. Write
\begin{equation}
\label{expandPenalty}
n\sum_{j=1}^{K-1} r_{\lambda_n}(\norm{\bfeta_j}; \omega_j) = 
n\sum_{k=1}^{K_0} \sum_{j: \alpha(j), \alpha(j+1) \in \calI_k}  r_{\lambda_n}(\norm{\bfeta_j}; \omega_j) + 
n\sum_{j \in S}  r_{\lambda_n}(\norm{\bfeta_j}; \omega_j),
\end{equation}
where, 
\begin{equation}
\label{sDef}
S = \bigcup_{k=1}^{K_0} \Big\{1 \leq j \leq K-1: \alpha(j) \in \calI_k, \ \alpha(j+1) \not\in \calI_k \Big\} = \{l_1, \dots, l_{K_0-1}\},
\end{equation}
such that $\norm{\bfeta_{l_k}-\bfeta_{0k}} = O_p(\gamma_n^{\frac 1 2})$, under condition (C). Therefore, 
\begin{align}
\label{Diffpenalty}
n&\left\{\sum_{k=1}^{K-1} r_{\lambda_n}(\norm{\bfeta_k}; \omega_k) - \sum_{K=1}^{K_0-1} r_{\lambda_n}(\norm{\abfeta_k}; \aomega_k) \right\} \nonumber \\
 &= 
n\sum_{k=1}^{K_0} \sum_{j: \alpha(j), \alpha(j+1) \in \calI_k}  r_{\lambda_n}(\norm{\bfeta_j}; \omega_j) + 
n\sum_{k=1}^{K_0-1} \Big\{  r_{\lambda_n}(\norm{\bfeta_{l_k}}; \omega_{l_k}) -   r_{\lambda_n}(\norm{\abfeta_k}; \aomega_k)\Big\}.
\end{align}

\noindent\textbf{Step 2: Bounding the Log-likelihood Difference.}
We now assess the order of $l_n(G) - l_n(\aG_n)$. We have, 
\[
l_n(G) - l_n(\aG_n) = \sum_{i=1}^n \log\Big \{1 + \Delta_i(G, \aG_n) \Big \},
\]                   
where, 
\[
\Delta_i(G, \aG_n) \equiv \Delta_i = \frac{p_G(\bY_i) - p_{\aG_n}(\bY_i)}{p_{\aG_n}(\bY_i)}.
\]
Using \eqref{GDef}, we have
\begin{equation}
\label{Delta}
\Delta_i 
 = \sum_{k=1}^{K_0} \rho_k ~\frac{p_{H_k}(\bY_i) - f(\bY_i; \abtheta_k)}{p_{\aG_n}(\bY_i)}
 = \sum_{k=1}^{K_0} \rho_k\int \frac{f(\bY_i; \btheta) - f(\bY_i; \abtheta_k)}{p_{\aG_n}(\bY_i)} dH_k(\btheta).
 \end{equation}
By the inequality $\log (1+x) \le x - x^2/2 + x^3/3$, for all $x \geq -1$, it then follows that,
\begin{equation}
\label{in-equal1}
l_n(G) - l_n(\aG_n) 
\le 
\sum_{i=1}^n \Delta_i - \frac 1 2 \sum_{i=1}^n \Delta_i^2 + \frac 1 3 \sum_{i=1}^n \Delta_i^3.
\end{equation}
We now perform an order assessment of the three terms on the right hand 
side of the above inequality.  

\noindent \textbf{Step 2.1. Bounding $\sum_{i=1}^n \Delta_i.$} We have,
\[
\sum_{i=1}^n \Delta_i  = \sum_{k=1}^{K_0} \rho_k ~\sum_{i=1}^n 
\frac{p_{H_k}(\bY_i) - f(\bY_i; \abtheta_k)}{p_{\aG_n}(\bY_i)},
\]
where the mixing measures $H_k$ are given in equation \eqref{Hk_defn}.
By a Taylor expansion, for any $\btheta=(\theta_1, \dots, \theta_d)$ close enough to each 
$\abtheta_k=(\atheta_{k1}, \dots, \atheta_{kd}), k=1, 2, \ldots, K$, 
there exists some $\bxi_k$ on the segment between $\btheta$ and $\abtheta_k$ such that, 
for all $i=1, \dots, n$,
the integrand in \eqref{Delta} can be written as,
\begin{align}
\nonumber \frac{f(\bY_i; \btheta) - f(\bY_i; \abtheta_k)}{p_{\aG_n}(\bY_i)}
=& \sum_{r=1}^d (\theta_r - \atheta_{kr}) U_{i,r}(\abtheta_k; \aG_n) \\ 
+ &
   \frac 1 2 \sum_{r=1}^d \sum_{l=1}^d 
   	(\theta_r - \atheta_{kr}) (\theta_l - \atheta_{kl}) 
   	U_{i,rl}(\abtheta_k; \aG_n)  \nonumber \\ 
+ & \frac 1 6 \sum_{r=1}^d \sum_{l=1}^d \sum_{h=1}^d 
   (\theta_r - \atheta_{kr}) (\theta_l - \atheta_{kl}) (\theta_h - \atheta_{kh})   
   U_{i,rlh}(\bxi_k; \aG_n),
   \label{tylor-exp}
\end{align}
where $U_{i,\cdot}(\btheta; G) \equiv U_{\cdot}(\bY_i;\btheta,G)$ are given in 
\eqref{U1} of the paper. 
 It then follows that
\begin{align}
\label{expansionFirstTerm}
\nonumber\sum_{i=1}^n &\frac{p_{H_k}(\bY_i) - f(\bY_i; \abtheta_k)}{p_{\aG_n}(\bY_i)}\\
 &=\nonumber \sum_{r=1}^d m_{k,r} \sum_{i=1}^n U_{i,r}(\abtheta_k; \aG_n) +
   \frac 1 2 \sum_{r=1}^d \sum_{l=1}^d 
   m_{k,rl} \sum_{i=1}^n U_{i,rl}(\abtheta_k; \aG_n) \\
 &+  \frac 1 6 \sum_{r=1}^d \sum_{l=1}^d \sum_{h=1}^d \int
   (\theta_r - \atheta_{kr}) (\theta_l - \atheta_{kl}) (\theta_h - \atheta_{kh})   
   \sum_{i=1}^n U_{i,rlh}(\bxi_k; \aG_n)  dH_k(\btheta),
\end{align}
where
$$m_{k,r} = \int (\theta_r - \atheta_{kr})dH_k(\btheta) \quad \text{and} \quad
  m_{k,rl}= \int (\theta_r - \atheta_{kr})(\theta_l - \atheta_{kl}) dH_k(\btheta)$$
for all $r, l = 1, 2, \dots, d$. 
Now, by construction we know that $\aG_n$ is a stationary point of $L_n(G)$. 
Therefore, its atoms satisfy the following equations, for all $r=1, \dots, d$,
\begin{align}
\label{systemEqns}
\nonumber \rho_{\aalpha(1)} \sum_{i=1}^n U_{i,r}(\abtheta_{\aalpha(1)}; \aG_n) + n \ \frac{\aeta_{1r}}{\norm{\abfeta_1}} \frac{\partial r_{\lambda_n}(\norm{\abfeta_1}; \aomega_1)}{\partial \eta} = 0 &, \\
\nonumber \rho_{\aalpha(k)} \sum_{i=1}^n U_{i,r}(\abtheta_{\aalpha(k)}; \aG_n) + n \ \frac{\aeta_{kr}}{\norm{\abfeta_{k}}} \frac{\partial r_{\lambda_n}(\norm{\abfeta_k}; \aomega_k)}{\partial \eta} - 
n \ \frac{\aeta_{(k-1)r}}{\norm{\abfeta_{k-1}}} \frac{\partial r_{\lambda_n}(\norm{\abfeta_{k-1}}; \aomega_{k-1})}{\partial \eta} = 0 &, \\ 
\nonumber \hfill k = 2, \dots, K_0-1,\\
\rho_{\aalpha(K_0)} \sum_{i=1}^n U_{i,r}(\abtheta_{\aalpha(K_0)}; \aG_n) - n \ \frac{\aeta_{(K_0-1)r}}{\norm{\abfeta_{K_0-1}}} \frac{\partial r_{\lambda_n}(\norm{\abfeta_{K_0-1}}; \aomega_{K_0-1})}{\partial \eta} = 0&,
\end{align}
\normalsize
where $\abfeta_k = (\aeta_{k1}, \dots, \aeta_{kd})$, for all $k=1, \dots, K_0-1$. 
Letting $u_{kr} = \aeta_{kr}/\norm{\abfeta_{k}}$, it follows that
\begin{align*}
\sum_{k=1}^{K_0} \rho_{\aalpha(k)} &\sum_{r=1}^d m_{\aalpha(k),r} \sum_{i=1}^n U_{i,r}(\abtheta_{\aalpha(k)}; \aG_n)  \\
 &= n\sum_{k=1}^{K_0-1} \sum_{r=1}^d  \Big\{m_{\aalpha(k+1),r} - m_{\aalpha(k),r}\Big\} u_{kr}
 		\frac{\partial r_{\lambda_n}(\norm{\abfeta_k}; \aomega_k)}{\partial\eta}\\
 &=  n\sum_{k=1}^{K_0-1}  \sum_{r=1}^d \left(\int \theta_r d H_{\aalpha(k+1)}(\btheta)- \int \theta_r dH_{\aalpha(k)}(\btheta) 
 -  \aeta_{kr} \right) u_{kr}\frac{\partial r_{\lambda_n}(\norm{\abfeta_k}; \aomega_k)}{\partial\eta}.
\end{align*}
Now, recall the set $S$ in \eqref{sDef} which has cardinality $K_0-1$, and was chosen such that
such that $\btheta_{l_k}$ is an atom of $H_{\aalpha(k)}$ and $\btheta_{l_{k}+1}$ is an atom of $H_{\aalpha(k+1)}$, under condition (C). Thus
\begin{align*}
\int &\theta_r d H_{\aalpha(k+1)}(\btheta) -
\int \theta_r dH_{\aalpha(k)}(\btheta) 
  \\
&=\eta_{l_kr} + \int (\theta_r -\theta_{(l_{k}+1)r})d H_{\aalpha(k+1)}(\btheta) - \int (\theta_r -\theta_{l_{k}r})dH_{\aalpha(k)}(\btheta) 
 \\
 &\leq \eta_{l_kr} + 2 \sum_{k=1}^{K_0} \sum_{h,j \in \calI_k} \norm{\btheta_h - \btheta_{j}}.
\end{align*}
We thus obtain from condition (P2) that for some constant $c > 0$, 
\begin{align}
\sum_{k=1}^{K_0} \rho_{\aalpha(k)} &\sum_{r=1}^d m_{\aalpha(k),r} \sum_{i=1}^n U_{i,r}(\abtheta_{\aalpha(k)}; \aG_n)  \nonumber \\
 &\leq \frac{cn\gamma_n^{\frac 3 2}}{\log n} \sum_{k=1}^{K_0} \sum_{h,j \in \calI_k} \norm{\btheta_h - \btheta_{j}}
+ n\sum_{k=1}^{K_0} \sum_{r=1}^d (\eta_{l_kr} - \aeta_{kr}) u_{kr} \frac{\partial r_{\lambda_n}(\norm{\abfeta_{k}}; \aomega_k)}{\partial \eta}
 \nonumber \\ 
 &= \frac{cn\gamma_n^{\frac 3 2}}{\log n}\sum_{k=1}^{K_0} \sum_{h,j \in \calI_k} \norm{\btheta_h - \btheta_{j}}
 + n\sum_{k=1}^{K_0}(\bfeta_{l_k}-\abfeta_{k})^\top \frac{\partial r_{\lambda_n}(\norm{\abfeta_{k}}; \aomega_k)}{\partial \bfeta}
 =:\Gamma_n.
 \label{Gamman}
\end{align}
We now consider the second term in \eqref{expansionFirstTerm}. 
Under condition (A3),  
$\mathbb E \{U_{i,rl}(\btheta, G_0) \} = 0$, 
by the Dominated Convergence Theorem, for 
all $\btheta \in \Theta$ and $r, l=1, 2, \dots, d$, so that
\begin{equation}
\label{secondTerm1}
\sum_{i=1}^n U_{i,rl}(\btheta, G_0) = O_p(n^{\frac 1 2}).
\end{equation}
Now, for all $k=1, \dots, K_0$, we write,
\begin{equation}
\label{secondTerm2}
\sum_{i=1}^n U_{i, rl}(\abtheta_k; \aG_n) = 
\sum_{i=1}^n U_{i, rl}(\abtheta_k; G_0) + 
\sum_{i=1}^n  \left\{ U_{i, rl}(\abtheta_k; \aG_n) - U_{i, rl}(\abtheta_k; G_0) 
\right\}
\end{equation}
The first term can be bounded as follows using \eqref{secondTerm1},
for some vectors $\tilde\btheta_{0ik}$ on the segment between $\btheta_{0k}$ and $\abtheta_k$,
\begin{align*}
\sum_{i=1}^n U_{i, rl}(\abtheta_k; G_0) 
 &= \sum_{i=1}^n U_{i, rl}(\btheta_{0k}; G_0) 
+\sum_{i=1}^n \Big[U_{i, rl}(\abtheta_k; G_0)  - U_{i, rl}(\btheta_{0k}; G_0) \Big] \\
 &= O_p(n^{1/2}) 
+\sum_{i=1}^n \sum_{s=1}^d U_{i, rls}(\tilde \btheta_{0ik}; G_0)(\atheta_{ks} - \theta_{0ks}) 
=O_p(n\gamma_n^{\frac 1 2}),
\end{align*}
where we invoked condition (A3) on the last line of the above display.
By the Cauchy-Schwarz inequality, 
we may bound the second term in \eqref{secondTerm2} as follows,
\begin{align*}
\sum_{i=1}^n & \left| U_{i, rl}(\abtheta_k; \aG_n) - U_{i, rl}(\abtheta_k; G_0) \right| \\
 &= \sum_{i=1}^n|U_{i, rl}(\abtheta_k; G_0)| \frac 1 {p_{\aG_n}(\bY_i)} \left| p_{G_0}(\bY_i) - p_{\aG_n}(\bY_i) \right| \\
 &\leq \sum_{i=1}^n |U_{i, rl}(\abtheta_k; G_0)| \times \\ 
   &\qquad\left\{\frac 1 {p_{\aG_n}(\bY_i)} \sum_{k=1}^{K_0} \pi_{0k} \big| f(\bY_i; \abtheta_k) - f(\bY_i; \btheta_{0k})\big| +
    \frac 1 {p_{\aG_n}(\bY_i)} \sum_{k=1}^{K_0} |\rho_k-\pi_{0k}|f(\bY_i; \abtheta_k) \right\}.   
\end{align*}
Now, for some $\bvarsigma_k$ on the segment joining $\btheta_{0k}$ to $\abtheta_k$, the above display is bounded above by
\begin{align*}
\phantom{\sum} 
 &\leq \sum_{i=1}^n |U_{i,rl}(\abtheta_k; G_0)| \times \\ 
  &\qquad\left\{ \frac 1 {p_{\aG_n}(\bY_i)} \sum_{k=1}^{K_0} \pi_{0k} \sum_{h=1}^d \left|\frac{\partial f(\bY_i; \bvarsigma_k)}{\partial \theta_h}\right| |\atheta_{kh} - \theta_{0kh}|  + \sum_{k=1}^{K_0} |\rho_k-\pi_{0k}|\frac {f(\bY_i; \abtheta_k)}{p_{\aG_n}(\bY_i)}   \right\}\\
&\asymp \sum_{i=1}^n |U_{i,rl}(\abtheta_k; G_0)|  
     \left\{ \sum_{k=1}^{K_0} \sum_{h=1}^d \left|U_{i,h}(\bvarsigma_k; \aG_n)\right| |\atheta_{kh} - \theta_{0kh}|  +  \sum_{k=1}^{K_0} |\rho_k-\pi_{0k}| U_i(\abtheta_k; \aG_n) \right\}\\    
 &\leq \left\{\sum_{i=1}^n |U_{i,rl}(\abtheta_k; G_0)|^2\right\}^{\frac 1 2}  \times \\ &\qquad
   \left\{\sum_{i=1}^n\left[ \sum_{k=1}^{K_0} \sum_{h=1}^d \left|U_{i,h}(\bvarsigma_k; \aG_n)\right| |\atheta_{kh} - \theta_{0kh}|  +  \sum_{k=1}^{K_0} |\rho_k-\pi_{0k}| U_i(\abtheta_k; \aG_n) \right]^2\right\}^{\frac 1 2} \\
 &\lesssim 
   \left\{
      \sum_{i=1}^n |U_{i,rl}(\abtheta_k; G_0)|^2
   \right\}^{\frac 1 2}   \times \\&\qquad
   \left\{ 
      \sum_{k=1}^{K_0} \sum_{h=1}^d \left[ |\atheta_{kh} - \theta_{0kh}|^2\sum_{i=1}^n\left|U_{i,h}(\bvarsigma_k; \aG_n)\right|^2  +  
      |\rho_k-\pi_{0k}|^2 \sum_{i=1}^n |U_i(\abtheta_k; \aG_n)|^2 \right]
   \right\}^{\frac 1 2},
\end{align*}
where $U_i(\cdot) \equiv U(\bY_i;\cdot)$ are given in \eqref{U0} of the paper, 
and we have used the Cauchy-Schwarz inequality in the second-to-last line. 
In view of condition (A3),  there exists $g \in L^3(P_0)$ such that 
$$\frac 1 n \sum_{i=1}^n |U_{i,rl}(\abtheta_k; G_0)|^2 \leq \frac 1 n \sum_{i=1}^n g^2(\bY_i) \overset{a.s.}{\longrightarrow} \bbE\left\{g^2(\bY)\right\},
~ \text{so, } \sum_{i=1}^n |U_{i,rl}(\abtheta_k; G_0)|^2 = O_p(n),$$
by Kolmogorov's Strong Law of Large Numbers. Similarly, 
by condition (A4),
$$\sum_{i=1}^n\left|U_{i,h}(\bvarsigma_k; \aG_n)\right|^2 = O_p(n), \quad \sum_{i=1}^n U_i^2(\abtheta_k; \aG_n) = O_p(n).$$
It follows that 
\begin{align}
\label{secondTerm3}
\nonumber\sum_{i=1}^n \big| &U_{i, rl}(\abtheta_k; \aG_n) - U_{i, rl}(\abtheta_k; G_0) \big|  \\
 &= O_p(n) \left\{
      \sum_{k=1}^{K_0} \sum_{h=1}^d \left[ |\atheta_{kh} - \theta_{0kh}|^2  +  
      |\rho_k-\pi_{0k}|^2 \right]
   \right\}^{\frac 1 2} = O_p(n\gamma_n^{\frac 1 2}).
\end{align}
Combining \eqref{secondTerm1}, \eqref{secondTerm2} and \eqref{secondTerm3}, we have
\begin{equation}
\label{secondTerm4}
\sum_{i=1}^n U_{i, rl}(\abtheta_k, \aG_n) = \sum_{i=1}^n U_{i, rl}(\abtheta_k, G_0) + O_p(n\gamma_n^{\frac 1 2})  = O_p(n\gamma_n^{\frac 1 2}).
\end{equation}
Regarding the third term in \eqref{expansionFirstTerm}, for all $r, l, d=1, 2, \dots, d$, 
under (A4), 
we again have
\begin{equation}
\label{thirdTerm}
\frac 1 n \sum_{i=1}^n U_{i, rld}(\bxi_k, \aG_n) = O_p(1).
\end{equation}
Thus, since all the atoms of the mixing measures $H_k$ in \eqref{GDef}
are in a $\gamma_n^{\frac 1 2}$-neighborhood of the true atoms of $G_0$,
\begin{align}
\label{thirdTerm2}
\nonumber  \frac 1 6 &\sum_{r=1}^d \sum_{l=1}^d \sum_{h=1}^d \int
(\theta_r - \atheta_{kr}) (\theta_l - \atheta_{kl}) (\theta_d - \atheta_{kh})   
 \sum_{i=1}^n U_{i,rld}(\bxi_k; \aG_n)  dH_k(\btheta)  \\
\nonumber   &= O_p(n) 
\sum_{r=1}^d \sum_{l=1}^d \sum_{h=1}^d \int |\theta_r - \atheta_{kr}| |\theta_l - \atheta_{kl}| |\theta_d - \atheta_{kh}| dH_k(\btheta)   \\
   &=|m_{2k}| O_p(n\gamma_n^{\frac 1 2})
\end{align}
where $m_{2k} = \sum_{r=1}^d \sum_{l=1}^d \int |\theta_r - \atheta_{kr}||\theta_l - \atheta_{kl}| dH_k(\btheta)$. 

Combining 
 \eqref{Gamman}, \eqref{secondTerm4} and \eqref{thirdTerm2}, we obtain
\begin{align}
\label{deltaLinear}
\sum_{i=1}^n \Delta_i 
\nonumber &= \Gamma_n + \sum_{k=1}^{K_0} \rho_k \left\{\frac 1 2 \sum_{r=1}^d \sum_{l=1}^d |m_{k, rl}| O_p(n\gamma_n^{\frac 1 2}) + \frac 1 6 |m_{2k}| O_p(n\gamma_n^{\frac 1 2})\right\}\\
 &\leq \Gamma_n +
 C_0 n \gamma_n^{\frac 1 2} \sum_{k=1}^{K_0}|m_{2k}|,
 \end{align}
in probability,
for some large enough constant $C_0 > 0$.

\noindent
\textbf{Step 2.2. Bounding $\sum_{i=1}^n \Delta_i^2$.}
By the Taylor expansion in \eqref{tylor-exp}, 
\begin{align*}
\sum_{i=1}^n \Delta_i^2 
 =& \sum_{i=1}^n \Bigg\{ \sum_{k=1}^{K_0} \rho_k \Bigg[
   \sum_{r=1}^d m_{k, r}U_{i, r}(\abtheta_k, \aG_n) +
   \frac 1 2 \sum_{r=1}^d \sum_{l=1}^d 
   m_{k,rl} U_{i,rl}(\abtheta_k; \aG_n) \\
 +& \frac 1 6 \sum_{r=1}^d \sum_{l=1}^d \sum_{h=1}^d \int
   (\theta_r - \atheta_{kr}) (\theta_l - \atheta_{kl}) (\theta_h - \atheta_{kh})   
   U_{i,rlh}(\bxi_k; \aG_n)  dH_k(\btheta)
\Bigg] \Bigg\}^2 \\
 =& (I) + (II) + (III)
\end{align*}
where, 
\begin{align*}
(I)  &=\sum_{i=1}^n \left\{
	   \sum_{k=1}^{K_0} \rho_k \left[ \sum_{r=1}^d 
 	    m_{k, r} U_{i, r}(\abtheta_k, \aG_n) + 
	   \frac 1 2 \sum_{r=1}^d \sum_{l=1}^d 
	   	m_{k,rl} U_{i,rl}(\abtheta_k; \aG_n) 
	  \right] \right\}^2  \\
(II) &= 
 \frac 1 {36} \sum_{i=1}^n \left\{\sum_{k=1}^{K_0} \rho_k \left[	  
   \sum_{r=1}^d \sum_{l=1}^d \sum_{h=1}^d \int
   (\theta_r - \atheta_{kr}) (\theta_l - \atheta_{kl})   
   (\theta_h - \atheta_{kh})   
   U_{i,rlh}(\bxi_k; \aG_n)  dH_k(\btheta)
 \right] \right\}^2   \\
(III) &= \frac 1 3 \sum_{i=1}^n \left\{ \sum_{k=1}^{K_0} \rho_k \left[
		   \sum_{r=1}^d m_{k, r} U_{i, l}(\abtheta_k, \aG_n) + 
	       \frac 1 2 \sum_{r=1}^d \sum_{l=1}^d 
	   	   m_{k,rl} U_{i,rl}(\abtheta_k; \aG_n) 
	     \right]\right\} \times  \\ &  \ \ \ \ \left\{ 
  		   \sum_{k=1}^{K_0} \rho_k 
   		   \sum_{r=1}^d \sum_{l=1}^d \sum_{h=1}^d \int
   		   (\theta_r - \atheta_{kr}) (\theta_l - \atheta_{kl})   
           (\theta_h - \atheta_{kh})   
           U_{i,rlh}(\bxi_k; \aG_n) dH_k(\btheta)
         \right\}.
\end{align*}
Define $\bM_{k1} = (m_{k,1}, \dots, m_{k, d})^\top$, $\bM_{k2} = (m_{k,11}, \dots, m_{k, dd})^\top$, 
$\bU_{i, 1}(\abtheta_k; \aG_n) 
= (U_{i, 1}(\abtheta_k; \aG_n)$, $\dots, U_{i, d}(\abtheta_k; \aG_n) )^\top$, $ \bU_{i, 2}(\abtheta_k; \aG_n)
 = (U_{i, 11}(\abtheta_k; \aG_n), \dots, U_{i, dd}(\abtheta_k; \aG_n))^\top$. Also, for $l=1, 2$, let $\bM_l = (\bM_{1l}, \dots, \bM_{K_0l})^\top$, $\bM = (\bM_1, \bM_2)^\top$, 
$$\bV_{il}(\abtheta; \aG_n) = \Big(\bU_{i,l}(\abtheta_1; \aG_n), \dots, \bU_{i,l}(\abtheta_{K_0}; \aG_n)\Big)^\top$$
for $l=1, 2$, and
$$\bV_i(\abtheta; \aG_n) = \Big(\bV_{i1}(\abtheta; \aG_n), \bV_{i2}(\abtheta; \aG_n)\Big)^\top$$
where $\abtheta = (\abtheta_1, \dots, \abtheta_{K_0})$. Then, since the $\rho_k$ are bounded away from zero in probability, we have,
\begin{align*}
(I) 
 &= \sum_{i=1}^n \left\{\sum_{k=1}^{K_0} \rho_k \left[ \bM_{k1}^\top \bU_{i1}(\abtheta_k; \aG_n) + \bM_{k2}^\top \bU_{i2}(\abtheta_k; \aG_n)\right]\right\}^2 \\
 &\asymp \sum_{i=1}^n \left\{\bM_1^\top \bV_{i1}(\abtheta; \aG_n) + \bM_2^\top \bV_{i2}(\abtheta_k; \aG_n)\right\}^2    
 = \sum_{i=1}^n \left\{\begin{pmatrix} \bM_1 &  \bM_2\end{pmatrix} \begin{pmatrix} \bV_{i1}(\abtheta_k; \aG_n) \\ \bV_{i2}(\abtheta_k; \aG_n) \end{pmatrix}\right\}^2  \\
 &= \sum_{i=1}^n \bM^\top \bV_i(\abtheta_k; \aG_n) \bV_i^\top(\abtheta_k; \aG_n) \bM 
 = \bM^\top \left( \sum_{i=1}^n \bV_i(\abtheta_k; \aG_n) \bV_i^\top(\abtheta_k; \aG_n) \right) \bM,
\end{align*} 
in probability. By \cite{serfling2002} (Lemma A, p. 253), as $n \to \infty$,  
\[
\frac 1 n \sum_{i=1}^n \bV_i(\abtheta_k; \aG_n) \bV_i^\top(\abtheta_k; \aG_n)  
 \overset{p}{\longrightarrow} \bSigma := \bbE\left\{\bV_1(\btheta_0; G_0) \bV_1^\top (\btheta_0; G_0)\right\}.
\]
It follows that for large $n$, the following holds in probability
$$\varrho_{\text{min}}(\bSigma)\norm{\bM}^2 \lesssim \frac 1 n (I)
\lesssim \varrho_{\text{max}}(\bSigma) \norm{\bM}^2.$$
By the Strong Identifiability Condition, $\bV_1(\btheta_0; G_0)$ is non-degenerate, 
so $\bSigma$ is positive definite and $\varrho_{\text{min}}(\bSigma) > 0$. Therefore, 
\begin{equation}
(I) \asymp
n \norm{\bM}^2,
\end{equation}
in probability, where $\norm{\bM}$ denotes the Frobenius norm of $\bM$.  

Using the same argument, and noting that $\big\lVert\btheta - \abtheta_k\big\rVert = o_p(1)$, 
for all $\btheta \in \Theta$ in a $\gamma_n^{\frac 1 2}$-neighborhood of an atom of $G_0$, we have,
\begin{equation*}
(II) = 
o_p(n)\norm{\bM_2}^2 =
o_p(n) \norm{\bM}^2.
\end{equation*}
By the Cauchy-Schwarz inequality, we also have,
\begin{equation*}
|(III)| \leq \sqrt{(I)(II)} = o_p(n) \norm{\bM}^2.
\end{equation*}
Combining the above inequalities, we deduce that for some constant $C' > 0$, 
\begin{equation}
\label{deltaSquared}
\sum_{i=1}^n \Delta_i^2 \geq n C' \norm{\bM}^2
 = nC'\sum_{k=1}^{K_0} \left\{\sum_{r=1}^d m_{k,r}^2 + \sum_{r=1}^d \sum_{l=1}^d m_{k, rl}^2 \right\},
\end{equation}
in probability.

\noindent
\textbf{Step 2.3. Bounding $\sum_{i=1}^n \Delta_i^3$.}
By a Taylor expansion, there exist vectors $\bxi_{ik}$   
on the segment joining $\btheta$ and $\abtheta_k$ such that
\begin{align}
\label{deltaCubed}
\sum_{i=1}^n \Delta_i^3  
 &=
  \sum_{i=1}^n \Bigg\{ \sum_{k=1}^{K_0} \rho_k \sum_{r=1}^d m_{k, r} U_{i, r}(\abtheta_k, \aG_n)  
  \\\nonumber &\qquad + 
\frac 1 2 \sum_{k=1}^{K_0} \rho_k \sum_{r=1}^d \sum_{l=1}^d \int (\theta_r - \atheta_{kr})(\theta_l
 - \atheta_{kl})dH_k(\btheta) U_{i,rl}(\bxi_{ik}, \aG_n) \Bigg\}^3 \\ \nonumber 
&=
 O_p(1) \sum_{k=1}^{K_0} \Bigg\{\sum_{j=1}^d |m_{k,r}|^3 \sum_{i=1}^n |U_{i, r}(\abtheta_k, \aG_n)|^3 
 \\ \nonumber &\qquad +
\sum_{r=1}^d \sum_{l=1}^d \int |\theta_r - \atheta_{kr}|^3 |\theta_l - \atheta_{kl}|^3 dH_k(\btheta) 
\sum_{i=1}^n U_{i, rl}^3(\bxi_{ik}, \aG_n) \Bigg\} \\
\nonumber  
&= 
O_p(n) \sum_{k=1}^{K_0} \left\{\sum_{r=1}^d |m_{k, r}|^3 
+ 
\sum_{r=1}^d \sum_{l=1}^d \int |\theta_r - \atheta_{kr}|^3 |\theta_l - \atheta_{kl}|^3 dH_k(\btheta) \right\} \\
 &= 
 o_p(n) \norm{\bM}^2,
\end{align}
where we have used Holder's inequality. Thus, 
\eqref{deltaSquared} and \eqref{deltaCubed} imply that $\sum_{i=1}^n \Delta_i^2$ 
dominates $\sum_{i=1}^n \Delta_i^3$, for large $n$. Hence, for large $n$, 
we can re-write \eqref{in-equal1} as 
\begin{equation}
\label{loglikDelta12}
l_n(G) - l_n(\aG_n) \leq \sum_{i=1}^n \Delta_i - \left(\frac 1 2 \sum_{i=1}^n \Delta_i^2 \right)(1 + o_p(1)) .
\end{equation}
Now, combining \eqref{deltaLinear} and \eqref{deltaSquared}, we have that for large $n$,
\begin{align*}
 \sum_{i=1}^n &\Delta_i - \left(\frac 1 2 \sum_{i=1}^n \Delta_i^2 \right) \\
 & \leq 
  \Gamma_n  + 
 C_0n\gamma_n^{\frac 1 2} \sum_{k=1}^{K_0}|m_{2k}|  
  - 
   n C' \sum_{k=1}^{K_0} \left( 
 	  \sum_{r=1}^dm_{k, r}^2 + \sum_{r=1}^d \sum_{l=1}^d m_{k, rl}^2
   \right) \\
 & \leq \Gamma_n + 
 C_d n\gamma_n^{\frac 1 2} \sum_{k=1}^{K_0}  \sum_{r=1}^d |m_{k,rr}| 
   - 
   n C' \sum_{k=1}^{K_0} 
 	  \sum_{r=1}^dm_{k, r}^2,\quad C_d=d^2C_0   \\
  &= \Gamma_n+ 
 C_dn\gamma_n^{\frac 1 2} \sum_{k=1}^{K_0} \sum_{r=1}^d\left\{ |m_{k, rr}| - m_{k, r}^2 \right\} - 
    n C' \sum_{k=1}^{K_0} \left( 
 		 \sum_{r=1}^d  m_{k, r}^2 - \frac{C_d \gamma_n^{\frac 1 2}}{C'} \sum_{r=1}^d m_{k,r}^2 \right).
\end{align*}
Notice that the final term of the above display is negative as $n \to \infty$, thus for large $n$,
\begin{align*}
 \sum_{i=1}^n \Delta_i - \left(\frac 1 2 \sum_{i=1}^n \Delta_i^2 \right) 
 &\leq \Gamma_n + 
 C_dn\gamma_n^{\frac 1 2} \sum_{k=1}^{K_0}\sum_{r=1}^d \left\{ |m_{k, rr}| -  m_{k, r}^2 \right\}\\
 &= \Gamma_n +  
    C_d n\gamma_n^{\frac 1 2} (1 + o_p(1)) \sum_{k=1}^{K_0} \sum_{r=1}^d  \sum_{h, i \in \calI_k} |\theta_{hr} - \theta_{ir}|^2  \\
 &= \Gamma_n +   O_p(n\gamma_n) \sum_{k=1}^{K_0} \sum_{h, i \in \calI_k} \|\btheta_{h} - \btheta_{i}\|.
\end{align*} 
Thus, returning to \eqref{loglikDelta12} and by using \eqref{Gamman}, we obtain for some constant $C_0 > 0$,
\begin{align}
\label{Step22Conclusion}
l_n(G) - l_n(\aG_n) 
& \leq 
C_0  n \gamma_n  \sum_{k=1}^{K_0} \sum_{h, i \in \calI_k} \|\btheta_{h} - \btheta_{i}\| 
+ n\sum_{k=1}^{K_0}(\bfeta_{l_k}-\abfeta_{k})^\top \frac{\partial r_{\lambda_n}(\norm{\abfeta_{k}}; \aomega_k)}{\partial \bfeta},
\end{align}
for large $n$, in probability. This concludes {\bf Step 2} of the proof.

\noindent 
\textbf{Step 3: Order assessment of the penalized log-likelihood difference. 
}

Combining \eqref{penLikDifference}, \eqref{Diffpenalty}
and \eqref{Step22Conclusion}, we obtain for a possibly different $C_0 > 0$,
\begin{align*} 
 L_n(G) &- L_n(\aG_n)  \\
&  \leq C_0  n\gamma_n \sum_{k=1}^{K_0} \sum_{h, i \in \calI_k} \|\btheta_{h} - \btheta_{i}\|
- n\sum_{k=1}^{K_0} \sum_{j: \alpha(j), \alpha(j+1) \in \calI_k}  r_{\lambda_n}(\norm{\bfeta_j}; \omega_j)
\nonumber 
\\ \nonumber
& 
- n\sum_{k=1}^{K_0-1}\left\{
 r_{\lambda_n}(\norm{\bfeta_{l_k}}; \omega_{l_k}) - 
 r_{\lambda_n}(\norm{\abfeta_k}; \aomega_{k}) - (\bfeta_{l_k}-\abfeta_{k})^\top \frac{\partial r_{\lambda_n}(\norm{\abfeta_{k}}; \aomega_k)}{\partial \bfeta}
 \right\},
 \end{align*}
for large $n$. 
Since $\omega_{l_k}=\aomega_k$ for large $n$ under condition (C), and 
$r_{\lambda_n}$ is nondecreasing and convex away from zero by (P1), the final term of the above display is negative.
Thus, 
\[
L_n(G) - L_n(\aG_n) \le 
C_0  n \gamma_n   \sum_{k=1}^{K_0} \sum_{h, i \in \calI_k} \|\btheta_{h} - \btheta_{i}\|
- n\sum_{k=1}^{K_0} \sum_{j: \alpha(j), \alpha(j+1) \in \calI_k}  r_{\lambda_n}(\norm{\bfeta_j}; \omega_j),
\]
for large $n$. 
By \eqref{thetaNeighborhood} and 
condition (P3) on $r_{\lambda_n}$, the right-hand-side of the above inequality is 
negative as $n \to \infty$. Thus any mixing measure $G$ with more than $K_0$ atoms cannot 
be the MPLE. This proves that
\begin{equation}
\label{sparsistency}
\bbP(\hat K_n = K_0) \to 1.
\end{equation}
Finally, we prove Part (ii), that is, we show that $\hat G_n$ converges to $G_0$ at the $\gamma_n$ 
rate with respect to the $W_1$ distance. 
In view of \eqref{sparsistency} and
Theorem \ref{thm:ho_long_results}, we have
\begin{align*}
\bbP&\left\{W_1(\hat G_n, G_0) > \gamma_n/k_0\right\}\\
 &= \bbP\left\{W_1(\hat G_n, G_0) > \gamma_n/k_0, \hat K_n = K_0\right\} + 
    \bbP\left\{W_1(\hat G_n, G_0) > \gamma_n/k_0, \hat K_n \neq K_0\right\} \\
 &\leq \bbP\left\{h(\hat p_n, p_0) > \gamma_n, \hat K_n = K_0\right\} + o(1) \\
 &= o(1),
\end{align*}
where the last line is due to Theorem 2. Thus, $W_1(\hat G_n, G_0) = O_p(\gamma_n)$. 
\end{proof}

\subsection*{C.4. Proofs of Strong Identifiability Results} 
In this section, we provide the proofs of Proposition \ref{binomialSI}
and Corollary \ref{corollaryMultinomial}. 

\begin{proof}[Proof (Of Proposition \ref{binomialSI})]
As in \cite{TEICHER1963}, write the probability generating function $(1-\theta + \theta z)^M$ 
of the family $\calF$ as $\psi(w;\theta) = (1 + \theta w)^M$, where $w = z-1$ for all $z \in \bbR$. 
For any fixed integer $K \geq 1$ and any distinct real numbers
$\theta_1, \dots, \theta_K \in (0,1)$, it is enough to show that if
 $\beta_{jl} \in \bbR$, $j=1, \dots, K$, $l=1, \dots, r$, satisfy
\begin{equation}
\label{binomialAssumption}
\sum_{j=1}^K \sum_{l=0}^r \beta_{jl} \frac{\partial^l \psi(w; \theta_j)}{\partial \theta^l}\\
 = 0
\end{equation}
uniformly in $w$, then $\beta_{jl} = 0$ for all $j,l$. 
Assume \eqref{binomialAssumption} holds. Writing $(m)_k = m!/(m-k)!$ for all
positive integers $m \geq k$, we have for all $w \in \bbR$,
\begin{align*}
0 =&
	\sum_{j=1}^K \sum_{l=0}^r \beta_{jl} \frac{\partial^l \psi(w; \theta_j)}{\partial \theta^l}\\
=& 
	\sum_{j=1}^K \sum_{l=0}^r \beta_{jl} (M)_l w^l(1 + w\theta_j)^{M-l}	\\
=& 
	\sum_{j=1}^K \sum_{l=0}^r  \beta_{jl} (M)_l w^l \sum_{s=0}^{M-l}{M-l \choose s} (w\theta_j)^s	\\
=& 
	\sum_{l=0}^r \sum_{s=0}^{M-l} \sum_{j=1}^K  {M-l \choose s} (M)_l \beta_{jl}w^{l+s} \theta_j^s	\\
=& 
	\sum_{l=0}^r \sum_{s=l}^M \sum_{j=1}^K {M-l \choose s-l}(M)_l  \beta_{jl}  w^s \theta_j^{s-l}	\\
=&
	 \sum_{s=0}^r \sum_{l=0}^s \sum_{j=1}^K {M-l \choose s-l}(M)_l  \beta_{jl}  w^s \theta_j^{s-l} + 
	 \sum_{s=r+1}^M \sum_{l=0}^r \sum_{j=1}^K {M-l \choose s-l}(M)_l  \beta_{jl}  w^s \theta_j^{s-l}. 
\end{align*}
This quantity is a uniformly vanishing polynomial in $w$.
It follows that its coefficients must vanish. We deduce 
\begin{align}
\label{system}
\begin{cases}
\displaystyle \sum_{l=0}^s \sum_{j=1}^K {M-l \choose s-l}(M)_l  \beta_{jl}  \theta_j^{s-l} = 0, & s = 0, \dots, r \\
\displaystyle \sum_{l=0}^r \sum_{j=1}^K {M-l \choose s-l}(M)_l  \beta_{jl}  \theta_j^{s-l} = 0, & s = r+1, \dots, M. 
\end{cases}
 \end{align}
This system of equations can
be written as $\calM_1 \bbeta = 0$,
where 
\[
\bbeta=(\beta_{10}, \dots, \beta_{K0}, 
\beta_{11}, \dots, \beta_{K1}, \beta_{12}, \dots, \beta_{Kr})^\top
\]
is a vector of length $K(r+1)$, and 
\scriptsize
\begin{align*}
\calM_1=\left(\begin{matrix}
{{M}\choose 0}(M)_0 \theta_1^0 & \dots & {{M}\choose 0}(M)_0 \theta_K^0 & 0 & \dots & 0 & \dots \\ 
{{M}\choose 1}(M)_0 \theta_1^1 & \dots & {{M}\choose 1}(M)_0 \theta_K^1 & 
 {{M-1}\choose 0}(M)_1 \theta_1^0 & \dots & {{M-1}\choose 0}(M)_1 \theta_K^0 & \dots \\
\vdots & & \vdots & \vdots & & \vdots & \\
 {{M}\choose r}(M)_0 \theta_1^r & \dots & {{M}\choose r}(M)_0 \theta_K^r & 
 {{M-1}\choose r-1}(M)_1 \theta_1^{r-1} & \dots & {{M-1}\choose r-1}(M)_1 \theta_K^{r-1} & \dots \\
 {{M}\choose r+1}(M)_0 \theta_1^{r+1} & \dots & {{M}\choose r+1}(M)_0 \theta_K^{r+1} & 
 {{M-1}\choose r}(M)_1 \theta_1^{r} & \dots & {{M-1}\choose r}(M)_1 \theta_K^{r} & \dots \\
\vdots & & \vdots & \vdots & & \vdots & \\
 {{M}\choose M}(M)_0 \theta_1^{M} & \dots & {{M}\choose M}(M)_0 \theta_K^{M} & 
 {{M-1}\choose M-1}(M)_1 \theta_1^{M-1} & \dots & {{M-1}\choose M-1}(M)_1 \theta_K^{M-1} & \dots \\
 \end{matrix}
 \right. \\ 
 \left.
 \begin{matrix}
\dots & 0 & \dots & 0 \\
\dots & 0 & \dots & 0 \\
\dots & \vdots & & \vdots\\
\dots &  {{M-r}\choose 0}(M)_r \theta_1^0 & \dots & {{M-r}\choose 0}(M)_r \theta_K^0  \\
\dots &  {{M-r}\choose 1}(M)_r \theta_1^1 & \dots & {{M-r}\choose 1}(M)_r \theta_K^1  \\
\dots &   \vdots & & \vdots \\ 
\dots & {{M-r}\choose M-r}(M)_r \theta_1^{M-r} & \dots & {{M-r}\choose M-r}(M)_r \theta_K^{M-r}  \\
 \end{matrix} \right)
\end{align*}
\normalsize
is a matrix of dimension $(M+1) \times K(r+1)$.

Now, using the fact that for all $k=0, \dots, r$ and $ 0 \leq m \leq M$,
$$(M)_k\frac{{{M-k}\choose m-k}}{{M\choose m}} = (m)_k,$$
we have that $\calM_1$ can be reduced by elementary operations to
 \scriptsize
\begin{equation*}
\calM_2 = \begin{pmatrix}
\theta_1^0 & \dots & \theta_K^0 & 0 & \dots & 0 & \dots & 0 & \dots & 0 \\
 \theta_1^1 & \dots &  \theta_K^1 & 
 (1)_1 \theta_1^0 & \dots & (1)_1 \theta_K^0 & \dots &  0 & \dots & 0 \\
\vdots & & \vdots & \vdots & & \vdots & & \vdots & & \vdots \\
\theta_1^r & \dots &  \theta_K^r & 
 (r)_1 \theta_1^{r-1} & \dots & (r)_1 \theta_K^{r-1} & \dots & 
 (r)_r \theta_1^0 & \dots & (r)_r \theta_K^0  \\
\theta_1^{r+1} & \dots &  \theta_K^{r+1} & 
 (r+1)_{1} \theta_1^{r} & \dots & (r+1)_{1} \theta_K^{r} & \dots & 
 (r+1)_r \theta_1^1 & \dots & (r+1)_{r} \theta_K^1  \\
\vdots & & \vdots & \vdots & & \vdots & & \vdots & & \vdots \\
\theta_1^{M} & \dots & \theta_K^{M} &
 (M)_{1} \theta_1^{M-1} & \dots & (M)_{1} \theta_1^{M-1} & \dots &
 (M)_{r} \theta_1^{M-r} & \dots & (M)_{r} \theta_K^{M-r}  \\
 \end{pmatrix}.
\end{equation*}
\normalsize
If $M+1 < (r+1)K$, namely if $\calM_2$ has
more columns than it has rows, the system \eqref{system} must have infinitely-many solutions, 
and so the family $\calF$ 
is not strongly identifiable in the $r$-th order. 
On the other hand, if $M+1 \geq (r+1)K$, let $\calM_3$ denote the top $(r+1)K \times (r+1)K$
block of $\calM_2$ (namely the square matrix consisting of the
first $(r+1)K$ rows of $\calM_2$). 
Then $\calM_3$ is the generalized (or confluent) Vandermonde matrix \citep{KALMAN1984}
corresponding to the polynomial
$$g(x) = \prod_{i=1}^K (x-\theta_i)^{r+1},$$
up to permutation of its columns.
It follows that 
$$|\det(\calM_3)| = \prod_{1 \leq i < j \leq K} (\theta_i - \theta_j)^{(r+1)^2}.$$
Since $\theta_1, \dots, \theta_K$ are assumed to be distinct, we deduce that $\calM_3$
is invertible, whence $\calM_2$ is full rank and the system of equations \eqref{system} has 
a unique solution $\beta=0$. The claim follows.
\end{proof}

\begin{proof}[Proof (Of Corollary \ref{corollaryMultinomial})]
Assume $3K-1 \leq M$. 
Suppose $\zeta_j \in \bbR$ and $\bbeta_j, \gamma_j \in \bbR^d$, $j=1, \dots, K$, 
are such that for any $\by$, 
\begin{align*}
\sum_{j=1}^K  \left\{ 
	\zeta_j f(\by; \btheta_j) + 
	\bbeta_j^\top \frac{\partial f(\by; \btheta_j)}{\partial \btheta} + 
	\bgamma_j^\top \frac{\partial^2 f(\by; \btheta_j)}{\partial \btheta \partial \btheta^\top} \bgamma_j
\right\}  &= 0.
\end{align*}
Then, writing $\by=(y_1, \dots, y_d)^\top$, we have for all $s \in \{1, \dots, d\}$,
\begin{align*}
\sum_{y_1, \dots, y_{s-1}, y_{s+1}, \dots, y_d=0}^M  	\sum_{j=1}^K  \left\{ 
	\zeta_j f(\by; \btheta_j) + 
	\bbeta_j^\top \frac{\partial f(\by; \btheta_j)}{\partial \btheta} + 
	\bgamma_j^\top \frac{\partial^2 f(\by; \btheta_j)}{\partial \btheta \partial \btheta^\top} \bgamma_j
\right\}  &= 0.
\end{align*}
Write $\bbeta_j = (\beta_{j1}, \dots, \beta_{jd})^\top$ and $\bgamma_j = (\gamma_{j1}, \dots, \gamma_{jd})^\top$
for all $j=1, \dots, K$. 
Letting $b(y; \theta) = {M \choose y} \theta^y (1-\theta)^{M-y}$ denote the binomial density, 
and using the fact that multinomial densities have binomial marginals, we have
\begin{align*}
\sum_{j=1}^K  \left\{ 
	\zeta_j b(y_s; \theta_{js}) + 
	\beta_{js} \frac{\partial b(y_s; \theta_{js})}{\partial \theta} + 
	\gamma_{js}^2 \frac{\partial^2 b(y_s; \theta_{js})}{\partial \theta^2}
\right\}  &= 0.
\end{align*}
Since $3K-1 \leq M$, it follows by Proposition \ref{binomialSI}  
that $\zeta_j = \beta_{js} = \gamma_{js} = 0$ for all $j=1, \dots, K$.
Since this holds for all $s=1, \dots, d$, the claim follows.
\end{proof}

\subsection*{C.5. Proof of Proposition \ref{prop:misspec_K}} 

First, we introduce some notation. 
Recall Theorem \ref{thm:KER2000LER1992}.(i), where for all $k \ge 1$, 
there is a mixing measure $G^*_k \in \calG_k$ for which 
$\KL(p_{G^*_k}, p_{G_0})  = \inf_{G \in \calG_k} \KL(p_G, p_{G_0})$.
Also, for all $k =1, \dots, K<K_0$, denote 
\[
\bar G_n^{(k)} = \argmax_{G \in \calG_k} l_n(G), \quad
\hat G_n^{(k)} = \argmax_{G \in \calG_k} L_n(G),  
\]
respectively as the MLE and MPLE with order $k$. In particular, 
$\hat G_n^{(K)} = \hat G_n$. 
We now turn to the proof.
\begin{proof}[Proof (Of Proposition \ref{prop:misspec_K})]
Notice that
\[
\bbP(\hat K_n = K) \geq \bbP\Big(L_n(\hat G_n^{(K)}) > L_n( \hat G_n^{(k)}), \ 1 \leq k \leq K-1\Big).
\]
In what follows we show that the right-hand side tends to one, as $n \to \infty$. 

For all $k=1, \dots, K-1$, by definition of the penalized likelihood 
$L_n$ in \eqref{penloglik} of the paper, and using the fact that 
$\hat G_n^{(K)}$ is the maximizer of $L_n$ over $\calG_K$, we have 
\begin{align*}
L_n(\hat G_n^{(k)})  - L_n(\hat G_n^{(K)})
  \leq l_n(\hat G_n^{(k)}) - 
  L_n(G^*_K).
\end{align*}
The above inequality, combined with the fact that $\bar G_n^{(k)}$ maximizes $l_n$
over the space $\calG_k$, for all $k=1, \dots, K-1$, leads to
\begin{align*}
L_n(\hat G_n^{(k)})& - L_n(\hat G_n^{(K)})
\\
  &\leq l_n(\bar G_n^{(k)}) - L_n(G^*_K)
  \\
  &= l_n(\bar G_n^{(k)}) - l_n(G^*_K)
  + \varphi(\bpi^*_K) 
  + n \sum_{j=1}^{K-1} r_{\lambda_n}(\|\bfeta^*_j\|; \omega_j)\\
 &= \Big[l_n(\bar G_n^{(k)}) - l_n(G_0)\Big] - 
 \Big[l_n(G^*_K) - l_n(G_0)\Big]
 \\ & \qquad\qquad\qquad\qquad\quad~~+ 
 \varphi(\bpi^*_K) 
  + n \sum_{j=1}^{K-1} r_{\lambda_n}(\|\bfeta^*_j\|; \omega_j)\\
  &= -n (1+o_p(1))\Big[\KL(p_{G^*_k}, p_{G_0}) - 
  \KL(p_{G^*_K}, p_{G_0})\Big]
 + o(n),
  \end{align*}
  for large $n$, where we invoked \eqref{eq:keribin} for the first log-likelihood difference 
in $[\cdot]$, the strong law of large numbers
for the second log-likelihood difference 
in $[\cdot]$, condition (P1) on the penalty $r_{\lambda_n}$, and 
that \eqref{eq:leroux} 
guarantees $\pi_j^* > 0$ for all $j=1,\dots, K$, whence $\varphi(\bpi^*_K) = o(n)$
by condition (F). 
By \eqref{eq:leroux}, 
the difference of the two KL divergences is strictly positive and bounded.  
It must then follow that, for all $k=1, \dots, K-1$,
\begin{align*}
L_n(\hat G_n^{(k)})& - L_n(\hat G_n^{(K)}) < 0
\end{align*}
with probability tending to one, as $n \to \infty$. The claim follows.
\end{proof}

\section*{Supplement D: Numerical Solution}

In this section, we provide computational strategies for implementing the GSF method. 
In Section \hyperref[sec:implement]{D.1}, we describe a modified EM algorithm 
to obtain an approximate solution to the optimization problem in \eqref{mple1}, 
and in Section \hyperref[sec:impel-specific]{D.2} we outline some implementation specifications.
Regularized plots and the choice of tuning parameter $\lambda$ in the penalty $r_{\lambda}$
are discussed in Section \ref{sec:simulation} of the paper.  
\subsection*{D.1 A Modified Expectation-Maximization Algorithm}
\label{sec:implement}

In what follows, we describe a numerical solution to the optimization problem in \eqref{mple1}
based on the Expectation-Maximization 
(EM) algorithm \citep{DEMPSTER1977} and the proximal gradient method
\citep{NESTEROV2004}. 

Given a fixed upper bound $K > K_0$, 
the penalized complete log-likelihood function is given by
\begin{equation}
\label{comp-like}
L_n^c (\bPsi) = 
\sum_{i=1}^n \sum_{j=1}^K Z_{ij} \left[ \log \pi_j + \log f(\by_i; \btheta_j)  \right] 
- \varphi(\pi_1, \dots, \pi_K) - n\sum_{j=1}^{K-1} r_{\lambda} (\norm{\bfeta_j}; \omega_j)
 \end{equation}
where the $Z_{ij}$ are latent variables indicating the component to 
which the $i$th observation $\by_i$ belongs, for all $i=1, 2, \dots, n$, $j=1, 2, \dots, K$, and 
$\bPsi = (\btheta_1, \dots, \btheta_K, \pi_1, \dots, \pi_{K-1})$ is the vector of 
all parameters. 
Since the $Z_{ij}$ are missing, our modified EM algorithm maximizes the 
conditional expected value (with respect to $Z_{ij}$) of the penalized complete 
log-likelihood \eqref{comp-like}, by iterating between the two steps which follow. 
We let $\bPsi^{(t)} = (\btheta_1^{(t)}, \dots, \btheta_K^{(t)}, \pi_1^{(t)}, \dots, \pi_{K-1}^{(t)})$ 
denote the parameter estimates on the $t$-th iteration of the algorithm. 
Inspired by the local linear approximation (LLA) for folded concave penalties
\citep{ZOU2008}, at the $(t+1)$-th iteration, the 
modified EM algorithm proceeds as follows. 

\noindent\textbf{E-Step.}
Compute the conditional expectation of $L_n^c (\bPsi)$  
with respect to $Z_{ij}$, given observations $\by_1, \by_2, \dots, \by_n$ and the 
current estimate $\bPsi^{(t)}$, as 
\begin{align*}
\label{pen-clike}
Q(\bPsi; \bPsi^{(t)}) &= 
\sum_{i=1}^n \sum_{j=1}^K w_{ij}^{(t)} [\log\pi_j + \log \{ f(\by_i; \btheta_j) \} ]
\\&  \qquad  -
\varphi(\pi_1, \dots, \pi_K)
- 
n \sum_{j=1}^{K-1} r'_{\lambda} (\|{\bfeta^{(t)}_j}\|; \omega_j) \norm{\bfeta_j}
\end{align*} 
where 
\[
\displaystyle w_{ij}^{(t)} = 
	\frac{\pi_j^{(t)} \log \{ f(\by_i; \btheta_j^{(t)}) \} }
	     {\displaystyle\sum_{l=1}^K \pi_l^{(t)} \log \{ f(\by_i; \btheta_l^{(t)}) \} },
  \qquad i=1, 2, \dots, n; \ j = 1, 2, \dots, K.
\]

\noindent\textbf{M-Step.} 
The updated estimate $\bPsi^{(t+1)}$ is obtained by 
minimizing $-Q(\bPsi; \bPsi^{(t)})$
with respect to $\bPsi$. The mixing proportions are updated by
$$\left(\pi_1^{(t+1)}, \dots, \pi_K^{(t+1)}\right)^\top = \bpi^{(t+1)} 
= \argmin_{\bpi} \left\{\sum_{i=1}^n \sum_{j=1}^Kw_{ij}^{(t)}  \log\pi_j - \varphi(\bpi)\right\}.
$$
For instance, if $\varphi(\bpi) = -  C \sum_{j=1}^K \log \pi_j$, 
for some constant $C =\gamma-1$ with $\gamma > 1$, we arrive at
\[
\pi_j^{(t+1)} = \frac{ \sum_{i=1}^n w_{ij}^{(t)} + C}{n + KC}, \qquad j = 1, 2, \dots, K.
\] 
On the other hand, there generally does not exist a closed form update for 
$\btheta_1, \dots, \btheta_K$. Inspired by the proximal gradient method,
we propose to locally majorize the objective function $-Q(\bPsi; \bPsi^{(t)})$, 
 holding the mixing probabilities $\pi_j$ constant. \cite{XU2015} considered a similar approach for 
 one-dimensional exponential families ${\cal F}$.

Let $\bfeta_0 = \btheta_{\alpha(1)}$ and recall 
$\bfeta_j = \btheta_{\alpha(j+1)} - \btheta_{\alpha(j)}$, 
for all $j=1, 2, \dots, K-1$. Define the matrix $\bfeta = (\bfeta_0, \dots, \bfeta_{K-1}) \in \bbR^{d \times K}$, 
and note that $\btheta_{\alpha(j)} = \sum_{l=0}^{j-1} \bfeta_l$, for all $j=1, 2, \dots, K$.  
We then rewrite the leading term of the function $-\frac 1 n Q(\bPsi; \bPsi^{(t)})$ as 
\[
\calL(\bfeta; \bPsi^{(t)}) = -\frac{1}{n} \sum_{i=1}^n \sum_{j=1}^K w_{i\alpha(j)}^{(t)} \log 
	f \left(\by_i; \textstyle \sum_{l=0}^{j-1} \bfeta_l \right).
	\]
Hence, the problem of of minimizing $-Q(\bPsi; \bPsi^{(t)})$ with respect to 
the $\btheta_j$ is equivalent to minimizing 
\[
\calQ(\bfeta; \bfeta^{(t)}) = \calL(\bfeta; \bPsi^{(t)}) + 
\sum_{j=1}^{K-1} r'_{\lambda} (\|{\bfeta^{(t)}_j}\|; \omega_j) \norm{\bfeta_j},
\] 
with respect to $\bfeta \in \bbR^{d \times K}$. Given a tuning parameter $\rho > 0$, 
we locally majorize $\calQ(\bfeta; \bfeta^{(t)})$ by the following isotropic quadratic function
\begin{align}
\nonumber 
\overline{\cal Q}(\bfeta; \bfeta^{(t)}) = \calL(\bfeta^{(t)}; \bPsi^{(t)}) &+ 
   \tr \left\{ 
	\left[\frac{\partial \calL}{\partial \bfeta} (\bfeta^{(t)})\right]^{\top} \left(\bfeta - \bfeta^{(t)} \right) 
   \right\} \\ & \qquad +  
  \frac{\rho}{2} \big\lVert \bfeta - \bfeta^{(t)} \big\rVert^2 
  +
\sum_{j=1}^{K-1} r'_{\lambda} (\|{\bfeta^{(t)}_j}\|) \norm{\bfeta_j}.
\end{align} 
Note that $\overline\calQ(\cdot; \bfeta^{(t)})$ majorizes 
$\calQ(\cdot; \bfeta^{(t)})$ at $\bfeta^{(t)}$ provided
\begin{equation}
\label{rhoBound}
\rho \geq \max\left\{
		\varrho_{\text{max}}\left(\frac{\partial^2 \calL(\bfeta^{(t)}; \bPsi^{(t)} )}
		{\partial \bfeta_j \partial \bfeta_k}\right): j, k=0, 1, \dots, K-1 \right\},
\end{equation}
where $\varrho_{\text{max}}(\bM)$ denotes the largest eigenvalue of any real 
and symmetric $d \times d$ matrix $\bM$.
The numerical choice of $\rho$ is discussed below. 
Then, setting $\bfeta^{(t, 0)} = \bfeta^{(t)}$, the $(m+1)$-th update 
of $\bfeta$ on the $(t+1)$-th iteration of the EM Algorithm is given by
\begin{equation}
\label{mmUpdate}
\bfeta^{(t+1, m+1)} 
= 
\argmin_{\bfeta \in \bbR^{d \times K}} \overline{\cal Q}(\bfeta; \bfeta^{(t+1, m)}),
\end{equation}
which has the following closed-form
\begin{eqnarray}
\label{threshold-solution}
\bfeta^{(t+1, m+1)}_0 & = & \bfeta^{(t+1, m)}_0  - \rho^{-1} \frac{\partial \calL(\bfeta^{(t+1, m)}; \bPsi^{(t)})}{\partial \bfeta_0}\\ 
\nonumber
 \\
\bfeta^{(t+1, m+1)}_j & = & S\left(\bz^{(t, m+1)}_j; \rho^{-1}r'_{\lambda} (\|{\bfeta^{(t)}_j}\|; \omega_j) \right),
\end{eqnarray}
for all $j= 1, 2, \ldots, K-1$, 
where $\bz^{(m, t+1)}_j = \bfeta^{(m, t+1)}_j  - \rho^{-1} \frac{\partial \calL(\bfeta^{(m, t+1)}; \bPsi^{(t)})}{\partial \bfeta_j}$, 
and $S(\bz; \lambda) = \left(1 - \frac {\lambda} {\norm{\bz}} \right)_+ \bz$ is the 
multivariate soft-thresholding operator (\citealt{BREHENY2015}, 
\citealt{DONOHO1994}). 

Returning to \eqref{rhoBound}, to avoid computing 
the second-order derivatives of $\calL(\bfeta; \bPsi^{(t)})$, 
we determine the value of $\rho$ by performing a line search at each 
iteration of \eqref{mmUpdate}. Specifically, given a small constant $\rho_0$, at the iteration $m+1$ 
we set $\rho = \rho_0$ and increase it by a factor 
$\gamma > 1$ until the local majorization property is satisfied:
$$\calQ(\bfeta^{(t+1, m+1)}; \bfeta^{(t+1, m)}) \leq \overline\calQ(\bfeta^{(t+1, m+1)}; \bfeta^{(t+1, m)}).$$
Let $\rho^{(t+1, m+1)}$ denote the selected value of $\rho$. 
To speed up the selection of $\rho$, we initialize it on the $(m+1)$-th iteration by 
$\max\left\{\rho_0, \gamma^{-1} \rho^{(t+1, m)}\right\}$, similarly to \cite{FAN2018}.

The update \eqref{mmUpdate} in the M-Step is iterated until an index $m$ satisfies 
$\big\lVert \bfeta^{(t+1, m+1)} - \bfeta^{(t+1, m)} \big\rVert < \epsilon$ 
for some small $\epsilon > 0$. We then set the values of the $(t+1)$-th 
iteration of the EM algorithm as $\bfeta^{(t+1)} := \bfeta^{(t+1, m_0)}$ 
and $\btheta^{(t+1)} := \bfeta^{(t+1)} \Lambda$, where $\Lambda$ 
is the triangular $K\times K$ matrix with ones above and on the 
diagonal. The iteration between the E-Step and M-Step 
is continued until a convergence criterion is met, say 
$\big\lVert \bPsi^{(t+1)} - \bPsi^{(t)} \big\rVert < \delta$, 
for some $\delta > 0$.

\begin{algorithm}[htbp!]
\SetAlgoNoLine
\init{$\btheta^{(t)}, \bpi^{(t)}, \by$}

\kwEstep{}{Compute
$w_{ij}^{(t)} \leftarrow
	\frac{\pi_j^{(t)} \log \{ f(\by_i; \btheta_j^{(t)}) \} }
	     {\sum_{l=1}^K \pi_l^{(t)} \log \{ f(\by_i; \btheta_l^{(t)}) \} },
  \qquad i=1, 2, \dots, n; \ j = 1, 2, \dots, K.$
}{}
\kwMstep{}{
  $\bpi^{(t+1)} = \argmin_{\bpi} \left\{\sum_{i=1}^n \sum_{j=1}^Kw_{ij}^{(t)}  \log\pi_j - \varphi(\bpi)\right\}$
  
  $m \leftarrow 0$
  
  $\rho^{(t+1, 0)} \leftarrow \rho_0$
    
  \Repeat{$\norm{\bfeta^{(t+1, m+1)} - \bfeta^{(t+1, m)}} \leq \epsilon$}{
    $\rho^{(t+1, m+1)} \leftarrow \max\left\{\gamma^{-1} \rho^{(t+1, m)}, \rho_0\right\}$
    
	\Repeat{$\calQ(\bfeta^{(t, m+1)}; \bfeta^{(t+1, m)}) \leq \overline \calQ(\bfeta^{(t+1, m+1)}; \bfeta^{(t+1, m)})$}{
	  $\bfeta^{(t+1, m+1)} \leftarrow \argmin_{\bfeta} \overline \calQ(\bfeta; \bfeta^{(t+1, m)})$
	  
	 \textbf{if} $\calQ(\bfeta^{(t+1, m+1)}; \bfeta^{(t+1, m)}) > \overline \calQ(\bfeta^{(t+1, m)})$ \textbf{then} $\rho^{(t+1, m+1)} \leftarrow \gamma \rho^{(t+1, m+1)}$
	}
	Set $m \leftarrow m+1$
  }
}{}
\caption{$(t+1)$-th Iteration of the Modified EM Algorithm.}
\end{algorithm}

\subsection*{D.2. Implementation Specifications}
\label{sec:impel-specific} 
Our numerical solution is implemented in the C++ programming language, 
and is publicly available in the \texttt{GroupSortFuse} R package at 
\url{https://github.com/tmanole/GroupSortFuse}. 
Currently, this package implements the GSF method for multinomial mixtures, 
multivariate and univariate location-Gaussian mixtures,  
univariate Poisson mixtures, and mixtures of exponential distributions.

In what follows, we elaborate upon
the implementation specifications for 
the simulation study in
Section \ref{sec:simulation} of the paper.
In Section \ref{sim-setting}, we analyzed the performance of the GSF under multinomial mixtures 
and multivariate location-Gaussian mixture models,
with unknown common covariance matrix.
The data for the former two models was generated using the \texttt{mixtools} R package
\citep{BENAGLIA2009}.
We used the penalty $\varphi(\pi_1, \dots, \pi_K) = -C \sum_{j=1}^K \log\pi_j$
throughout, with $C = 3 \approx \log 20$ following the suggestion of \citet{CHEN1996}. 
We set the convergence criteria to $\epsilon = 10^{-5}$ and $\delta = 10^{-8}$,
and halted the modified 
EM algorithm and the nested PGD algorithm if they did not converge after 2500 and 1000 iterations, respectively. 
We initialized the EM algorithm for multinomial mixture models
using the MCMC algorithm described by \cite{GRENIER2016} for 100 iterations.
While 100 iterations may be insufficient for this algorithm to
approach the vicinity of a global maximum of the penalized
log-likelihood function, 
we found it yields reasonable performance in our simulations.
For the Gaussian mixtures, we used a binning method which is 
analogous to that of the \texttt{mixtools} package \citep{BENAGLIA2009}. 

The tuning parameter $\lambda$ for the penalty $r_\lambda$ was chosen by minimizing the BIC
criterion over a grid of candidate values $[\lambda_{\text{min}}, \lambda_{\text{max}}]$, 
as outlined in Section \ref{sec:simulation} of the paper. Based on our asymptotic results, we chose
$\lambda_{\text{max}} = n^{-1/4} \log n$ for the SCAD and MCP penalties.
For the ALasso penalty, we found that the rate $\frac{n^{-3/4}}{\log n}$ was too small in
practice and we instead used
$\lambda_{\text{max}} = n^{-\frac 1 2} \log n, n^{-0.35}$ for the Gaussian and multinomial
simulations respectively, which fall 
within the range discussed in item
(III), Section \ref{remark}
of the paper.
For the SCAD and MCP penalties, we chose $\lambda_{\text{min}} = 0.1, 0.4$ for
the Gaussian and multinomial simulations respectively, matching the lower bounds
used in the discrete and continuous mixture models of \cite{CHENKH2008}. 
For the ALasso penalty, 
we chose $\lambda_{\text{min}} = 0.01$ across both models. 

In Section \ref{sec:sim-merging} of the paper, 
we 
compared the GSF to the AIC, BIC, GSF-Hard, and MTM 
methods under location-Gaussian mixture models
with a known scale parameter. 
Our implementation of the AIC, BIC
and GSF-Hard in this section
was based on our own implementation of the EM 
algorithm, written in the Python 3.6 programming
language. The  simulations
were predominantly
performed on Linux machines
with Intel\textregistered \  Xeon\textregistered \
CPU E5-2690  (2.90GHz) processors.
We tuned the GSF-Hard            
using the BIC over the favourable range 
$\lambda_{\mathrm{min}} = 1.25 n^{- 1/4} \log n$
and
$\lambda_{\text{max}} = 1.5 n^{-1/ 4} \log n$.
A precise description of the GSF-Hard
method is given in Algorithm \ref{alg:gsf-hard}.

Throughout our simulations, we used
the cluster ordering $\alpha_{\bt}$
defined in equation 
\eqref{alpha-example}.  
We chose $\alpha_{\bt}(1)$ according to the following
heuristic procedure, which ensures
that $\alpha_{\bt}$ reduces to the natural
ordering on the real line when $d=1$. Define
$$(m_1,m_2) = \argmax_{1 \leq i < j \leq K} \norm{\bt_i-\bt_j}.$$
Let $\alpha^{(1)}_{\bt},\alpha_{\bt}^{(2)}$
denote the cluster orderings
given in equation \eqref{alpha-example},
respectively satisfying
$\alpha_{\bt}^{(1)}(1) = m_1$
and $\alpha_{\bt}^{(2)}(1) = m_2$.
We then define,
\[
\alpha_{\bt} = \argmin_{\phi \in \{\alpha_{\bt}^{(1)},
\alpha_{\bt}^{(2)}\}} 
\sum_{j=1}^{K-1} \norm{\bt_{\phi(j)} - \bt_{\phi(j+1)}}.
\]

\begin{algorithm}[H]
\SetAlgoNoLine
\init{$\widetilde G_n = \sum_{j=1}^K \tilde \pi_j \delta_{\tilde\btheta_j}, \lambda_n$, $K$.}

Let $\bphi_j = \tilde \btheta_{\alpha_{\tilde \btheta}(j)}$,
$p_j = \pi_{\alpha_{\tilde\btheta}(j)}$,\
$j=1, \dots, K$. 

Let $\calC = \{1\}$, $j=1$, $\hat K=1$.

\Repeat{$j = K-1$}{
    \textbf{if} $\norm{\bphi_j-\bphi_{j+1}} \leq \lambda_n$
    \textbf{then}
    
        \quad $\calC \leftarrow \calC \cup \{j+1\}$,

    \textbf{else} 
    
        \quad Let $\hat\bphi_j = \frac 1 {|\calC|} \sum_{i\in \calC}\bphi_i, \ \hp_j   = \sum_{i \in \calC} p_i$
            
        \quad $\calC \leftarrow \{j+1\}, \hat K \leftarrow \hat K + 1$
        
    \textbf{end}
        
    $j \leftarrow j+1$
 }
 
\textbf{Return:} $G_n = \sum_{k=1}^{\hat K} \hat p_k \delta_{\hat\bphi_k}$
 
\caption{\label{alg:gsf-hard} The GSF-Hard algorithm.}
\end{algorithm}

\section*{Supplement E: Additional Numerical Results}
In this section, we report the complete results of the simulations presented
in the form of plots in Section \ref{sec:simulation} of the paper. 
We also report 
a second real data analysis, for the Seeds data, 
based on a Gaussian mixture model.

\subsection*
{E.1. Simulation Results for the Multinomial Mixture Models}
In this section, we report the simulation
results for all the multinomial mixture Models 1-7 with 
$M=50$ and $M=35$, respectively, 
in Tables \ref{tab:supp_multinomial_1}-\ref{tab:supp_multinomial_7} and 
Tables \ref{tab:supp_multinomial_35_1}-\ref{tab:supp_multinomial_35_7}.

\begin{table}[!htbp]
\begin{tabular}{c c | c  c   c  c  c }
\firsthline
$n$ & $\hat K_n$ & AIC & BIC & SCAD &  MCP & AL\\
\hline
100 & 1 & .000 & .000 & .000 & .000 & .002 \\
    & \textbf{2} & \textbf{.876} & \textbf{.980} & \textbf{.922} & \textbf{.924} & \textbf{.968} \\
    & 3 & .116 & .018 & .078 & .076 & .030 \\
    & 4 & .008 & .002 & .000 & .000 & .000 \\
  \hline          
200 & 1 & .000 & .000 & .000 & .000 & .000 \\
    & \textbf{2} & \textbf{.864} & \textbf{.988} & \textbf{.936} & \textbf{.944} & \textbf{1.00} \\
    & 3 & .116 & .012 & .064 & .056 & .000 \\ 
    & 4 & .012 & .000 & .000 & .000 & .000 \\
    & 5 & .006 & .000 & .000 & .000 & .000 \\
    & 6 & .002 & .000 & .000 & .000 & .000 \\
\hline
400 & 1 & .000 & .000 & .000 & .000 & .000 \\
    & \textbf{2} & \textbf{.828} & \textbf{.994} & \textbf{.948} & \textbf{.952} & \textbf{1.00} \\
    & 3 & .136 & .006 & .052 & .048 & .000 \\     
    & 4 & .036 & .000 & .000 & .000 & .000 \\
\end{tabular}%
\centering
\captionsetup{justification=centering}
\caption{\label{tab:supp_multinomial_1} 
Order selection results for multinomial mixture Model 1 ($M=50$), 
with true order $K_0= 2$ indicated in {\bf bold} in the second column.  
For each method and sample size, the most frequently selected order 
is indicated in {\bf bold}.
} 
\end{table}

\begin{table}[!htbp]
\begin{tabular}{c c | c  c  c  c  c  c }
\firsthline
$n$ & $\hat K_n$ & AIC & BIC & SCAD &  MCP & AL\\
\hline
100 & 1 & .000 & .000 & .000 & .000 & .000 \\
    & 2 & .000 & .012 & .000 & .000 & .000 \\
    & \textbf 3 & \textbf{.808} & \textbf{.958} & \textbf{.642} & \textbf{.676} & \textbf{.898} \\
    & 4 & .152 & .030 & .338 & .304 & .096 \\
    & 5 & .034 & .000 & .020 & .020 & .006 \\
    & 6 & .006 & .000 & .000 & .000 & .000 \\
  \hline         
200 & 1 & .000 & .000 & .000 & .000 & .000 \\
    & 2 & .000 & .000 & .000 & .000 & .000 \\
    & \textbf 3 & \textbf{.804} & \textbf{.984} & \textbf{.666} & \textbf{.698} & \textbf{.980} \\ 
    & 4 & .146 & .016 & .312 & .288 & .020 \\
    & 5 & .040 & .000 & .022 & .014 & .000 \\
    & 6 & .010 & .000 & .000 & .000 & .000 \\    
\hline
400 & 1 & .000 & .000 & .000 & .000 & .000 \\
    & 2 & .000 & .000 & .000 & .000 & .000 \\  
    & \textbf 3 & \textbf{.836} & \textbf{.992} & \textbf{.698} & \textbf{.738} & \textbf{.996} \\     
    & 4 & .116 & .008 & .284 & .254 & .004 \\
    & 5 & .040 & .000 & .018 & .008 & .000 \\    
    & 6 & .008 & .000 & .000 & .00 & .000 \\  
\end{tabular}%
\centering
\captionsetup{justification=centering}
\caption{\label{tab:supp_multinomial_2} 
Order selection results for multinomial mixture Model 2 ($M=50$), 
with true order $K_0= 3$ indicated in {\bf bold} in the second column.  
For each method and sample size, the most frequently selected order 
is indicated in {\bf bold}.
}
\end{table}

\begin{table}[!htbp]
\begin{tabular}{c c | c  c  c  c  c }
\firsthline
$n$ & $\hat K_n$ & AIC & BIC &  SCAD &  MCP & AL\\
\hline
100 & 3 & .000 & .000 & .000 & .000 & .002 \\
    & \textbf{4} & \textbf{.686} & \textbf{.762} & \textbf{.788} & \textbf{.816} & \textbf{.962} \\
    & 5 & .260 & .218 & .194 & .174 & .034 \\
    & 6 & .046 & .018 & .018 & .010 & .002 \\
    & 7 & .008 & .002 & .000 & .000 & .000 \\
  \hline          
200 & 3 & .000 & .000 & .000 & .000 & .000 \\
    & \textbf{4} & \textbf{.690} & \textbf{.788} & \textbf{.800} & \textbf{.820} & \textbf{.978} \\
    & 5 & .260 & .200 & .180 & .162 & .022 \\ 
    & 6 & .044 & .012 & .020 & .018 & .000 \\ 
    & 7 & .006 & .000 & .000 & .000 & .000 \\
\hline
400 & 3 & .000 & .000 & .000 & .000 & .002 \\
    & \textbf{4} & \textbf{.702} & \textbf{.806} & \textbf{.824} & \textbf{.828} & \textbf{.986} \\
    & 5 & .260 & .186 & .158 & .154 & .010 \\  
    & 6 & .030 & .008 & .018 & .018 & .002 \\   
    & 7 & .006 & .000 & .000 & .000 & .000 \\             
    & 8 & .002 & .000 & .000 & .000 & .000 \\             
\end{tabular}%
\centering
\captionsetup{justification=centering}
\caption{\label{tab:supp_multinomial_3} 
Order selection results for multinomial mixture Model 3 ($M=50$), 
with true order $K_0= 4$ indicated in {\bf bold} in the second column.  
For each method and sample size, the most frequently selected order 
is indicated in {\bf bold}.
} 
\end{table}

\begin{table}[!htbp]
\begin{tabular}{c c | c c c c c c}
\firsthline
$n$ & $\hat K_n$ & AIC & BIC & SCAD & MCP & AL\\
\hline
100 & 3 & .000 & .000 & .000 & .000 & .002 \\
    & 4 & .026 & .140 & .082 & .076 & .148 \\
    & \textbf{5} & \textbf{.494} & \textbf{.546} & \textbf{.618} & \textbf{.660} & \textbf{.806} \\
    & 6 & .312 & .264 & .262 & .238 & .042 \\
    & 7 & .138 & .048 & .034 & .024 & .002 \\
    & 8 & .030 & .002 & .002 & .002 & .000 \\
    & 9 & .000 & .000 & .002 & .000 & .000 \\
  \hline          
200 & 3 & .000 & .000 & .000 & .000 & .000 \\
    & 4 & .000 & .012 & .014 & .012 & .040 \\
    & \textbf{5} & \textbf{.494} & \textbf{.556} & \textbf{.702} & \textbf{.724} & \textbf{.934} \\ 
    & 6 & .344 & .340 & .250 & .240 & .026 \\ 
    & 7 & .140 & .084 & .034 & .024 & .000 \\
    & 8 & .020 & .006 & .000 & .000 & .000 \\ 
    & 9 & .000 & .000 & .000 & .000 & .000 \\
    &10 & .002 & .002 & .000 & .000 & .000 \\
\hline
400 & 4 & .000 & .000 & .000 & .000 & .034 \\
    & \textbf{5} & \textbf{.468} & \textbf{.550} & \textbf{.764} & \textbf{.790} & \textbf{.960} \\  
    & 6 & .356 & .340 & .200 & .178 & .006 \\     
    & 7 & .150 & .102 & .034 & .030 & .000 \\
    & 8 & .022 & .004 & .002 & .002 & .000 \\    
    & 9 & .002 & .002 & .000 & .000 & .000 \\ 
    &10 & .002 & .002 & .000 & .000 & .000 \\
\end{tabular}%
\centering
\captionsetup{justification=centering}
\caption{\label{tab:supp_multinomial_4} 
Order selection results for multinomial mixture Model 4 ($M=50$), 
with true order $K_0= 5$ indicated in {\bf bold} in the second column.  
For each method and sample size, the most frequently selected order 
is indicated in {\bf bold}.
}
\end{table}

\clearpage

\begin{table}[H]
\setlength{\tabcolsep}{0.4em}
\begin{tabular}{c | l | c c c c c | l | c c c c c}
\multicolumn{1}{c|}{$n$} &
\multicolumn{6}{c}{Model 5} &\multicolumn{6}{|c}{Model 6}\\
\firsthline
& \multicolumn{1}{c|}{$\hat K_n$} & AIC & BIC & SCAD & MCP & ALasso & \multicolumn{1}{c|}{$\hat K_n$} & AIC & BIC & SCAD & MCP & ALasso \\
\cline{2-13}
100 & 4 & .016 & .212 & .086 & .088 & .144 & 4 & .052 & \textbf{.434} & .072 & .088 & \textbf{.332} \\
    & 5 & .328 & \textbf{.478} & \textbf{.486} & \textbf{.516} & \textbf{.544} & 5 & .172 & .352 & .272 & .282 & .306 \\
    & \textbf 6 & \textbf{.394} & .264 & .344 & .320 & .252 & 6 & \textbf{.296} & .170 & \textbf{.336} & \textbf{.330} & .216 \\
    & 7 & .200 & .042 & .074 & .066 & .054 & \textbf{7} & .286 & .040 & .276 & .270 & .126 \\
    & 8 & .062 & .004 & .010 & .010 & .006 & 8 & .194 & .004 & .044& .030 & .020 \\
  \hline
200 & 4 & .002 & .028 & .006 & .006 & .024 & 5 & .028 & .448 & .088 & .100 & .316 \\
    & 5 & .126 & \textbf{.474} & .286 & .306 & .408 & 6 & .134 & \textbf{.320} & .228 & .260 & .268 \\
    & \textbf 6 & \textbf{.380} & .390 & \textbf{.574} & \textbf{.572} & \textbf{.476} & \textbf{7} & .326 & .182 & \textbf{.544} & \textbf{.538} & \textbf{.368} \\
    & 7 & .300 & .094 & .120 & .108 & .072 & 8 & \textbf{.358} & .050 & .132 & .100 & .044 \\
    & 8 & .192 & .014 & .014 & .008 & .020 & 9  \scriptsize $\geq$ & .154 & .000 & .008 & .002 & .004 \\                                          
\hline
400 & 4 & .000 & .002 & .000 & .000 & .004 & 5 \scriptsize $\geq$ & .000 & .094 & .010 & .014 & .042 \\
    & 5 & .016 & .260 & .052 & .056 & .156 & 6 & .010 & .254 & .062 & .064 & .130 \\
    & \textbf 6 & \textbf{.384} & \textbf{.480} & \textbf{.740} & \textbf{.740} & \textbf{.716} & \textbf{7} & .342 & \textbf{.388} & \textbf{.694} & \textbf{.738} & \textbf{.738} \\
    & 7 & .336 & .214 & .190 & .178 & .110 & 8 & \textbf{.380} & .232 & .202 & .158 & .076 \\
    & 8 & .264 & .044 & .018 & .026 & .014 & 9 \scriptsize $\geq$ & .268 & .032 & .032 & .026 & .014 \\
                 \hline                         
\end{tabular}%
\centering
\captionsetup{justification=centering}
\caption{\label{tab:supp_multinomial_5_6} 
Order selection results for multinomial mixture Models 5 and 6 ($M=50$), 
with true orders $K_0=6$ and $7$ indicated in {\bf bold} in their corresponding 
second columns.   
For each method and sample size, the most frequently selected order 
is indicated in {\bf bold}.
}
\end{table}

\begin{table}[!htbp]
\begin{tabular}{c l | c c c c c c}
\firsthline
$n$ & $\hat K_n$ & AIC & BIC & SCAD & MCP & ALasso \\
\hline
100 & 5 \scriptsize$\geq$  & .000 & .106 & .000 & .002 & .022 \\
    & 6 & .034 & \textbf{.412} & .028 & .036 & .148 \\
    & 7 & .258 & .334 & .246 & .246 & .274 \\
    & \textbf 8 & \textbf{.392} & .144 & \textbf{.532} & \textbf{.554} & \textbf{.436} \\
    & 9  \scriptsize$\leq$ & .316 & .004 & .194 & .162 & .120 \\
  \hline         
200 & 6  \scriptsize$\geq$& .000 & .034 & .000 & .000 & .006 \\ 
    & 7 & .014 & .308 & .016 & .012 & .046 \\
    & \textbf 8 & \textbf{.496} & \textbf{.546} & \textbf{.626} & \textbf{.648} & \textbf{.674} \\
    & 9 & .340 & .112 & .304 & .294 & .222 \\    
    &10 \scriptsize$\leq$  & .150 & .000 & .054 & .046 & .052 \\
\hline
400 & 6  \scriptsize$\geq$& .000 & .000 & .000 & .000 & .000 \\  
    & 7 & .000 & .022 & .000 & .000 & .004 \\     
    & \textbf 8 & \textbf{.552} & \textbf{.638} & \textbf{.674} & \textbf{.698} & \textbf{.696} \\
    & 9 & .314 & .312 & .284 & .264 & .224 \\    
    &10 \scriptsize$\leq$  & .134 & .028 & .042 & .038 & .076 \\  
    \hline
\end{tabular}%
\centering
\captionsetup{justification=centering}
\caption{\label{tab:supp_multinomial_7} 
Order selection results for multinomial mixture Model 7 ($M=50$), 
with true order $K_0= 8$ indicated in {\bf bold} in the second column.  
For each method and sample size, the most frequently selected order 
is indicated in {\bf bold}.
}
\end{table}

\begin{table}[!htbp]
\begin{tabular}{c c | c  c  c  c  c }
\firsthline
$n$ & $\hat K_n$ & AIC & BIC &  SCAD &  MCP & AL\\
\hline
100 & 1 & .000 & .000 & .016 & .016 & .016 \\
	& \textbf 2 & \textbf{.842} & \textbf{.982} & \textbf{.962} & \textbf{.960} & \textbf{.942} \\ 
    & 3 & .146 & .018 & .022 & .024 & .042 \\ 
    & 4 & .012 & .000 & .000 & .000 & .000 \\ 
    & 6 & .000 & .000 & .000 & .000 & .000 \\  
  \hline       
200 & \textbf 2 & \textbf{.822} & \textbf{.982} & \textbf{.990} & \textbf{.990} & \textbf{.992} \\ 
    & 3 & .156 & .018 & .010 & .010 & .008 \\ 
    & 4 & .018 & .000 & .000 & .000 & .000 \\ 
    & 5 & .004 & .000 & .000 & .000 & .000 \\ 
    & 6 & .000 & .000 & .000 & .000 & .000 \\ 
\hline
400 & \textbf 2 & \textbf{.826} & \textbf{.994} & \textbf{1.00} & \textbf{1.00} & \textbf{.996} \\ 
    & 3 & .146 & .006 & .000 & .000 & .004 \\ 
    & 4 & .026 & .000 & .000 & .000 & .000 \\  
    & 5 & .002 & .000 & .000 & .000 & .000 \\      
    & 6 & .000 & .000 & .000 & .000 & .000 \\                       
\end{tabular}%
\centering
\captionsetup{justification=centering}
\caption{\label{tab:supp_multinomial_35_1} 
Order selection results for multinomial mixture Model 1 ($M=35$), 
with true order $K_0= 2$ indicated in {\bf bold} in the second column.  
For each method and sample size, the most frequently selected order 
is indicated in {\bf bold}.}
\end{table}

\begin{table}[!htbp]
\begin{tabular}{c c | c  c  c  c  c }
\firsthline
$n$ & $\hat K_n$ & AIC & BIC &  SCAD &  MCP & AL\\
\hline
100 & 2 & .032 & .322 & .004 & .016 & .114 \\
    & \textbf 3 & \textbf{.806} & \textbf{.672} & \textbf{.868} & \textbf{.888} & \textbf{.830} \\
    & 4 & .146 & .006 & .126 & .096 & .054 \\
    & 5 & .016 & .000 & .002 & .000 & .002 \\
    & 6 & .000 & .000 & .000 & .000 & .000 \\
  \hline
200 & 2 & .000 & .040 & .000 & .000 & .014 \\
    & \textbf 3 & \textbf{.794} & \textbf{.938} & \textbf{.850} & \textbf{.892} & \textbf{.956} \\
    & 4 & .172 & .022 & .150 & .106 & .030 \\
    & 5 & .026 & .000 & .000 & .002 & .000 \\
    & 6 & .008 & .000 & .000 & .000 & .000 \\
\hline
400 & \textbf 3 & \textbf{.796} & \textbf{.988} & \textbf{.868} & \textbf{.894} & \textbf{.990} \\
    & 4 & .162 & .012 & .130 & .106 & .010 \\
    & 5 & .038 & .000 & .002 & .000 & .000 \\
    & 6 & .004 & .000 & .000 & .000 & .000 \\
    & 7 & .000 & .000 & .000 & .000 & .000 \\

\end{tabular}%
\centering
\captionsetup{justification=centering}
\caption{\label{tab:supp_multinomial_35_2} 
Order selection results for multinomial mixture Model 2 ($M=35$), 
with true order $K_0= 3$ indicated in {\bf bold} in the second column.  
For each method and sample size, the most frequently selected order 
is indicated in {\bf bold}.}

\end{table}

\begin{table}[!htbp]
\begin{tabular}{c c | c  c  c  c  c }
\firsthline
$n$ & $\hat K_n$ & AIC & BIC &  SCAD &  MCP & AL\\
\hline
100 & 2 & .000 & .000 & .000 & .000 & .000 \\
    & 3 & .000 & .008 & .000 & .000 & .002 \\
    & \textbf 4 & \textbf{.716} & \textbf{.796} & \textbf{.850} & \textbf{.876} & \textbf{.876} \\
    & 5 & .230 & .182 & .142 & .116 & .108 \\
    & 6 & .050 & .014 & .008 & .008 & .014 \\
    & 7 & .004 & .000 & .000 & .000 & .000 \\
  \hline
200 & 2 & .000 & .000 & .000 & .000 & .000 \\
    & 3 & .000 & .000 & .000 & .000 & .000 \\
    & \textbf 4 & \textbf{.698} & \textbf{.826} & \textbf{.870} & \textbf{.894} & \textbf{.892} \\
    & 5 & .262 & .162 & .120 & .100 & .094 \\
    & 6 & .038 & .012 & .010 & .004 & .014 \\
    & 7 & .002 & .000 & .000 & .002 & .000 \\
\hline
400 & 3 & .000 & .000 & .000 & .000 & .000 \\
    & 4 & \textbf{.742} & \textbf{.860} & \textbf{.880} & \textbf{.888} & \textbf{.898} \\
    & 5 & .226 & .136 & .108 & .102 & .084 \\
    & 6 & .030 & .004 & .012 & .010 & .018 \\
    & 7 & .002 & .000 & .000 & .000 & .000 \\          
\end{tabular}%
\centering
\captionsetup{justification=centering}
\caption{\label{tab:supp_multinomial_35_3} 
Order selection results for multinomial mixture Model 3 ($M=35$), 
with true order $K_0= 4$ indicated in {\bf bold} in the second column.  
For each method and sample size, the most frequently selected order 
is indicated in {\bf bold}.}
\end{table}

\begin{table}[!htbp]
\begin{tabular}{c c | c  c  c  c  c }
\firsthline
$n$ & $\hat K_n$ & AIC & BIC &  SCAD &  MCP & AL\\
\hline
100 & 2 & .000 & .000 & .000 & .000 & .000 \\
    & 3 & .000 & .010 & .006 & .004 & .004 \\
    & 4 & .060 & .306 & .322 & .330 & .260 \\
    & \textbf 5 & \textbf{.534} & \textbf{.536} & \textbf{.508} & \textbf{.522} & \textbf{.628} \\
    & 6 & .300 & .144 & .150 & .130 & .096 \\
    & 7 & .094 & .004 & .012 & .012 & .008 \\
    & 8 & .012 & .000 & .002 & .002 & .002 \\
    & 9 & .000 & .000 & .000 & .000 & .002 \\
  \hline
200 & 2 & .000 & .000 & .000 & .000 & .000 \\
    & 3 & .000 & .000 & .000 & .000 & .000 \\
    & 4 & .004 & .102 & .102 & .096 & .058 \\
    & \textbf 5 & \textbf{.532} & \textbf{.596} & \textbf{.662} & \textbf{.706} & \textbf{.788} \\
    & 6 & .352 & .276 & .224 & .188 & .134 \\
    & 7 & .100 & .026 & .012 & .010 & .020 \\
    & 8 & .012 & .000 & .000 & .000 & .000 \\
\hline
400 & 3 & .000 & .000 & .000 & .000 & .000 \\
    & 4 & .000 & .016 & .004 & .004 & .018 \\
    & \textbf 5 & \textbf{.532} & \textbf{.622} & \textbf{.766} & \textbf{.790} & \textbf{.864} \\
    & 6 & .352 & .308 & .208 & .194 & .104 \\
    & 7 & .106 & .050 & .022 & .012 & .012 \\       
    & 8 & .010 & .004 & .000 & .000 & .002 \\             
\end{tabular}%
\centering
\captionsetup{justification=centering}
\caption{\label{tab:supp_multinomial_35_4} 
Order selection results for multinomial mixture Model 4 ($M=35$), 
with true order $K_0= 5$ indicated in {\bf bold} in the second column.  
For each method and sample size, the most frequently selected order 
is indicated in {\bf bold}.
}
\end{table}

\begin{table}[!htbp]
\begin{tabular}{c c | c  c  c  c  c }
\firsthline
$n$ & $\hat K_n$ & AIC & BIC &  SCAD &  MCP & AL\\
\hline
100 & 2 & .000 & .000 & .000 & .000 & .000 \\
    & 3 & .000 & .044 & .000 & .002 & .032 \\
    & 4 & .122 & .398 & .320 & .352 & .336 \\
    & 5 & \textbf{.456} & \textbf{.430} & \textbf{.530} & \textbf{.524} & \textbf{.482} \\
    & \textbf 6 & .334 & .126 & .134 & .106 & .138 \\
    & 7 & .074 & .002 & .016 & .016 & .012 \\
    & 8 & .014 & .000 & .000 & .000 & .000 \\
  \hline
200 & 2 & .000 & .000 & .000 & .000 & .000 \\
    & 3 & .000 & .000 & .000 & .000 & .002 \\
    & 4 & .014 & .180 & .110 & .116 & .170 \\
    & 5 & .344 & \textbf{.570} & \textbf{.546} & \textbf{.566} & \textbf{.550} \\
    & \textbf{6} & \textbf{.440} & .230 & .298 & .276 & .234 \\
    & 7 & .168 & .018 & .042 & .040 & .042 \\
    & 8 & .030 & .002 & .004 & .002 & .002 \\
    & 9 & .004 & .000 & .000 & .000 & .000 \\
\hline
400 & 3 & .000 & .000 & .000 & .000 & .000 \\
    & 4 & .000 & .046 & .016 & .022 & .000 \\
    & 5 & .146 & \textbf{.520} & .370 & .372 & .062 \\
    & \textbf 6 & \textbf{.442} & .374 & \textbf{.538} & \textbf{.532} & \textbf{.462} \\
    & 7 & .316 & .058 & .072 & .072 & .418 \\       
    & 8 & .082 & .002 & .004 & .002 & .050 \\             
    & 9 & .012 & .000 & .000 & .000 & .008 \\
    & 10& .002 & .000 & .000 & .000 & .000 \\       
\end{tabular}%
\centering
\captionsetup{justification=centering}
\caption{\label{tab:supp_multinomial_35_5} 
Order selection results for multinomial mixture Model 5 ($M=35$), 
with true order $K_0= 6$ indicated in {\bf bold} in the second column.  
For each method and sample size, the most frequently selected order 
is indicated in {\bf bold}.
}

\end{table}

\begin{table}[!htbp]
\begin{tabular}{c c | c  c  c  c  c }
\firsthline
$n$ & $\hat K_n$ & AIC & BIC &  SCAD &  MCP & AL\\
\hline
100 & 2 & .000 & .000 & .000 & .000 & .000 \\
    & 3 & .000 & .000 & .000 & .000 & .000 \\
    & 4 & .238 & \textbf{.640} & .328 & .362 & \textbf{.534} \\
    & 5 & \bf .320 & .298 & \textbf{.440} & \textbf{.420} & .276 \\
    & 6 & .270 & .056 & .182 & .180 & .134 \\
    & \textbf 7 & .146 & .004 & .048 & .038 & .054 \\
    & 8 & .024 & .002 & .002 & .000 & .002 \\
    & 9 & .002 & .000 & .000 & .000 & .000 \\
  \hline
200 & 2 & .000 & .000 & .000 & .000 & .000 \\
    & 3 & .000 & .000 & .000 & .000 & .000 \\
    & 4 & .014 & \textbf{.522} & .110 & .124 & \textbf{.404} \\
    & 5 & .344 & .354 & \textbf{.406} & \textbf{.408} & .306 \\
    & 6 & \textbf{.440} & .106 & .302 & .310 & .160 \\
    & \textbf 7 & .168 & .018 & .170 & .146 & .104 \\
    & 8 & .030 & .000 & .012 & .012 & .022 \\
    & 9 & .004 & .000 & .000 & .000 & .004 \\
\hline
400 & 3 & .000 & .000 & .000 & .000 & .000 \\
    & 4 & .000 & .218 & .012 & .012 & .166 \\
    & 5 & .024 & \textbf{.378} & .184 & .192 & .262 \\
    & 6 & .162 & .262 & .304 & .324 & .224 \\
    & \textbf 7 & \textbf{.412} & .134 & \textbf{.446} & \textbf{.428} & \textbf{.282} \\       
    & 8 & .292 & .008 & .046 & .044 & .066 \\             
    & 9 & .094 & .000 & .008 & .000 & .000 \\
    & 10& .016 & .000 & .000 & .000 & .000 \\       
\end{tabular}%
\centering
\captionsetup{justification=centering}
\caption{\label{tab:supp_multinomial_35_6} 
Order selection results for multinomial mixture Model 6 ($M=35$), 
with true order $K_0= 7$ indicated in {\bf bold} in the second column.  
For each method and sample size, the most frequently selected order 
is indicated in {\bf bold}.
}

\end{table}

\begin{table}[!htbp]
\begin{tabular}{c c | c  c  c  c  c }
\firsthline
$n$ & $\hat K_n$ & AIC & BIC &  SCAD &  MCP & AL\\
\hline
100 & 3 & .000 & .086 & .000 & .000 & .000 \\
    & 4 & .002 & \textbf{.590} & .012 & .006 & .038 \\
    & 5 & .080 & .290 & .172 & .180 & .346 \\
    & 6 & .308 & .030 & \textbf{.412} & \textbf{.412} & \textbf{.350} \\
    & 7 & \textbf{.384} & .004 & .296 & .284 & .198 \\
    & \textbf 8 & .170 & .000 & .102 & .114 & .050 \\
    & 9 & .044 & .000 & .004 & .004 & .018 \\
    &10 & .010 & .000 & .002 & .000 & .000 \\
    &11 & .002 & .000 & .000 & .000 & .000 \\
  \hline
200 & 3 & .000 & .000 & .000 & .000 & .000 \\
    & 4 & .000 & .000 & .000 & .000 & .000 \\
    & 5 & .000 & .194 & .014 & .014 & .060 \\
    & 6 & .038 & \textbf{.508} & .114 & .132 & .262 \\
    & 7 & .290 & .266 & .338 & .318 & .274 \\
    & \textbf 8 & \textbf{.408} & .032 & \textbf{.450} & \textbf{.468} & \textbf{.354} \\
    & 9 & .208 & .000 & .082 & .066 & .050 \\
    &10 & .052 & .000 & .002 & .002 & .000 \\
    &11 & .004 & .000 & .000 & .000 & .000 \\
\hline
400 & 3 & .000 & .000 & .000 & .000 & .000 \\
    & 4 & .000 & .000 & .000 & .000 & .000 \\
    & 5 & .000 & .002 & .000 & .000 & .006 \\
    & 6 & .000 & .108 & .002 & .004 & .046 \\
    & 7 & .016 & .400 & .044 & .048 & .094 \\       
    & \textbf 8 & \textbf{.572} & \textbf{.470} & \textbf{.700} & \textbf{.696} & \textbf{.644} \\             
    & 9 & .316 & .020 & .216 & .216 & .172 \\
    & 10& .082 & .000 & .036 & .034 & .032 \\    
    & 11& .014 & .000 & .002 & .002 & .006 \\          
\end{tabular}%
\centering
\captionsetup{justification=centering}
\caption{\label{tab:supp_multinomial_35_7} 
Order selection results for multinomial mixture Model 7 ($M=35$), 
with true order $K_0= 8$ indicated in {\bf bold} in the second column.  
For each method and sample size, the most frequently selected order 
is indicated in {\bf bold}.}

\end{table}

\clearpage

\subsection*
{E.2. Simulation Results for the Multivariate Location-Gaussian Mixture Models}
In this section, we report the simulation
results for all the multivariate Gaussian mixture Models (1.a, 1.b), (2.a, 2.b), 
(3.a, 3.b), (4.a, 4.b), (5.a, 5.b), in Tables \ref{tab:supp_normal_1}-\ref{tab:supp_normal_5}.

\begin{table}[!htbp]
\setlength{\tabcolsep}{0.36em}
\begin{tabular}{c | l | c c c c c | l | c c c c c}
\multicolumn{1}{c|}{$n$} &
\multicolumn{6}{c}{Model 1.a} &\multicolumn{6}{|c}{Model 1.b}\\
\firsthline
& \multicolumn{1}{c|}{$\hat K_n$} & AIC & BIC & SCAD & MCP & ALasso & \multicolumn{1}{c|}{$\hat K_n$} & AIC & BIC & SCAD & MCP & ALasso \\
\cline{2-13}
200 & 1 & .006 & .212 & .236 & .230 & .118 & 1 & .094 &  \textbf{.662} &  \textbf{.680} &  \textbf{.672} & .390 \\
    & \textbf 2 &  \textbf{.694} &  \textbf{.786} &  \textbf{.762} &  \textbf{.768} &  \textbf{.844} & \textbf 2 &  \textbf{.566} & .332 & .316 & .324 &  \textbf{.594} \\
    & 3 & .088 & .002 & .002 & .002 & .036 & 3 & .158 & .006 & .004 & .004 & .016 \\
    & 4 & .080 & .000 & .000 & .000 & .000 & 4 & .064 & .000 & .000 & .000 & .000 \\
    & 5 & .040 & .000 & .000 & .000 & .002 & 5 & .044 & .000 & .000 & .000 & .000 \\
    & 6 \tiny$\leq$ & .092 & .000 & .000 & .000 & .000 & 6 \tiny$\leq$  & .074 & .000 & .000 & .000 & .000 \\
  \hline
400 & 1 & .000 & .006 & .012 & .012 & .008 & 1 & .004 & .290 & .288 & .284 & .262 \\
    & \textbf 2 &  \textbf{.762} &  \textbf{.994} &  \textbf{.988} &  \textbf{.988} &  \textbf{.990} & 2 &  \textbf{.758} &  \textbf{.708} &  \textbf{.712} &  \textbf{.716} &  \textbf{.738} \\
    & 3 & .074 & .000 & .000 & .000 & .002 & 3 & .122 & .002 & .000 & .000 & .000 \\
    & 4 & .072 & .000 & .000 & .000 & .000 & 4 & .026 & .000 & .000 & .000 & .000 \\
    & 5 & .026 & .000 & .000 & .000 & .000 & 5 & .036 & .000 & .000 & .000 & .000 \\
    & 6 \tiny$\leq$  & .066 & .000 & .000 & .000 & .000 & 6 \tiny$\leq$  & .054 & .000 & .000 & .000 & .000 \\
\hline
600 & 1 & .000 & .000 & .002 & .002 & .000 & 1 & .002 & .098 & .086 & .088 & .084 \\
    &  \textbf 2 &  \textbf{.782} &  \textbf{1.00} &  \textbf{.998} &  \textbf{.998} &  \textbf{1.00} &  \textbf 2 &  \textbf{.808} &  \textbf{.896} &  \textbf{.914} &  \textbf{.912} &  \textbf{.912} \\
    & 3 & .084 & .000 & .000 & .000 & .000 & 3 & .106 & .006 & .000 & .000 & .004 \\
    & 4 & .062 & .000 & .000 & .000 & .000 & 4 & .030 & .000 & .000 & .000 & .000 \\
    & 5 & .028 & .000 & .000 & .000 & .000 & 5 & .024 & .000 & .000 & .000 & .000 \\
    & 6 \tiny$\leq$  & .044 & .000 & .000 & .000 & .000 & 6 \tiny$\leq$  & .030 & .000 & .000 & .000 & .000 \\
\hline
800 & 1 & .000 & .000 & .000 & .000 & .000 & 1 & .000 & .012 & .004 & .004 & .006 \\
    &  \textbf 2 &  \textbf{.844} &  \textbf{1.00} &  \textbf{1.00} &  \textbf{1.00} &  \textbf{1.00} &  \textbf 2 &  \textbf{.870} &  \textbf{.978} &  \textbf{.996} &  \textbf{.996} &  \textbf{.994} \\
    & 3 & .050 & .000 & .000 & .000 & .000 & 3 & .066 & .010 & .000 & .000 & .000 \\
    & 4 & .062 & .000 & .000 & .000 & .000 & 4 & .026 & .000 & .000 & .000 & .000 \\
    & 5 & .020 & .000 & .000 & .000 & .000 & 5 & .016 & .000 & .000 & .000 & .000 \\
    & 6 \tiny$\leq$  & .024 & .000 & .000 & .000 & .000 & 6 \tiny$\leq$  & .022 & .000 & .000 & .000 & .000 \\
\end{tabular}%

\centering
\captionsetup{justification=centering}
\caption{\label{tab:supp_normal_1} 
Order selection results for multivariate Gaussian mixture Models 1.a and 1.b, 
with true order $K_0= 2$ indicated in {\bf bold} in the second column.  
For each method and sample size, the most frequently selected order 
is indicated in {\bf bold}.}
\end{table}

\begin{table}[H]
\setlength{\tabcolsep}{0.4em}
\begin{tabular}{c | l | c c c c c | l | c c c c c}
\hline
\multicolumn{1}{c|}{$n$} &
\multicolumn{6}{c}{Model 2.a} &\multicolumn{6}{|c}{Model 2.b}\\
\firsthline
& \multicolumn{1}{c|}{$\hat K_n$} & AIC & BIC & SCAD & MCP & ALasso & \multicolumn{1}{c|}{$\hat K_n$} & AIC & BIC & SCAD & MCP & ALasso \\
\cline{2-13}
200 & 1        & .000 & .028 & .012 & .016 & .000 & 1 & .002 & .112 & .012 & .012 & .008 \\  
    & 2        & .080 & \textbf{.758} & \textbf{.728} & \textbf{.720} & .368 & 2 & .258 & \textbf{.808} & \textbf{.812} & \textbf{.800} & \textbf{.494} \\  
    & 3        & .200 & .180 & .180 & .192 & \textbf{.372} & 3 & .344 & .078 & .172 & .176 & .394 \\  
    & \textbf 4& \textbf{.270} & .034 & .078 & .072 & .240 & \textbf 4 & .176 & .002 & .004 & .012 & .102 \\  
    & 5        & .168 & .000 & .002 & .000 & .020 & 5 & .082 & .000 & .000 & .000 & .002 \\  
    & 6        & .100 & .000 & .000 & .000 & .000 & 6 & .044 & .000 & .000 & .000 & .000 \\  
    & 7  \scriptsize$\leq$        & .182 & .000 & .000 & .000 & .000 & 7 \scriptsize$\leq$ & .094 & .000 & .000 & .000 & .000 \\  
\hline
400 & 2  \scriptsize$\geq$      & .004 & \textbf{.416} & \textbf{.356} & \textbf{.376} & .304 & 2 \scriptsize$\geq$ & .080 & \textbf{.806} & \textbf{.628} & \textbf{.630} & \textbf{.590} \\
    & 3        & .096 & .332 & .294 & .340 & .258 & 3 & \textbf{.456} & .190 & .354 & .354 & .382 \\
    & \textbf 4& \textbf{.382} & .242 & .338 & .280 & \textbf{.416} & \textbf 4 & .250 & .004 & .018 & .016 & .028 \\
    & 5        & .212 & .010 & .012 & .004 & .022 & 5 & .078 & .000 & .000 & .000 & .000 \\
    & 6        & .130 & .000 & .000 & .000 & .000 & 6 & .044 & .000 & .000 & .000 & .000  \\
    & 7  \scriptsize$\leq$        & .176 & .000 & .000 & .000 & .000 & 7 \scriptsize$\leq$ & .092 & .000 & .000 & .000 & .000 \\
\hline                                              
600 & 2 \scriptsize$\geq$       & .000 & .160 & .136 & .164 & .150 & 2 \scriptsize$\geq$& .028 & \textbf{.662} & .448 & .456 & \textbf{.568} \\
    & 3        & .040 & .384 & .254 & .332 & .174 & 3 & \textbf{.454} & .326 & \textbf{.526} & \textbf{.514} & .396 \\
    & \textbf 4& \textbf{.472} & \textbf{.422} & \textbf{.604} & \textbf{.500} &  \textbf{.646} & \textbf 4 & .306 & .012 & .026 & .030 & .036 \\
    & 5        & .224 & .034 & .006 & .004 & .030 & 5 & .108 & .000 & .000 & .000 & .000 \\
    & 6        & .122 & .000 & .000 & .000 & .000 & 6 & .034 & .000 & .000 & .000 & .000  \\
    & 7  \scriptsize$\leq$        & .142 & .000 & .000 & .000 & .000 & 7 \scriptsize$\leq$ & .070 & .000 & .000 & .000 & .000 \\
\hline                                              
800 & 2  \scriptsize$\geq$       & .000 & .078 & .050 & .054 & .046 & 2 \scriptsize$\geq$ & .012 & .482 & .300 & .310 & \textbf{.488} \\
    & 3        & .008 & .308 & .178 & .308 & .136 & 3 & \textbf{.446} & \textbf{.502} & \textbf{.654} & \textbf{.662} & .470 \\
    & \textbf 4        & \textbf{.492} & \textbf{.518} & \textbf{.766} & \textbf{.636} & \textbf{.766} & 
    \textbf 4 & .322 & .016 & .044 & .028 & .042 \\
    & 5        & .234 & .096 & .006 & .002 & .052 & 5 & .102 & .000 & .002 & .000 & .000 \\
    & 6        & .134 & .000 & .000 & .000 & .000 & 6 & .044 & .000 & .000 & .000 & .000  \\
    & 7   \scriptsize$\leq$       & .132 & .000 & .000 & .000 & .000 & 7 \scriptsize$\leq$ & .074 & .000 & .000 & .000 & .000 \\
\hline
\end{tabular}%
\centering
\captionsetup{justification=centering}
\caption{\label{tab:supp_normal_2} 
Order selection results for multivariate Gaussian mixture Models 2.a and 2.b, 
with true order $K_0= 4$ indicated in {\bf bold} in the second column.  
For each method and sample size, the most frequently selected order 
is indicated in {\bf bold}.
}
\end{table}

\begin{table}[!htbp]
\setlength{\tabcolsep}{0.4em}
\begin{tabular}{c | l | c c c c c | l | c c c c c}
\multicolumn{1}{c|}{$n$} &
\multicolumn{6}{c}{Model 3.a} &\multicolumn{6}{|c}{Model 3.b}\\
\firsthline
& \multicolumn{1}{c|}{$\hat K_n$} & AIC & BIC & SCAD & MCP & ALasso & \multicolumn{1}{c|}{$\hat K_n$} & AIC & BIC & SCAD & MCP & ALasso \\
\cline{2-13}
200 & 1 & .000 & .000 & .000 & .000 & .000 & 1 & .000 & .296 & .184 & .168 & .004\\
    & 2 & .080 & \textbf{.854} & \textbf{.818} & \textbf{.798} & \textbf{.552} & 2 & .008 & .300 & \textbf{.546} & \textbf{.534} & .364\\
    & \textbf 3 & \textbf{.444} & .146 & .182 & .202 & \textbf{.430} & \textbf{3} & \textbf{.432} & .402 & .264 & .294 & .532\\
    & 4 & .112 & .000 & .000 & .000 & .018 & 4 & .160 & .002 & .006 & .004 & .076\\
    & 5 & .072 & .000 & .000 & .000 & .000 & 5 & .086 & .000 & .000 & .000 & .016\\
    & 6 \tiny$\leq$  & .292 & .000 & .000 & .000 & .000 & 6 \tiny$\leq$  & .314 & .000 & .000 & .000 & .008 \\
  \hline                                             
400 & 1 & .000 & .000 & .000 & .000 & .000 & 1 & .000 & .004 & .006 & .012 & .010\\
    & 2 & .000 & \textbf{.534} & .486 & .466 & .470 & 2 & .000 & .068 & .264 & .224 & .250\\
    & \textbf 3 & \textbf{.626} & .466 & \textbf{.512} & \textbf{.532} & \textbf{.528} & \textbf 3 & \textbf{.604} & \textbf{.922} & \textbf{.722} & \textbf{.748} & \textbf{.652} \\  
    & 4 & .124 & .000 & .002 & .002 & .002 & 4 & .130 & .006 & .004 & .012 & .048\\
    & 5 & .040 & .000 & .000 & .000 & .000 & 5 & .056 & .006 & .000 & .004 & .026\\
    & 6 \tiny$\leq$  & .210 & .000 & .000 & .000 & .000 & 6  \tiny$\leq$ & .210 & .000 & .004 & .000 & .014\\     
\hline                        
600 & 1 & .000 & .000 & .000 & .000 & .000 & 1 & .000 & .000 & .000 & .000 & .000\\
    & 2 & .000 & .216 & .202 & .200 & .212 & 2 & .000 & .002 & .098 & .110 & .162\\   
    & \textbf 3 & \textbf{.720} & \textbf{.784} & \textbf{.796} & \textbf{.800} & \textbf{.788} & \textbf 3 & \textbf{.688} & \textbf{.996} & \textbf{.856} & \textbf{.834} & \textbf{.734} \\      
    & 4 & .084 & .000 & .002 & .000 & .000 & 4 & .086 & .002 & .030 & .038 & .066\\
    & 5 & .050 & .000 & .000 & .000 & .000 & 5 & .066 & .000 & .010 & .006 & .012\\     
    & 6 \tiny$\leq$  & .146 & .000 & .000 & .000 & .000 & 6 \tiny$\leq$  & .160 & .000 & .006 & .012 & .026 \\   
\hline                        
800 & 1 & .000 & .000 & .000 & .000 & .000 & 1 & .000 & .000 & .000 & .000 & .002\\
    & 2 & .000 & .072 & .058 & .058 & .090 & 2 & .000 & .002 & .028 & .034 & .128\\   
    & \textbf 3 & \textbf{.752} & \textbf{.928} & \textbf{.940} & \textbf{.940} & \textbf{.910} & \textbf 3 & \textbf{.738} & \textbf{.996} & \textbf{.910} & \textbf{.896} & \textbf{.728} \\      
    & 4 & .080 & .000 & .002 & .002 & .000 & 4 & .088 & .002 & .048 & .050 & .068\\
    & 5 & .022 & .000 & .000 & .000 & .000 & 5 & .052 & .000 & .008 & .016 & .020\\     
    & 6 \tiny$\leq$  & .146 & .000 & .000 & .000 & .000 & 6 \tiny$\leq$  & .122 & .000 & .006 & .004 & .054\\       
\end{tabular}%
\centering
\captionsetup{justification=centering}
\caption{\label{tab:supp_normal_3} 
Order selection results for multivariate Gaussian mixture Models 3.a and 3.b, 
with true order $K_0= 3$ indicated in {\bf bold} in the second column.  
For each method and sample size, the most frequently selected order 
is indicated in {\bf bold}.
}

\end{table}

\begin{table}[!htbp]
\setlength{\tabcolsep}{0.4em}
\begin{tabular}{c | l | c c c c c | l | c c c c c}
\multicolumn{1}{c|}{$n$} &
\multicolumn{6}{c}{Model 4.a} &\multicolumn{6}{|c}{Model 4.b}\\
\firsthline
& \multicolumn{1}{c|}{$\hat K_n$} & AIC & BIC & SCAD & MCP & ALasso & \multicolumn{1}{c|}{$\hat K_n$} & AIC & BIC & SCAD & MCP & ALasso \\
\cline{2-13}
200 & 1 & .000 & .304 & .102 & .084 & .000 & 1 & .000 & .000 & .000 & .000 & .000 \\
    & 2 & .000 & .070 & .286 & .250 & .078 & 2 & .000 & .000 & .002 & .002 & .000 \\
    & 3 & .000 & \textbf{.312} & \textbf{.304} & \textbf{.344} & \textbf{.406} & 3 & .000 & .030 & .006 & .002 & .008 \\
    & 4 & .090 & .290 & .256 & .256 & .402 & 4 & .112 & \textbf{.798} & \textbf{.530} & \textbf{.560} & \textbf{.510} \\
    & \textbf{5} & .236 & .024 & .052 & .064 & .106 & \textbf{5} & .312 & .172 & .378 & .394 & .418 \\
    & 6 & .108 & .000 & .000 & .002 & .006 & 6 & .154 & .000 & .078 & .036 & .050 \\
    & 7 \tiny$\leq$  & \textbf{.566} & .000 & .000 & .000 & .002 & 7 \tiny$\leq$  & \textbf{.422} & .000 & .006 & .006 & .014 \\
 \hline                                    
400 & 1 & .000 & .002 & .000 & .000 & .002 & 1 & .000 & .000 & .000 & .000 & .000 \\
    & 2 & .000 & .008 & .008 & .014 & .014 & 2 & .000 & .000 & .000 & .000 & .000 \\
    & 3 & .000 & .102 & .124 & .124 & .138 & 3 & .000 & .000 & .000 & .000 & .000 \\
    & 4 & .028 & \textbf{.784} & \textbf{.472} & \textbf{.490} & \textbf{.434} & 4 & .028 & \textbf{.814} & .326 & .346 & .398 \\
    & \textbf{5} & .398 & .104 & .382 & .362 & .402 & \textbf{5} & \textbf{.458} & .182 & \textbf{.578} & \textbf{.590} & \textbf{.518} \\
    & 6 & .160 & .000 & .012 & .010 & .010 & 6 & .172 & .004 & .080 & .050 & .070 \\
    & 7 \tiny$\leq$  & \textbf{.414} & .000 & .002 & .000 & .000 & 7 \tiny$\leq$  & .342 & .000 & .016 & .014 & .014 \\
\hline                                     
600 & 1 & .000 & .000 & .000 & .000 & .000 & 1 & .000 & .000 & .000 & .000 & .000 \\
    & 2 & .000 & .000 & .000 & .000 & .002 & 2 & .000 & .000 & .000 & .000 & .000 \\
    & 3 & .000 & .000 & .008 & .008 & .026 & 3 & .000 & .000 & .000 & .000 & .000 \\
    & 4 & .016 & \textbf{.736} & .330 & .324 & .362 & 4 & .004 & \textbf{.636} & .230 & .216 & .302 \\
    & \textbf 5 & \textbf{.544} & .264 & \textbf{.654} & \textbf{.652} & \textbf{.584} & \textbf{5} & \textbf{.604} & .362 & \textbf{.668} & \textbf{.714} & \textbf{.616} \\
    & 6 & .150 & .000 & .008 & .016 & .024 & 6 & .152 & .002 & .076 & .062 & .068 \\
    & 7 \tiny$\leq$  & .290 & .000 & .000 & .000 & .002 & 7 \tiny$\leq$  & .240 & .000 & .026 & .008 & .014 \\
\hline                                     
800 & 1 & .000 & .000 & .000 & .000 & .000 & 1 & .000 & .000 & .000 & .000 & .000 \\
    & 2 & .000 & .000 & .000 & .000 & .000 & 2 & .000 & .000 & .000 & .000 & .000 \\
    & 3 & .000 & .000 & .002 & .000 & .012 & 3 & .000 & .000 & .000 & .000 & .000 \\
    & 4 & .000 & \textbf{.518} & .174 & .190 & .216 & 4 & .000 & .392 & .136 & .146 & .272 \\
    & \textbf 5 & \textbf{.662} & .482 & \textbf{.808} & \textbf{.788} & \textbf{.752} & \textbf 5 & \textbf{.670} & \textbf{.606} & \textbf{.768} & \textbf{.802} & \textbf{.672} \\
    & 6 & .144 & .000 & .012 & .018 & .016 & 6 & .162 & .002 & .070 & .040 & .046 \\
    & 7 \tiny$\leq$  & .194 & .000 & .004 & .004 & .004 & 7 \tiny$\leq$  & .168 & .000 & .026 & .012 & .010 \\
\end{tabular}%
\centering
\captionsetup{justification=centering}
\caption{\label{tab:supp_normal_4} 
Order selection results for multivariate Gaussian mixture Models 4.a and 4.b, 
with true order $K_0= 5$ indicated in {\bf bold} in the second column.  
For each method and sample size, the most frequently selected order 
is indicated in {\bf bold}.
}
\end{table}

\clearpage

\begin{table}[H]
\setlength{\tabcolsep}{0.4em}
\begin{tabular}{c | l | c c c c c | l | c c c c c}
\hline
\multicolumn{1}{c|}{$n$} &
\multicolumn{6}{c}{Model 5.a} &\multicolumn{6}{|c}{Model 5.b}\\
\firsthline
& \multicolumn{1}{c|}{$\hat K_n$} & AIC & BIC & SCAD & MCP & ALasso & \multicolumn{1}{c|}{$\hat K_n$} & AIC & BIC & SCAD & MCP & ALasso \\
\cline{2-13}
200 & 1 & .000 & .006 & .000 & .000 & .000 & 1 & .000 & .154 & .052 & .046 & .000 \\
    & 2 & .000 & \textbf{.914} & \textbf{.704} & \textbf{.652} & .100 & 2 & .000 & .006 & .046 & .028 & .000 \\
    & 3 & .012 & .078 & .274 & .312 & \textbf{.540} & 3 & .000 & .184 & \textbf{.370} & .356 & .034 \\
    & 4 & .090 & .002 & .020 & .034 & .290 & 4 & .000 & .244 & .186 & .166 & .214 \\
    & \textbf 5 & .200 & .000 & .002 & .002 & .064 & \textbf 5 & .280 & \textbf{.412} & .324 & \textbf{.374} & \textbf{.642} \\
    & 6 & .140 & .000 & .000 & .000 & .004 & 6 & .146 & .000 & .022 & .030 & .082 \\
    & 7 \scriptsize$\leq$  & \textbf{.558} & .000 & .000 & .000 & .002 & 7 \scriptsize$\leq$ & \textbf{.574} & .000 & .000 & .000 & .028 \\
 \hline
400 & 1 & .000 & .000 & .000 & .000 & .000 & 1 & .000 & .000 & .000 & .000 & .000 \\
    & 2 & .000 & \textbf{.508} & .382 & .358 & .310 & 2 & .000 & .000 & .000 & .000 & .000 \\
    & 3 & .000 & .358 & \textbf{.504} & \textbf{.510} & \textbf{.530} & 3 & .000 & .000 & .030 & .024 & .004 \\
    & 4 & .010 & .100 & .074 & .086 & .108 & 4 & .000 & .068 & .042 & .038 & .046 \\
    & \textbf 5 &.408 & .034 & .040 & .046 & .052 & \textbf 5 & \textbf{.458} & \textbf{.910} & \textbf{.814} & \textbf{.840} & \textbf{.776} \\
    & 6 & .138 & .000 & .000 & .000 & .000 & 6 & .158 & .022 & .084 & .084 & .140 \\
    & 7 \scriptsize$\leq$  & \textbf{.444} & .000 & .000 & .000 & .000 & 7  \scriptsize$\leq$& .384 & .000 & .030 & .014 & .034 \\
\hline
600 &2 \scriptsize$\geq$& .000 & .090 & .092 & .086 & .134 &2\ \scriptsize$\geq$& .000 & .000 & .000 & .000 & .000 \\
    & 3 & .000 & \textbf{.392} & \textbf{.496} & \textbf{.470} & \textbf{.468} & 3 & .000 & .000 & .002 & .000 & .000 \\
    & 4 & .000 & .278 & .106 & .096 & .132 & 4 & .000 & .002 & .010 & .012 & .028 \\
    & \textbf 5 & \textbf{.550} & .240 & .306 & .346 & .262 & \textbf 5 & \textbf{.604} & \textbf{.964} & \textbf{.834} & \textbf{.874} & \textbf{.828} \\
    & 6 & .146 & .000 & .000 & .002 & .004 & 6 & .120 & .032 & .114 & .080 & .086 \\
    & 7 \scriptsize$\leq$  & .304 & .000 & .000 & .000 & .000 & 7 & .276 & .002 & .040 & .034 & .058 \\
\hline
800 & 2 \scriptsize$\geq$ & .000 & .002 & .006 & .002 & .012 &2\ \scriptsize$\geq$& .000 & .000 & .000 & .000 & .000 \\
    & 3 & .000 & .122 & .218 & .188 & .224 & 3 & .000 & .000 & .000 & .000 & .000 \\
    & 4 & .000 & .228 & .084 & .058 & .082 & 4 & .000 & .000 & .000 & .000 & .012 \\
    & \textbf 5 & \textbf{.664} & \textbf{.648} & \textbf{.682} & \textbf{.744} & \textbf{.676} & \textbf 5 & \textbf{.718} & \textbf{.980} & \textbf{.830} & \textbf{.858} & \textbf{.824} \\
    & 6 & .116 & .000 & .010 & .008 & .006 & 6 & .104 & .020 & .128 & .088 & .120 \\
    & 7 \scriptsize$\leq$  & .220 & .000 & .000 & .000 & .000 & 7\ \scriptsize$\leq$  & .178 & .000 & .042 & .054 & .044 \\
\hline
\end{tabular}%
\centering
\captionsetup{justification=centering}
\caption{\label{tab:supp_normal_5} 
Order selection results for multivariate Gaussian mixture Models 5.a and 5.b, 
with true order $K_0= 5$ indicated in {\bf bold} in the second column.  
For each method and sample size, the most frequently selected order 
is indicated in {\bf bold}.
} 
\end{table}

\clearpage

\subsection*{E.3. Simulation Results for Section \ref{sensit}}
In this section, we report detailed simulation results
for the sensitivity analyses performed
in Section \ref{sensit} of the paper.

 \begin{table}[H]
{\small
\begin{tabular}{c | c c c c c c c c c c c c c }
\firsthline
$\hat K_n$ & 2 & 3 & 4 & 5 & 6 & 7 & 8 & 9 & 10 & 11 & 12 & 13 \\
\hline
2 & \bf 1.000 & .0000 & .0000 & .0000 & .0000 & .0000 & .0000 & .0000 & .0000 & .0000 & .0000 & .0000 \\
3 & .0000 &  \bf 1.000 & .0625 & .0000 & .0000 & .0000 & .0000 & .0000 & .0000 & .0000 & .0000 & .0000 \\
\bf  4 & .0000 & .0000 &  \bf .9375 &  \bf .9375 &  \bf .9000 &  \bf .8625 &  \bf .9125 &  \bf .8500 & \bf  .8125 &  \bf .8375 &  \bf .8125 &  \bf .7875 \\
5 & .0000 & .0000 & .0000 & .0625 & .0875 & .1000 & .0750 & .1375 & .1625 & .1250 & .1500 & .1875 \\
6 & .0000 & .0000 & .0000 & .0000 & .0125 & .0375 & .0125 & .0125 & .0250 & .0375 & .0375 & .0250 \\
\hline
$\hat K_n$ & 14 & 15 & 16 & 17 & 18 & 19 & 20 & 21 & 22 & 23 & 24 & 25\\
\hline
\bf 4 & \bf  .7625 &  \bf .7625 &  \bf .8125 &  \bf .7250 & \bf  .7625 &  \bf .7375 & \bf  .7000 &  \bf .8125 & \bf  .7375 & \bf  .7500 &  \bf .6375 &  \bf .7000 \\
5 & .2125 & .2000 & .1500 & .2250 & .1750 & .1875 & .2625 & .1500 & .2375 & .1875 & .3125 & .2750 \\
6 & .0250 & .0375 & .0250 & .0375 & .0625 & .0625 & .0250 & .0375 & .0250 & .0625 & .0500 & .0250 \\
7 & .0000 & .0000 & .0125 & .0125 & .0000 & .0125 & .0125 & .0000 & .0000 & .0000 & .0000 & .0000 \\
    \end{tabular}}%
\centering
\captionsetup{justification=centering}
\caption{\label{tab:supp_multiResults1} Sensitivity Analysis
for Multinomial Model 3
with true order $K_0= 4$ indicated in {\bf bold} in the first column,
and with respect to the bounds $K=2, \dots, 25$.  
For each bound, the most frequently selected order 
is indicated in {\bf bold}.}
\end{table}

\begin{table}[H]
{\small
\begin{tabular}{c |  cc c c c c c c c c c c c c }
\firsthline
$\hat K_n$ & 2 & 3 & 4 & 5 & 6 & 7 & 8 & 9 & 10 & 11 & 12 & 13 \\
\hline
2 &  \bf 1.000 & .0000 & .0000 & .0000 & .0000 & .0000 & .0000 & .0000 & .0000 & .0000 & .0000 & .0000 \\
3 & .0000 &  \bf 1.000 & .0000 & .0000 & .0000 & .0000 & .0000 & .0000 & .0000 & .0000 & .0000 & .0000  \\
4 & .0000 & .0000 & \bf  1.000 & .0750 & .0250 & .0000 & .0000 & .0000 & .0000 & .0000 & .0000 & .0000  \\
5 & .0000 & .0000 & .0000 &  \bf .9250 & .3750 & .1125 & .1000 & .0750 & .0250 & .0625 & .0625 & .1125  \\
\bf 6 & .0000 & .0000 & .0000 & .0000 &  \bf .6000 &  \bf .8250 & \bf  .7250 &  \bf .7500 & \bf  .7500 & \bf  .7250 & \bf  .7000 &  \bf .6750  \\ 
7 & .0000 & .0000 & .0000 & .0000 & .0000 & .0625 & .1750 & .1750 & .1875 & .1875 & .2000 & .1875  \\
8 & .0000 & .0000 & .0000 & .0000 & .0000 & .0000 & .0000 & .0000 & .0375 & .0250 & .0250 & .0250  \\
9 & .0000 & .0000 & .0000 & .0000 & .0000 & .0000 & .0000 & .0000 & .0000 & .0000 & .0125 & .0000  \\
\hline
$\hat K_n$ & 14 & 15 & 16 & 17 & 18 & 19 & 20 & 21 & 22 & 23 & 24 & 25\\
\hline
5 & .0750 & .0500 & .0750 & .0875 & .0625 & .0750 & .0500 & .0625 & .0500 & .0125 & .0625 & .0250 \\      
\bf 6 &  \bf .6625 &  \bf .7125 &  \bf .6500 & \bf  .6375 & \bf  .6500 &  \bf .6375  &\bf  .6500 & \bf  .6500 &  \bf .6500 &  \bf .6375 & \bf  .6750 & . \bf 6750 \\      
7 & .2375 & .1875 & .2250 & .2125 & .2250 & .2750 & .2625 & .2500 & .2625 & .3125 & .2000 & .2500 \\       
8 & .0250 & .0500 & .0500 & .0625 & .0625 & .0125 & .0250 & .0125 & .0375 & .0375 & .0500 & .0375 \\      
9 & .0000 & .0000 & .0000 & .0000 & .0000 & .0000 & .0125 & .0250 & .0000 & .0000 & .0000 & .0125 \\      
10 & .0000 & .0000 & .0000 & .0000 & .0000 & .0000 & .0000 & .0000 & .0000 & .0000 & .0125 & .0000 \\ 
\end{tabular}}%
\centering
\captionsetup{justification=centering}
\caption{\label{tab:supp_multiResults2} Sensitivity Analysis
for Multinomial Model 5 
with true order $K_0= 4$ indicated in {\bf bold} in the first column,
and with respect to the bounds $K=2, \dots, 25$.  
For each bound, the most frequently selected order 
is indicated in {\bf bold}.}
\end{table}

\begin{table}[H]
{\small
\begin{tabular}{c| c  c c c c c c c c c c c c c }
\firsthline
$\hat K_n$ &2 & 3 & 4 & 5 & 6 & 7 & 8 & 9 & 10 & 11 & 12 & 13 \\
\hline
1 & .1500 & .0375 & .0000 & .0000 & .0000 & .0000 & .0000 & .0000 & .0000 & .0000 & .0000 & .0000 \\
2 &  \bf .8500 & .1875 & .2125 & .2250 & .2375 & .2500 & .2500 & .2750 & .2625 & .2375 & .2375 & .2500 \\ 
\bf 3 & .0000 &  \bf .7750 & \bf  .7875 & \bf  .7625 &  \bf .7500 &  \bf .7375 &  \bf .7500 &  \bf .7125 &  \bf .7375 &  \bf .7625 & \bf  .7625 &  \bf .7375 \\ 
4 & .0000 & .0000 & .0000 & .0125 & .0125 & .0125 & .0000 & .0125 & .0000 & .0000 & .0000 & .0125 \\     
\hline
$\hat K_n$ & 14 & 15 & 16 & 17 & 18 & 19 & 20 & 21 & 22 & 23 & 24 & 25\\
\hline
1 & .0000 & .0000 & .0000 & .0000 & .0000 & .0000 & .0000 & .0000 & .0000 & .0000 & .0000 & .0000 \\
2 & .2625 & .2500 & .2375 & .2500 & .2250 & .2500 & .2500 & .2750 & .2375 & .2375 & .2500 & .2625 \\ 
\bf 3 &  \bf .7375 &  \bf .7375 &  \bf .7500 &  \bf .7375 &  \bf .7625 &  \bf .7375 &  \bf .7375 &  \bf .7250 & \bf  .7625 &  \bf .7375 & \bf  .7500 &  \bf .7125 \\ 
4 & .0000 & .0125 & .0125 & .0125 & .0125 & .0125 & .0125 & .0000 & .0000 & .2500 & .0000 & .0250 \\ 
\end{tabular}}%
\centering
\captionsetup{justification=centering}
\caption{\label{tab:supp_normalResults1} Sensitivity Analysis
for Gaussian Model 3.a 
with true order $K_0= 4$ indicated in {\bf bold} in the first column,
and with respect to the bounds $K=2, \dots, 25$.  
For each bound, the most frequently selected order 
is indicated in {\bf bold}.}
\end{table}

\begin{table}[H]
{\small
\begin{tabular}{c | c c c c c c c c c c c c c c }
\firsthline
$\hat K_n$ & 2 & 3 & 4 & 5 & 6 & 7 & 8 & 9 & 10 & 11 & 12 & 13 \\
\hline
1 & .1125 & .0000 & .0000 & .0000 & .0000 & .0000 & .0000 & .0000 & .0000 & .0000 & .0000 & .0000 \\
 2 &  \bf .8875 & .0250 & .0000 & .0000 & .0000 & .0000 & .0000 & .0000 & .0000 & .0000 & .0000 & .0000 \\ 
3 & .0000 &  \bf .9750 & .1750 & .0250 & .0000 & .0000 & .0125 & .0000 & .0000 & .0000 & .0000 & .0000 \\ 
4 & .0000 & .0000 &  \bf .8250 & .2625 & .2500 & .2500 & .3375 & .2625 & .2500 & .2875 & .2875 & .3250 \\ 
\bf 5 & .0000 & .0000 & .0000 & \bf  .7125 &  \bf .7125 &  \bf .7125 &  \bf .6375 & \bf  .7000 & \bf  .7375 &  \bf .7125 &  \bf .7000 & \bf  .6500 \\ 
6 & .0000 & .0000 & .0000 & .0000 & .0375 & .0375 & .0000 & .0125 & .0125 & .0000 & .0125 & .0125 \\ 
7 & .0000 & .0000 & .0000 & .0000 & .0000 & .0000 & .0125 & .0125 & .0000 & .0000 & .0000 & .0125 \\ 
8 & .0000 & .0000 & .0000 & .0000 & .0000 & .0000 & .0000 & .0125 & .0000 & .0000 & .0000 & .0000 \\ 
\hline
$\hat K_n$ & 14 & 15 & 16 & 17 & 18 & 19 & 20 & 21 & 22 & 23 & 24 & 25\\
\hline
3 & .0125 & .0000 & .0125 & .0125 & .0000 & .0000 & .0125 & .0000 & .0000 & .0125 & .0250 & .0000 \\ 
4 & .3000 & .2875 & .2875 & .3250 & .3250 & .2500 & .2750 & .3000 & .2875 & .2500 & .2750 & .3625 \\ 
\bf 5 &  \bf .6750 &  \bf .7000 &  \bf .6875 &  \bf .6625 &  \bf .6625 &  \bf .7375 &  \bf .7125 & \bf  .6875 & \bf  .7125 & \bf  .7375 & \bf  .6750 &  \bf .6375 \\ 
6 & .0125 & .0125 & .0125 & .0000 & .0125 & .0125 & .0000 & .0125 & .0000 & .0000 & .0250 & .0000 \\     
\end{tabular}}%
\centering
\captionsetup{justification=centering}
\caption{\label{tab:supp_normalResults2} Sensitivity Analysis
for Gaussian Model 4.a 
with true order $K_0= 4$ indicated in {\bf bold} in the first column,
and with respect to the bounds $K=2, \dots, 25$.  
For bound, the most frequently selected order 
is indicated in {\bf bold}.}
\end{table}

\clearpage 

\subsection*{E.4. Simulation Results for Section \ref{sec:sim-merging}}
In this section, we report detailed   results
for the simulation study reported in Section \ref{sec:sim-merging}.

\begin{table}[H]\setlength\tabcolsep{3pt} 
\begin{tabular}{cc || m{1.2cm}m{1.2cm}| m{1.2cm}m{1.2cm} m{1.2cm} | m{1.2cm} | m{1.2cm}m{1.2cm} m{1.2cm}m{1.2cm}}
\firsthline
$n$ & $\hat K_n$ & AIC & BIC & GSF-SCAD & GSF-MCP & GSF-ALasso & GSF-Hard & MTM $c=.2$ & MTM $c=.25$ & MTM
$c=.3$ & MTM $c=.35$   \\
\hline
50 &  1 & .1625 & .3250 & .3375 & .3250 & .3250 & .3250 & .3750 & \bf.6000 & \bf.5875 & \bf.7500 \\
    & \bf 2 & \bf.7875 & \bf.6750 & \bf.6500 & \bf.6750 &\bf .6750 & \bf.6625 &\bf .3875 & .3625 & .3750 & .2375 \\
    & 3 & .0500 & .0000 & .0125 & .0000 & .0000 & .0125 & .2250 & .0375 & .0375 & .0125 \\
    & 4 & .0000 & .0000 & .0000 & .0000 & .0000 & .0000 & .0125 & .0000 & .0000 & .0000 \\
\hline                                                 
100 & 1 & .0000 & .1000 & .1250 & .1125 & .1250 & .2875 & .3875 &\bf .5250 & \bf.5125 &\bf .8125 \\
    & \bf 2 & \bf.9500 &\bf .9000 & \bf.8750 & \bf.8875 & \bf.8625 &\bf .7000 & \bf.5375 & .4750 & .4875 & .1875 \\
    & 3 & .0500 & .0000 & .0000 & .0000 & .0125 & .0125 & .0750 & .0000 & .0000 & .0000 \\
\hline                                                 
200 & 1 & .0000 & .0125 & .0125 & .0125 & .0125 & .1750 & .1375 & .1625 & .3500 &\bf .6000 \\
    & 2 &\bf .9250 &\bf .9875 & \bf.9875 & \bf.9875 &\bf .9875 &\bf .8125 & \bf.6875 &\bf .7500 &\bf .6500 & .4000 \\
    & 3 & .0750 & .0000 & .0000 & .0000 & .0000 & .0125 & .1750 & .0875 & .0000 & .0000 \\
\hline                                                 
400 & 1 & .0000 & .0000 & .0000 & .0000 & .0000 & .1125 & .0250 & .0750 & .2000 & .4875 \\
    &\bf 2 & \bf.9125 & \bf 1.000 & \bf 1.000 & \bf 1.000 & \bf 1.000 &  \bf .8625 &  \bf .5750 &  \bf .8375 &  \bf .7875 &  \bf .5125 \\
    & 3 & .0875 & .0000 & .0000 & .0000 & .0000 & .0250 & .4000 & .0875 & .0125 & .0000 \\
\hline                                                 
600 & 1 & .0000 & .0000 & .0000 & .0000 & .0000 & .0625 & .0250 & .0125 & .1125 & .3625 \\
    &  \bf 2 &  \bf .9125 &  \bf 1.000 & \bf  1.000 &  \bf 1.000 &  \bf 1.000 &  \bf .9125 &  .3625 &  \bf .8250 &  \bf .8875 &  \bf .6375 \\
    & 3 & .0875 & .0000 & .0000 & .0000 & .0000 & .0250 &  \bf .5875 & .1625 & .0000 & .0000 \\
    & 4 & .0000 & .0000 & .0000 & .0000 & .0000 & .0000 & .0250 & .0000 & .0000 & .0000 \\
\hline                                                 
800 & 1 & .0000 & .0000 & .0000 & .0000 & .0000 & .0500 & .0000 & .0250 & .1125 & .3000 \\
    &  \bf 2 &  \bf .9000 &  \bf 1.000 &  \bf 1.000 &  \bf 1.000 &  \bf 1.000 &  \bf .9125 &  \bf .5500 &  \bf .9125 &  \bf .8875 &  \bf .7000 \\
    & 3 & .1000 & .0000 & .0000 & .0000 & .0000 & .0375 & .4500 & .0625 & .0000 & .0000 \\
\end{tabular}%
\centering
\captionsetup{justification=centering}
\caption{Order selection results for multivariate Gaussian mixture Models 1.b
with known covariance matrix, and 
with true order $K_0= 2$ indicated in {\bf bold} in the second column.  
For each method and sample size, the most frequently selected order 
is indicated in {\bf bold}.}
\end{table}

\begin{table}[H]\setlength\tabcolsep{3pt} 
\begin{tabular}{cc || m{1.2cm}m{1.2cm}| m{1.2cm}m{1.2cm} m{1.2cm} | m{1.2cm} | m{1.2cm}m{1.2cm} m{1.2cm}m{1.2cm}}
\firsthline
$n$ & $\hat K_n$ & AIC & BIC & GSF-SCAD & GSF-MCP & GSF-ALasso & GSF-Hard & MTM $c=.2$ & MTM $c=.25$ & MTM
$c=.3$ & MTM $c=.35$   \\
\hline
50  & 1 & .0000 & .0000 & .0000 & .0000 & .0000 & .0125 & .0000 & .0000 & .0000 & .0000 \\
    & 2 & .0000 & .0125 & .0000 & .0000 & .0000 & .0750 & .0000 & .0000 & .0000 & .0000 \\
    & 3 & .4500 &\bf  .8500 & .3000 & \bf .4875 & \bf .5375 & \bf .5625 & .0000 &\bf  .6000 &\bf  .5875 & \bf 1.000 \\
    &\bf  4 & \bf \bf  .5000 & .1375 & \bf .6500 & \bf .4875 & .4250 & .3375 &\bf  1.000 & .4000 & .4125 & .0000 \\
    & 5 & .0500 & .0000 & .0500 & .0250 & .0375 & .0125 & .0000 & .0000 & .0000 & .0000 \\
\hline                                                            
100 & 1 & .0000 & .0000 & .0000 & .0000 & .0000 & .0125 & .0000 & .0000 & .0000 & .0000 \\
    & 2 & .0000 & .0000 & .0000 & .0000 & .0000 & .0375 & .0000 & .0000 & .0000 & .0000 \\
    & 3 & .0625 & \bf .5125 & .0250 & .1125 & .2750 & \bf .5500 & .0000 & .3000 & .2875 &\bf  .9875 \\
    & \bf 4 & \bf .8375 & .4875 & \bf .8750 & \bf .8375 & \bf .7125 & .3875 &\bf  1.000 & \bf .7000 & \bf .7125 & .0125 \\
    & 5 & .0875 & .0000 & .1000 & .0500 & .0125 & .0125 & .0000 & .0000 & .0000 & .0000 \\
    & 6 & .0125 & .0000 & .0000 & .0000 & .0000 & .0000 & .0000 & .0000 & .0000 & .0000 \\
\hline                                                            
200 & 1 & .0000 & .0000 & .0000 & .0000 & .0000 & .0125 & .0000 & .0000 & .0000 & .0000 \\
    & 2 & .0000 & .0000 & .0000 & .0000 & .0000 & .0250 & .0000 & .0000 & .0000 & .0000 \\
    & 3 & .0000 & .0500 & .0000 & .0250 & .0250 & \bf .5250 & .0000 & .0375 & .4125 & \bf .9250 \\
    & \bf 4 &\bf  .9125 & \bf .9500 & \bf .9000 & \bf .9250 & \bf .9625 & .4125 & \bf 1.000 & \bf .9625 &\bf  .5875 & .0750 \\
    & 5 & .0875 & .0000 & .0875 & .0500 & .0125 & .0250 & .0000 & .0000 & .0000 & .0000 \\
    & 6 & .0000 & .0000 & .0125 & .0000 & .0000 & .0000 & .0000 & .0000 & .0000 & .0000 \\
\hline                                                            
400 & 1 & .0000 & .0000 & .0000 & .0000 & .0000 & .0125 & .0000 & .0000 & .0000 & .0000 \\
    & 2 & .0000 & .0000 & .0000 & .0000 & .0000 & .0250 & .0000 & .0000 & .0000 & .0000 \\
    & 3 & .0000 & .0000 & .0000 & .0000 & .0000 & .3875 & .0000 & .1125 & .3750 & \bf .9250 \\
    & \bf  4 &\bf  .9375 &\bf  1.000 &\bf  .8875 & \bf .9875 &\bf  1.000 &\bf  .5500 &\bf .9875 &\bf  .8875 & \bf .6250 & .0750 \\
    & 5 & .0500 & .0000 & .0875 & .0125 & .0000 & .0250 & .0125 & .0000 & .0000 & .0000 \\
    & 6 & .0125 & .0000 & .0250 & .0000 & .0000 & .0000 & .0000 & .0000 & .0000 & .0000 \\
\hline                                                            
600 & 1 & .0000 & .0000 & .0000 & .0000 & .0000 & .0000 & .0000 & .0000 & .0000 & .0000 \\
    & 2 & .0000 & .0000 & .0000 & .0000 & .0000 & .0250 & .0000 & .0000 & .0000 & .0000 \\
    & 3 & .0000 & .0000 & .0000 & .0000 & .0000 & .3375 & .0000 & .0625 & .3500 & \bf .9000 \\
    &  \bf 4 & \bf .8375 & \bf 1.000 &\bf  .9625 &\bf  .9750 &\bf  1.000 & \bf .6000 & \bf 1.000 & \bf .9375 &\bf  .6500 & .1000 \\
    & 5 & .1250 & .0000 & .0250 & .0250 & .0000 & .0375 & .0000 & .0000 & .0000 & .0000 \\
    & 6 & .0375 & .0000 & .0125 & .0000 & .0000 & .0000 & .0000 & .0000 & .0000 & .0000 \\
\hline                                                            
800 & 1 & .0000 & .0000 & .0000 & .0000 & .0000 & .0000 & .0000 & .0000 & .0000 & .0000 \\
    & 2 & .0000 & .0000 & .0000 & .0000 & .0000 & .0000 & .0000 & .0000 & .0000 & .0000 \\
    & 3 & .0000 & .0000 & .0000 & .0000 & .0000 & .3500 & .0125 & .0250 & .3000 &\bf  .8500 \\
    &  \bf 4 &\bf  .8875 & \bf 1.000 &\bf  .8875 & \bf .9750 & \bf 1.000 & \bf .5875 & \bf .9625 & \bf .9750 & \bf .7000 & .1500 \\
    & 5 & .1125 & .0000 & .1000 & .0250 & .0000 & .0625 & .0250 & .0000 & .0000 & .0000 \\
    & 6 & .0000 & .0000 & .0125 & .0000 & .0000 & .0000 & .0000 & .0000 & .0000 & .0000 \\
\end{tabular}%
\centering
\captionsetup{justification=centering}
\caption{Order selection results for multivariate Gaussian mixture Models 2.a
with known covariance matrix, and 
with true order $K_0= 4$ indicated in {\bf bold} in the second column.  
For each method and sample size, the most frequently selected order 
is indicated in {\bf bold}.}\end{table}

\clearpage

\subsection*{E.5. Simulation Results for  
the Bayesian Method of Mixtures of Finite Mixtures}

In this section, we perform a simulation study comparing the GSF with 
the method of Mixtures of Finite Mixtures (MFM)  \citep{miller2018}, 
whereby in addition to 
the prior specification given in
\eqref{eq:overfitted_prior_pi}-\eqref{eq:overfitted_prior_theta}
of Section \ref{remark} of the paper, a prior is 
also placed on the mixture order.

We performed simulations for the MFM method using 
the publicly available implementation of \cite{miller2018}, for which two
models are available: 
\underline{l}ocation-\underline{s}cale Gaussian mixture models
with a \underline{c}onstrained diagonal
covariance matrix (MFM-LSC)
and \underline{l}ocation-\underline{s}cale Gaussian mixture models
with a \underline{u}nconstrained diagonal
covariance matrix (MFM-LSU).
We compare both of these models
to the GSF under location-Gaussian mixtures
with a common but unknown covariance matrix.
Due to the difference in the underlying model
presumed by these methods, our MFM
simulations are neither comparable to those in 
Section~\ref{sim-setting}, nor Section~\ref{sec:sim-merging}, thus we  
report the results in this supplement, 
Tables \ref{GSF-MFM1}-\ref{GSF-MFM5} below. 
  
The simulation results are based on 500 samples of sizes $n=200,400, 600, 800$. 
For the MFM, under each simulated sample, we consider 
the posterior mode as 
the estimated mixture order. 
For the MFM-LSC
 we ran the split-merge sampler
 described by \cite{miller2018}, Section~7.1.2., for conjugate priors,
for 100,000 iterations, 
including 10,000 burn-in iterations.
For the MFM-LSU, we used
the split-merge sampler described
by  \cite{miller2018}, 
Section 7.3.1., again 
based on 100,000 iterations.
We use 5,000 burn-in iterations and record
the full state of the chain only
once every 100 iterations to
reduce its memory burden.
The results in the tables below 
denote the relative frequency 
of estimated orders. We also
included the results of the GSF-SCAD from Section \ref{sim-setting}
for comparison, fitted under (LU) models.

From Table \ref{GSF-MFM1}, Model 1.a, we can see that for 
the sample size $n=200$,
in estimating the true order $K_0=2$, 
MFM-LSC outperforms the GSF-SCAD and MFM-LSU; and for sample sizes 
$n=400, 600, 800$, the three methods perform similarly. 
From Tables \ref{GSF-MFM2}--\ref{GSF-MFM5}, 
the MFM-LSU underestimates 
the true mixture order by one to three components across 
all the sample sizes $n=200,400,600,800$. On the other hand, 
the MFM-LSC that uses knowledge of the diagonal 
covariance matrix of the true underlying models, outperforms 
the GSF-SCAD in Models 2.a, 3.a, 5.a; and in  
Model 4a,
GSF-SCAD outperforms MFM-LSC.

\begin{table}[!htbp]
\begin{tabular}{c c | c c c c c c}
\firsthline
$n$ & $\widehat K_n$ & GSF-SCAD & MFM-LSC & MFM-LSU \\
\hline
200 & 1 & .236 & .000 & .400 \\
    & {\bf 2} & \bf .762 & \bf .976 & \bf .588 \\
    & 3 & .002 & .024 & .012 \\ 
\hline
400 & 1 & .012 & .000 & .018 \\
    & \textbf 2 & \bf .988 & \bf .978 & \bf .966 \\
    & 3 & .000 & .022 & .016 \\
\hline
600 & 1 & .002 & .000 & .000 \\
    & \textbf 2 & \bf .998 & \bf .986 & \bf .980 \\
    & 3 & .000 & .014 & .020 \\
\hline
800 & 1 & .000 & .000 & .000 \\
    & \textbf 2 & \bf 1.00 & \bf .986 & \bf .988 \\
    & 3 & .000 & .014 & .010  \\
    & 4 & .000 & .000 & .002  \\
\end{tabular}%
\centering
\captionsetup{justification=centering}
\caption{\label{GSF-MFM1} 
Order selection results for {\bf Gaussian Model 1.a}, with true order $K_0=2$ indicated 
in {\bf bold}  in the second column.
For each method and sample size, the most frequently selected order is 
indicated in {\bf bold}.}
\end{table}

\begin{table}[!htbp]
\begin{tabular}{c c | c c c c c c}
\firsthline
$n$ & $\widehat K_n$ & GSF-SCAD & MFM-LSC & MFM-LSU \\
\hline
200 & 1 & .012 & .000 & .072 \\
    & 2 & \bf .728 & .002 & \bf .918 \\
    & 3 & .180 & \bf .998 & .010 \\
    & \bf 4 & .078 & .000 & .000 \\
    & 5 & .002 & .000 & .000 \\
    & 6 & .000 & .000 & .000 \\
\hline
400 & 1 & .000 & .000 & .000 \\
    & 2 & \bf .356 & .000 & \bf .940 \\
    & 3 & .294 & \bf .632 & .060 \\
    & \bf 4 & .338 & .368 & .000 \\
    & 5 & .012 & .000 & .000 \\
    & 6 & .000 & .000 & .000 \\
\hline
600 & 1 & .000 & .000 & .000 \\
    & 2 & .136 & .000 & \bf .530 \\
    & 3 & .254 & .046 & .468 \\
    & \bf 4 & \bf .604 & \bf .954 & .002 \\
    & 5 & .006 & .000 & .000 \\
    & 6 & .000 & .000 & .000 \\
\hline
800 & 1 & .000 & .000 & .000 \\
    & 2 & .050 & .000 & .180 \\
    & 3 & .178 & .002 & \bf .818 \\
    & \bf 4 & \bf .766 & \bf .998 & .002 \\
    & 5 & .006 & .000 & .000 \\ 
\end{tabular}%
\centering
\captionsetup{justification=centering}
\caption{\label{GSF-MFM2} 
Order selection results for {\bf Gaussian Model 2.a}, with true order $K_0=4$ indicated 
in {\bf bold}  in the second column.
For each method and sample size, the most frequently selected order is 
indicated in {\bf bold}.} 
\end{table}

\begin{table}[!htbp]
\begin{tabular}{c c | c c c c c c}
\firsthline
$n$ & $\widehat K_n$ & GSF-SCAD & MFM-LSC & MFM-LSU \\
\hline
200 & 1 & .000 & .000 & .002 \\
    & 2 & \bf .818 & .412 & \bf .990 \\
    & \bf 3 & .182 & \bf .588 & .008 \\ 
\hline
400 & 1 & .000 & .000 & .000 \\
    & 2 & .486 & .024 & \bf .998 \\
    & \bf 3 & \bf .512 & \bf .974 & .002 \\
    & 4 & .002 & .002 & .000 \\
\hline
600 & 1 & .000 & .000 & .000 \\
    & 2 & .202 & .000 & \bf .998 \\
    & \bf 3 & \bf .796 & \bf .990 & .002 \\
    & 4 & .002 & .010 & .000 \\
\hline
800 & 1 & .000 & .000 & .000 \\
    & 2 & .058 & .000 & \bf 1.00 \\
    & \bf 3 & \bf .940 & \bf .998 & .000 \\
    & 4 & .002 & .002 & .000 \\
\end{tabular}%
\centering
\captionsetup{justification=centering}
\caption{\label{GSF-MFM3}
Order selection results for {\bf Gaussian Model 3.a}, with true order $K_0=3$ indicated 
in {\bf bold}  in the second column.
For each method and sample size, the most frequently selected order is 
indicated in {\bf bold}.} 
\end{table}

\begin{table}[!htbp]
\begin{tabular}{c c | c c c c c c}
\firsthline
$n$ & $\widehat K_n$ & GSF-SCAD & MFM-LSC & MFM-LSU \\
\hline
200 & 1 & .102 & .000 & .214 \\
    & 2 & .286 & .000 & \bf .786 \\
    & 3 & \bf .304 & .060 & .000 \\
    & 4 & .256 & \bf .888 & .000 \\
    & \bf 5 & .052 & .052 & .000 \\
\hline
400 & 1 & .000 & .000 & .000 \\
    & 2 & .008 & .000 & \bf .986 \\
    & 3 & .124 & .000 & .014 \\
    & 4 & \bf .472 & \bf .754 & .000 \\
    & \bf 5 & .382 & .244 & .000 \\
    & 6 & .012 & .002 & .000 \\
    & 7 & .002 & .000 & .000 \\
\hline
600 & 1 & .000 & .000 & .000 \\
    & 2 & .000 & .000 & \bf .764 \\
    & 3 & .008 & .000 & .236 \\
    & 4 & .330 & .494 & .000 \\
    & \bf 5 & \bf .654 & \bf .498 & .000 \\
    & 6 & .008 & .008 & .000 \\
\hline
800 & 1 & .000 & .000 & .000 \\
    & 2 & .000 & .000 & .408 \\
    & 3 & .002 & .000 & \bf .592 \\
    & 4 & .174 & .242 & .000 \\
    & \bf 5 & \bf .808 & \bf .750 & .000 \\
    & 6 & .012 & .008 & .000 \\
    & 7 & .004 & .000 & .000 \\
    & 8 & .000 & .000 & .000 \\
\end{tabular}%
\centering
\captionsetup{justification=centering}
\caption{\label{GSF-MFM4} 
Order selection results for {\bf Gaussian Model 4.a}, with true order $K_0=5$ indicated 
in {\bf bold}  in the second column.
For each method and sample size, the most frequently selected order is 
indicated in {\bf bold}.} 
\end{table}

\begin{table}[!htbp]
\begin{tabular}{c c | c c c c c c}
\firsthline
$n$ & $\widehat K_n$ & GSF-SCAD & MFM-LSC & MFM-LSU \\
\hline
200 & 1 & .000 & .000 & \bf .590 \\
    & 2 & \bf .704 & .010 & .410 \\
    & 3 & .274 & .396 & .000 \\
    & 4 & .020 & \bf .478 & .000 \\
    & \bf 5 & .002 & .116 & .000 \\
    & 6 & .000 & .000 & .000 \\
\hline
400 & 1 & .000 & .000 & .000 \\
    & 2 & .382 & .000 & \bf 1.00 \\
    & 3 & \bf .504 & .014 & .000 \\
    & 4 & .074 & .230 & .000 \\
    & \bf 5 & .040 & \bf .756 & .000 \\
    & 6 & .000 & .000 & .000 \\
\hline
600 & 1 & .000 & .000 & .000 \\
    & 2 & .092 & .000 & \bf .996 \\
    & 3 & \bf .496 & .000 & .004 \\
    & 4 & .106 & .022 & .000 \\
    & \bf 5 & .306 & \bf .974 & .000 \\
    & 6 & .000 & .004 & .000\\
\hline
800 & 1 & .000 & .000 & .000 \\
    & 2 & .006 & .000 & \bf .920 \\
    & 3 & .218 & .000 & .080 \\
    & 4 & .084 & .000 & .000 \\
    & \bf 5 & \bf .682 & \bf .994 & .000 \\
    & 6 & .010 & .006 & .000 \\
\end{tabular}%
\centering
\captionsetup{justification=centering}
\caption{\label{GSF-MFM5} 
Order selection results for {\bf Gaussian Model 5.a}, with true order $K_0=5$ indicated 
in {\bf bold}  in the second column. 
For each method and sample size, the most frequently selected order is 
indicated in {\bf bold}.}
\end{table}

 \clearpage
 
Finally, 
in Table \ref{tab-runtime} below we report 
average runtime (in seconds) per simulated sample (over 500 replications)
by the GSF-SCAD, MFM-LSC and MFM-LSU.

\begin{table}[!htbp]
\begin{tabular}{c c | c c c c c c}
\firsthline
 & &  & Average runtime (in seconds) &  \\
\hline
Model & $n$ & GSF-SCAD & MFM-LSC & MFM-LSU \\
\hline
1.a & 200 & 10.58 & 155.4 & 68.22 \\ 
    & 400 & 19.83 & 303.9 & 140.3 \\ 
    & 600 & 32.03 & 437.8 & 206.8 \\ 
    & 800 & 44.73 & 539.6 & 265.5 \\ 
\hline
2.a & 200 & 14.69 & 150.7 & 67.41 \\ 
    & 400 & 35.27 & 297.2 & 138.8 \\ 
    & 600 & 62.13 & 442.0 & 194.2 \\ 
    & 800 & 113.6 & 538.3 & 258.9 \\ 
\hline
3.a & 200 & 34.35 & 180.6 & 81.82 \\ 
    & 400 & 85.57 & 365.3 & 178.3 \\ 
    & 600 & 121.4 & 552.7 & 262.8 \\ 
    & 800 & 184.4 & 585.6 & 321.7 \\ 
\hline
4.a & 200 & 95.56 & 228.9 & 70.03 \\ 
    & 400 & 174.9 & 467.7 & 227.8 \\ 
    & 600 & 262.0 & 683.2 & 300.3 \\ 
    & 800 & 293.1 & 793.1 & 407.6 \\ 
\hline
5.a & 200 & 140.1 & 276.9 & 106.7 \\ 
    & 400 & 264.5 & 541.8 & 284.4 \\ 
    & 600 & 289.1 & 765.3 & 430.8 \\ 
    & 800 & 476.1 & 960.8 & 744.8 \\ 
\end{tabular}%
\centering
\captionsetup{justification=centering}
\caption{\label{tab-runtime}
Average runtime in seconds per-simulated sample (over 500 replications) 
for Gaussian Models 1.a-5.a.} 
\end{table}

\subsection*{E.6. Analysis of the Seeds Data}
We consider the seeds data of \cite{CHARYTANOWICZ2010}, in which 7 
geometric parameters were measured by X-Ray in 210 seeds. The seeds belong to three 
varieties: Kama, Rosa and Canadian. The number of seeds from each variety is 70, 
suggesting that the data may be modelled by a balanced mixture of three components. 
\cite{ZHAO2015} fitted a Gaussian mixture model to a standardization of this data, 
since its seven coordinates do not have the same units of measurement.
\cite{CHARYTANOWICZ2010} analyzed a projection of the data on its first two principal 
components using a gradient clustering algorithm, 
and \cite{LEE2013} fitted various mixtures of skewed distributions to two of the 
seven geometric parameters of the 
seeds, namely their asymmetry and perimeter. 
We used the GSF method 
to fit a bivariate Gaussian mixture model in mean, with common
but unknown covariance matrix, based on both of the latter approaches. 
In both cases, all three penalties of the GSF 
resulted in $\hat K= 3$ components. 
In what follows, we report the details of our analysis based on the 
approach of \cite{LEE2013}, namely by only fitting a mixture to 
the asymmetry and perimeter coordinates of the data.
A plot of this dataset is shown in Figure \ref{fig:seedData}.(a).

\begin{figure}[H] 
\begin{center}
	\begin{minipage}{0.45\textwidth}
  	  \includegraphics[width=\linewidth]{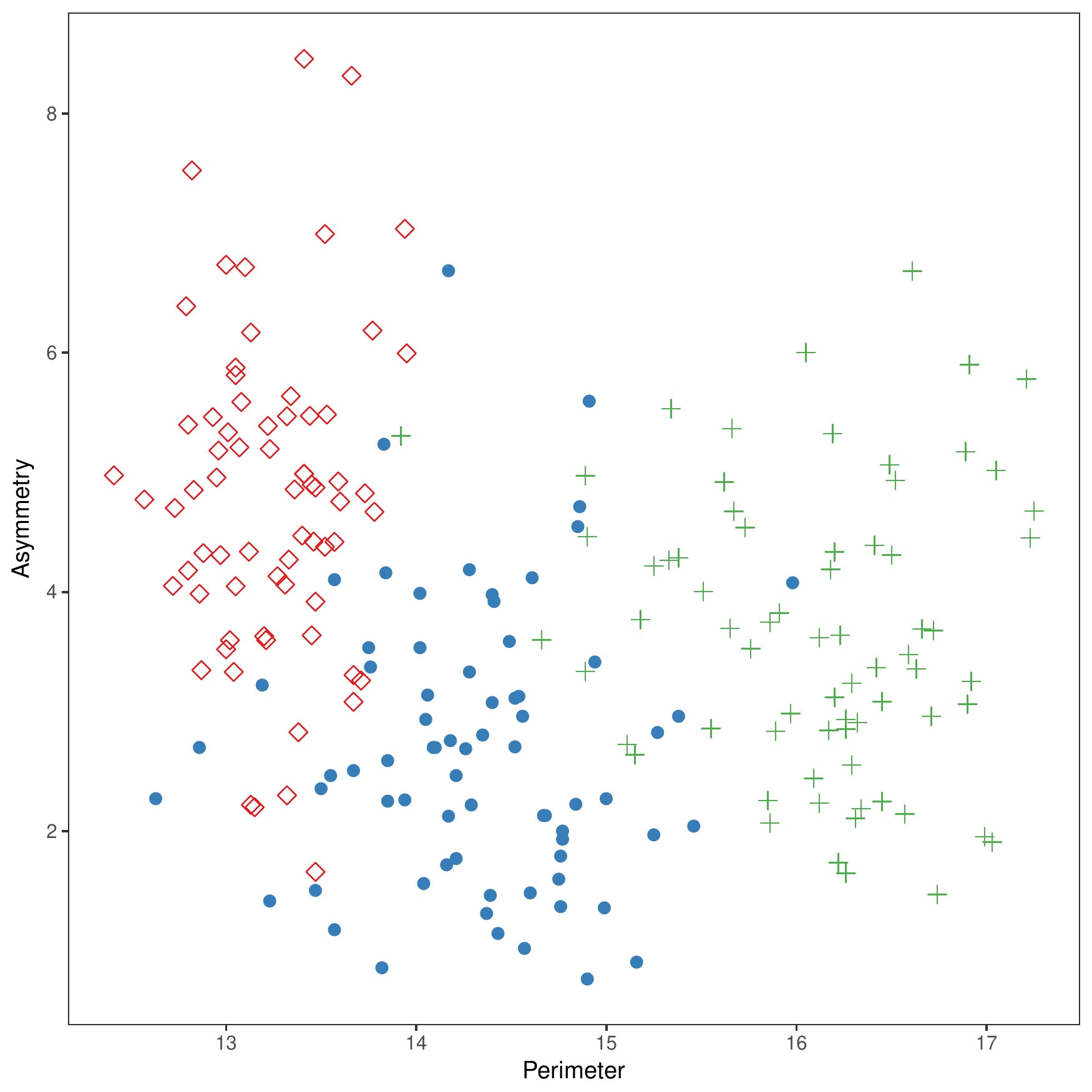}
      \caption*{(a) True Clustering.} 
    \end{minipage}
	\begin{minipage}{0.45\textwidth}
	  \includegraphics[width=\linewidth]{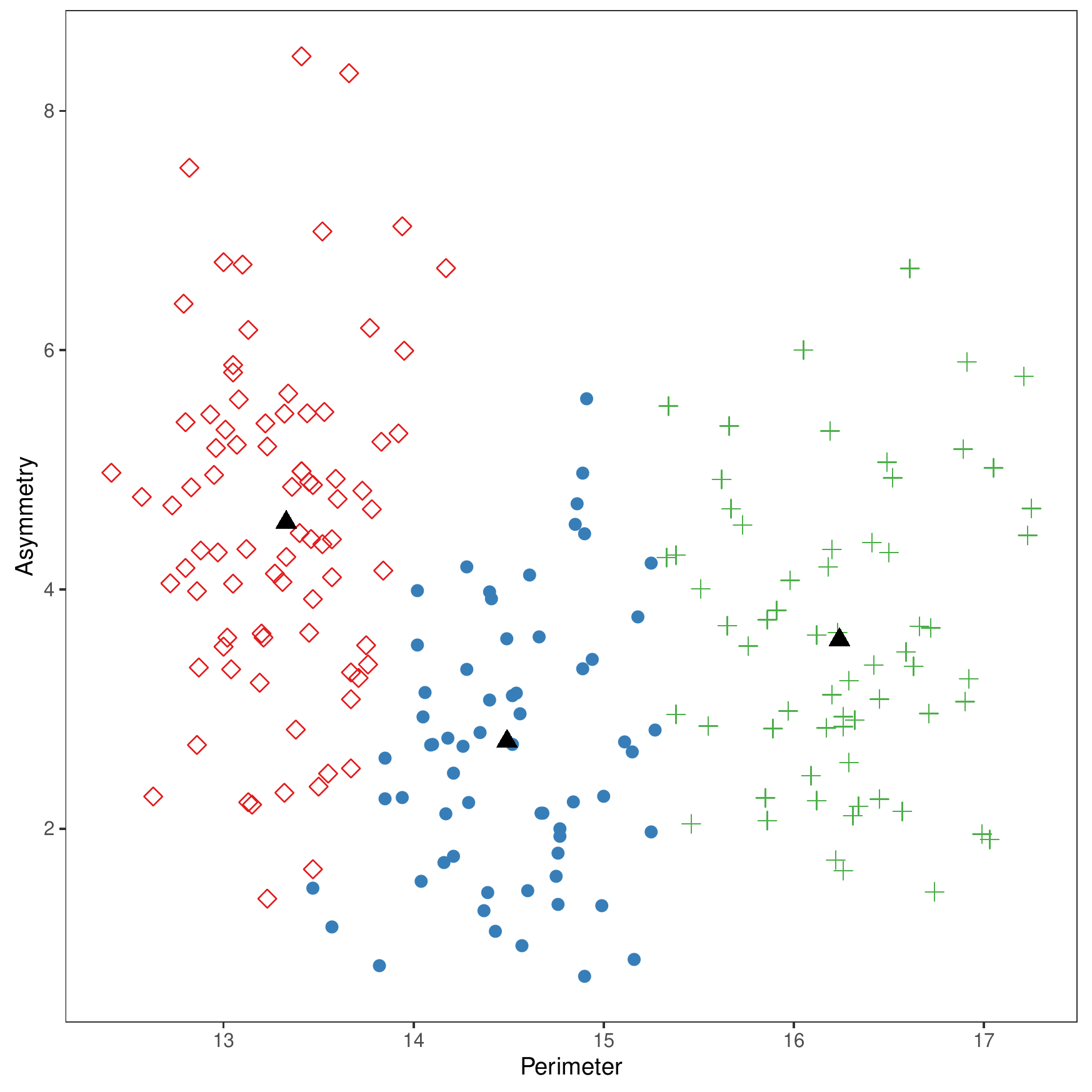}
	  \caption*{(b) Fitted Clustering.}
	 \end{minipage}
  \caption{
  Plots of the true (a) and fitted (b) clusterings of the Seeds dataset, using the GSF-MCP. The lozenges
  		   ({\color{red}$\diamond$}) indicate the Kama, the blue 
		   points ({\color{blue}$\bullet$}) indicate the Rosa
  		   and the green positive symbols ({\color{green} +}) indicate the Canadian seeds. 
		   The black triangles ($\blacktriangle$) in the right-hand plot
  		   show the means of the fitted Gaussian mixture 
		   model by the GSF method using the MCP penalty.}
\label{fig:seedData}	
\end{center}
\end{figure}

The fitted model by the GSF-MCP, 
with an upper bound $K=12$, is
\[
0.37~ {\cal N} \left({{13.33}\choose{4.56}}, \widehat \bSigma \right)
                 + 0.31 ~ {\cal N} \left({{14.50}\choose{2.73}}, \widehat \bSigma \right)   
                 + 0.32 ~ {\cal N} \left({{16.24}\choose{3.58}}, \widehat \bSigma \right)
                 \]
with 
$\widehat \bSigma = \begin{pmatrix} 0.21 & 0.04 \\ 0.04 & 1.66 \end{pmatrix}$. 
The log-likelihood value at this estimate
is -681.85, and the GSF-MCP correctly
classified 88.1\% of the data points. The corresponding 
coefficient plot is reported in Supplement E. 
We also ran the AIC, BIC and ICL, and they all selected the 
three-component model 
\[
0.40~ {\cal N} \left({{13.31}\choose{4.52}}, \widehat \bSigma \right)
                 + 0.31~ {\cal N} \left({{14.55}\choose{2.75}}, \widehat \bSigma \right)   
                 + 0.29~ {\cal N} \left({{16.29}\choose{3.58}}, \widehat \bSigma \right),
                 \]
where $ \bSigma = \begin{pmatrix} 0.20 & 0.04 \\ 0.04 & 1.70 \end{pmatrix}$.
The log-likelihood value at this estimate is given by -655.83, 
and the corresponding classification rate is $87.1\%$.

\begin{center}
\begin{figure}[H]
\centering
\vspace{-1.2in}
  \includegraphics[width=5in,height=5in]{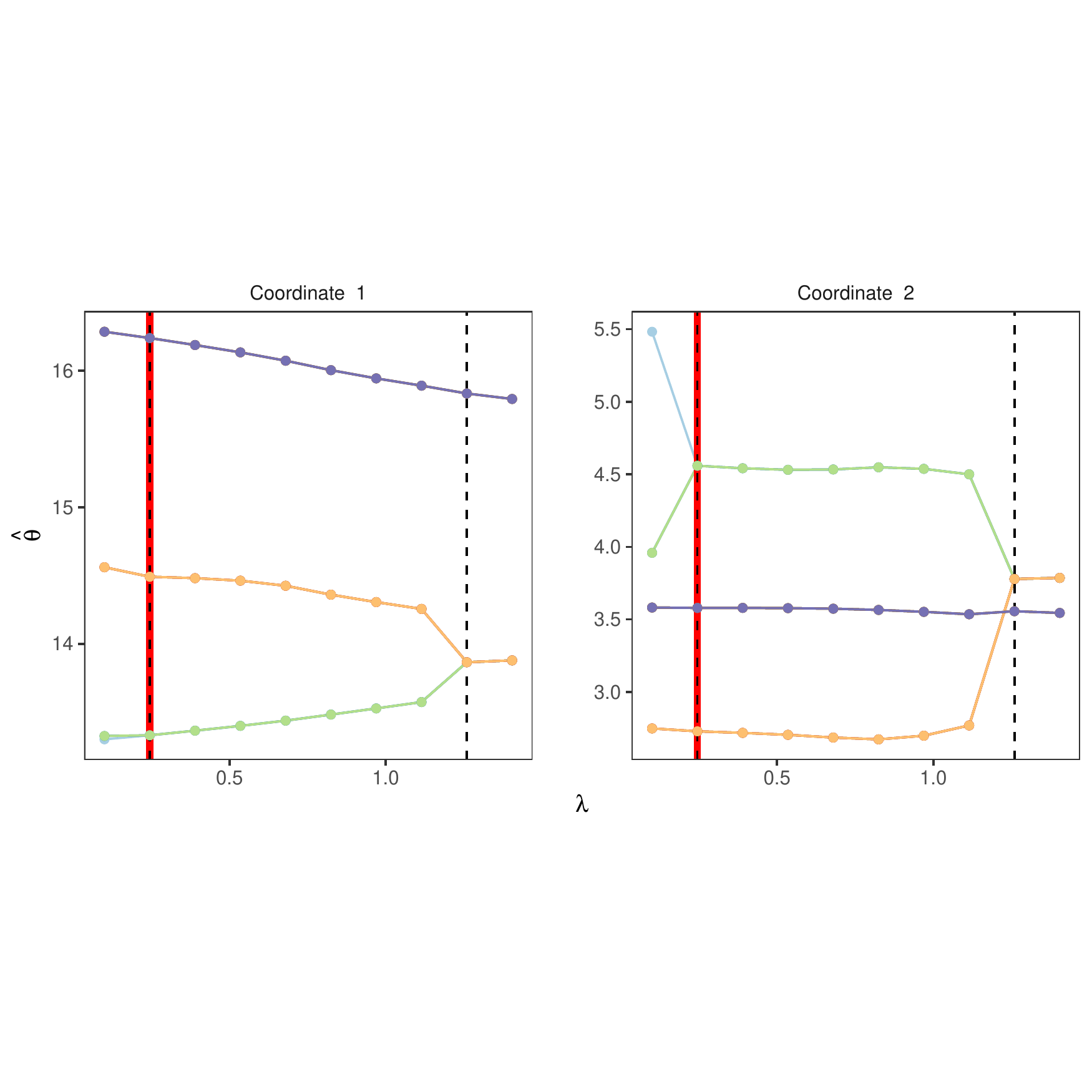}
  \vspace{-1.2in}
  \caption{Coefficient plots for the GSF-MCP on the seeds data. } 
  \label{fig:seedCoef}
\end{figure}
\end{center}
  \vspace{-0.2in}
  
\subsection*{E.7. Regularization plot based on a simulated sample}
Figure \ref{fig:tuning} shows an alternate
regularization plot 
for the  simulated sample used in Figures \ref{fig:coeff-fig1}
and \ref{fig:clusterPlot} of the paper.
\begin{figure}[H]
\centering
\includegraphics[width=.5\linewidth]
{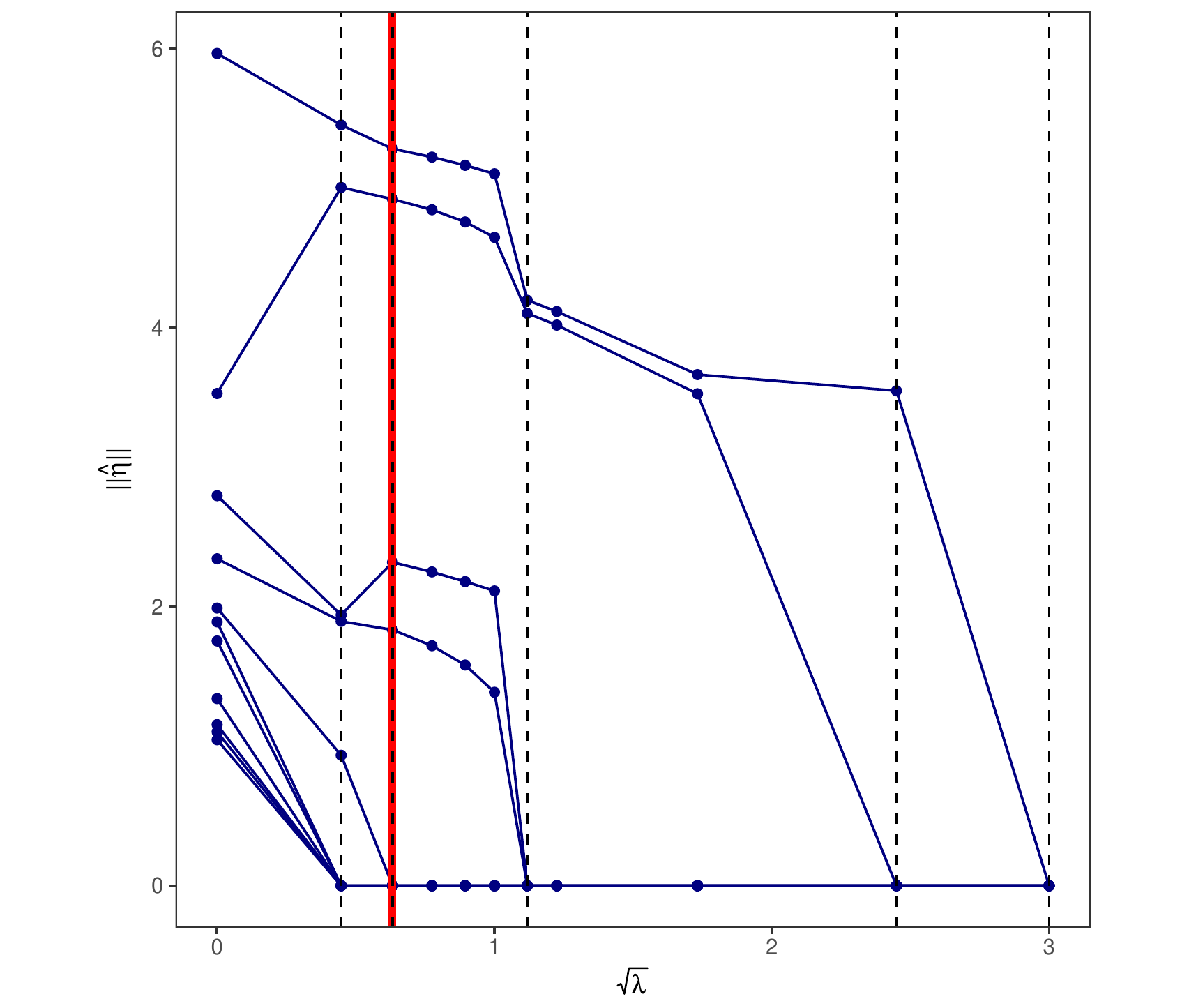}
\caption{
Regularization plot based on the same simulated data as in Figures \ref{fig:coeff-fig1}
and \ref{fig:clusterPlot} of the paper. 
The estimates $\norm{\hbfeta_{j}(\lambda)}, j=1, \dots, K-1=11$, 
are plotted against $\lambda$. The red line shows the value $\lambda^*$ chosen by the BIC. 
 Since there are four non-zero $\norm{\hbfeta_j(\lambda^*)}$, the fitted model has order $\hat K = 5$.}
\label{fig:tuning}	
\end{figure}

\clearpage

\section*{Supplement F: Comparison of the GSF and the Naive GSF}
In this section we provide Models F.1 and F.2 
cited in Figure \ref{fig:comparisonFigure} of the paper, and we elaborate
on the simulation results summarized therein. For both models, 
$\calF$ is chosen to be the family of
two-dimensional location-Gaussian densities, with common but
unknown covariance, that is
$$p_G(\by) = \sum_{j=1}^K \pi_j \frac{1}{\sqrt{(2 \pi)^d |\bSigma}|} 
	\exp \left\{- \frac 1 2 (\by- \bmu_j)^\top \bSigma^{-1} (\by- \bmu_j) \right\},$$
where $\bmu_j \in \bbR^d$, $\pi_j \geq 0$, $\sum_{j=1}^K \pi_j = 1$, 
$j=1, \dots, K$, and we choose $\bSigma = \bI_d$.				
The true mixing measure $G_0 = \sum_{j=1}^K \pi_{0j} \delta_{\mu_{0j}}$
under Models F.1 and F.2 is respectively given by
$$.5 \delta_{(-2, 0)} + .5 \delta_{(0, 1)}, \quad
\frac 1 3 \delta_{(1, 2)} + \frac 1 3  \delta_{(1, 0)} + \frac 1 3 \delta_{(-1, -1)}.$$
We implement the Naive GSF for the SCAD penalty using a modification of the EM algorithm with a
Local Quadratic Approximation (LQA) of the penalty, 
as described by \cite{FAN2001}. In this case, the M-Step of the EM algorithm admits a closed-form
solution. For fairness of comparison, we also reimplement the GSF using this numerical solution. 
All other implementation details are analogous to those listed in Section 
\hyperref[sec:impel-specific]{D.2}.

We run both the GSF and the Naive GSF on 500 samples of size $n=200$ from 
Models F.1 and F.2, for the upper bound $K$ ranging from $5$ to $30$ in increments
of 5. The simulation results are reported in Table \ref{tab:comparisonResults}.
\begin{table}[htbp]
\begin{center}
\resizebox{\textwidth}{!}{
\begin{tabular}{c c | c c c c c c | c c c c c c}
\firsthline
& &  \multicolumn{6}{c}{GSF} &\multicolumn{6}{c}{Naive GSF}\\
\cline{3-6}
\cline{7-14}
Model & $\hat{K}$ & 5 & 10 & 15 & 20 & 25 & 30 & 5 & 10 & 15 & 20 & 25 & 30 \\
\hline
F.1 & 1 & .000 & .000 & .000 & .000 & .000 & .000 & .000 & .000 & .000 & .002 & .000 & .000 \\
  &    {\bf 2} & \textbf{.964} & \textbf{.944} & \textbf{.916} & \textbf{.884} & \textbf{.896} & \textbf{.900} & \textbf{.950} & \textbf{.884} & \textbf{.858} & \textbf{.792} & \textbf{.762} & \textbf{.704} \\
  &    3 & .034 & .054 & .082 & .104 & .088 & .094 & .046 & .106 & .120 & .162 & .204 & .236 \\
  &    4 & .002 & .002 & .002 & .012 & .016 & .006 & .004 & .010 & .020 & .040 & .030 & .050 \\
  &    5 & .000 & .000 & .000 & .000 & .000 & .000 & .000 & .000 & .002 & .004 & .002 & .010 \\
  &    6 & .000 & .000 & .000 & .000 & .000 & .000 & .000 & .000 & .000 & .000 & .002 & .000 \\
\hline
F.2 &  2 & .274 & .276 & .264 & .286 & .274 & .304 & .242 & .274 & .288 & .296 & .318 & .294 \\
  &    {\bf 3} & \textbf{.690} & \textbf{.684} & \textbf{.688} & \textbf{.656} & \textbf{.664} & \textbf{.630} & \textbf{.714} & \textbf{.656} & \textbf{.616} & \textbf{.594} & \textbf{.524} & \textbf{.538} \\
  &    4 & .036 & .040 & .048 & .052 & .060 & .066 & .044 & .066 & .092 & .106 & .144 & .144 \\
  &    5 & .000 & .000 & .000 & .006 & .002 & .000 & .000 & .004 & .002 & .000 & .010 & .020 \\
  &    6 & .000 & .000 & .000 & .000 & .000 & .000 & .000 & .000 & .002 & .004 & .004 & .004 \\
\end{tabular}
}
\caption{Results of the simulation studies, for $K$ ranging from 5 to 30. \label{tab:comparisonResults}}
\end{center}
\end{table} 
 
\clearpage

\bibliography{manuscript_arxiv}

\end{document}